\documentclass[journal,12pt,onecolumn,draftclsnofoot,]{IEEEtran}
\usepackage{blindtext, graphicx}
\usepackage{amsmath,amsfonts,amssymb,bbm}
\usepackage{algorithm,algorithmic}
\usepackage{epstopdf}
\usepackage{epsfig}
\graphicspath{{images/}}
\usepackage{mathtools,dsfont}
\usepackage{psfrag,bbm,graphicx}
\usepackage{algorithm,algorithmic}
\usepackage{comment,balance}
\usepackage{epstopdf}
\usepackage{epsfig}
\usepackage{multicol}
\usepackage{multirow, cite}
\usepackage{amsfonts}
\usepackage{color}
\usepackage{colortbl}
\usepackage{wrapfig}
\usepackage{geometry}
\ifCLASSINFOpdf
  
\else
 
\fi
\DeclareMathAlphabet      {\mathbfit}{OML}{cmm}{b}{it}

\newtheorem{theorem}{Theorem}
\newtheorem{lemma}{Lemma}

\newtheorem{corollary}[theorem]{Corollary}

\newtheorem{remark}{Remark}
\newenvironment{proof}{}{\hfill\rule{2mm}{2mm}}
\newcommand{\remove}[1]{}
\begin{document}
%
% paper title
% can use linebreaks \\ within to get better formatting as desired
\title{Rate-Memory Trade-off for Multi-access Coded  Caching with Uncoded Placement}

\author{Kota Srinivas Reddy and Nikhil Karamchandani \\
	Department of Electrical Engineering, \\
	 Indian Institute of Technology, Bombay \\
	Email: ksreddy@ee.iitb.ac.in, nikhilk@ee.iitb.ac.in
}

% make the title area
\maketitle

\begin{abstract}
	We study a multi-access variant of the popular coded caching framework, which consists of a central server with a catalog of $N$ files, $K$ caches with limited memory $M$, and $K$ users such that each user has access to  $L$ consecutive caches with a cyclic wrap-around and requests one file from the central server's catalog. The server assists in file delivery by transmitting a message of size $R$ over a shared error-free link and the goal is to characterize the optimal rate-memory trade-off. This setup was  studied previously by Hachem et al., where an achievable rate and an information-theoretic lower bound were derived. However, the multiplicative gap between them was shown to scale linearly with the access degree $L$ and thus order-optimality could not be established. 
	
	A series of recent works have used a natural mapping of the coded caching problem to the well-known index coding problem to derive tighter characterizations of the optimal rate-memory trade-off under the additional assumption that the caches store uncoded content. We follow a similar strategy for the multi-access framework and provide new bounds for the optimal rate-memory trade-off $R^*(M)$ over all uncoded placement policies. In particular, we derive a new achievable rate for any $L \ge 1$ and a new lower bound, which works for any uncoded placement policy and $L \ge K/2$. We then establish that the (multiplicative) gap between the new achievable rate and the lower bound is at most $2$ independent of all parameters, thus establishing an order-optimal characterization of $R^*(M)$ for any $L\ge K/2$.  This is a significant improvement over the previously known gap result, albeit under the restriction of uncoded placement policies. Finally, we also characterize $R^*(M)$ exactly for a few special cases. 
%(mentioned in Theorem \ref{thm_ubexact}). This extends our other recent work \cite{reddy2018multiaccess} which studies the exact $R^*(M)$  over all uncoded placement policies for small examples and also for the case of $L=K-1$.

%We study a cache-aided content delivery network consisting of a central server which hosts a  catalog of $N$ files, and $K$ caches each with limited memory $M$ which store content related to the files. There are $K$ users, each of which requests a file from the catalog, and has access to the data stored in $L \ge 1$ neighboring caches (with a cyclic wrap-around). The server transmits a common message to all the users, so that each of them can recover their requested file. This setup was recently studied in \cite{hachem2017codedmulti}, where a coloring-based placement and coded-delivery policy was proposed and the required server transmission size was shown to be \textit{order-optimal} with respect to information-theoretic bounds. We propose an alternate index coding-based placement and delivery scheme for this setup, which performs better than the previously proposed strategy. Furthermore, for multiple special cases including the $( N, K \le 4, L)$ and $(N, K, L = K-1)$-setups, we show that the scheme is \textit{exactly optimal} under the restriction of uncoded placement. This extends other recent work \cite{wan2016optimality,yu2017exact} which studies exact optimality for the single cache-access case ($L=1$) to the multi cache-access case $(L > 1$). 

%
%%%%%%%%%%%COMMENT: Idea behind using `Resource Pooling'%%%%%%%%%
\end{abstract}
{\let\thefootnote\relax\footnote{Preliminary versions of this work appeared in  \cite{reddy2018multiaccess} and  \cite{reddy2019rate}.}}
\IEEEpeerreviewmaketitle

\section{Introduction}\label{sec:introduction}
%In recent years, there has been a rapid increase in the number of smart devices, which leads to an unprecedented growth of (a steep increase in)  the Internet traffic and Video on Demand  (VoD) services account for an increasing data traffic  \cite{Cisco}. One of the most effective ways to meet this rise in demand is caching i.e., instead of serving data from the origin through a bandwidth constrained back-haul link, the data is served from / using a local resource located near  the user. In this paper, we study a cache-aided content delivery network (CCDN), shown in Figure \ref{fig:problemsetupmultiaccess}.
The rapid increase in the usage of smart devices has lead to an unprecedented growth in internet traffic. A recent study \cite{Cisco1} shows that data traffic from Video on Demand (VoD) services will increase exponentially in the forthcoming years. One way to meet this rise in demand is by prefetching and caching some of the data locally. Motivated by this, we study a cache-aided content delivery network (CCDN), as shown in Figure \ref{fig:problemsetupmultiaccess}.

In the last few years, there has been a lot of interest in characterizing the fundamental performance limits of the cache-aided content delivery network (CCDN), see for example \cite{jsac1, jsac2}. %, see for example Figure \ref{fig:problemsetupmultiaccess}. 
A CCDN consists of a central server with a catalog of files, a collection of users, and several caches with limited storage capabilities. The caches  can pre-fetch and store some of the content from the server, such that when users request files from the central server, the caches help the central server in serving the user requests. The main challenges in a CCDN are designing the (\emph{i}) placement policy, which decides what to store in the caches,   (\emph{ii}) delivery policy, which decides how to serve the user requests, with the goal of minimizing the server's transmission rate.

In particular, the seminal work of \cite{maddah2014fundamental, maddah-ali2013} studied the fundamental limits of server's transmission rate of a particular CCDN setup (Ali-Niesen setup) which is as follows: there is a central server with $N$ files of unit size and $K$ users, each one associated with a distinct cache of size $M$ units. Each user requests a file from the central server and based on the request profile as well as the content stored in the caches, the server broadcasts messages on a shared error-free link to the users. The goal is to minimize the server's transmission rate while ensuring that each user can recover its' requested file. \cite{maddah2014fundamental} proposed a (uncoded) placement and (coded) delivery policy for this setup and analyzed the achievable server transmission rate of the scheme as a function of parameters of the system ($N, K, M$). Moreover, \cite{maddah2014fundamental} also showed that the rate achieved by the policy is within a constant multiplicative factor ($12$) of the information-theoretic optimal rate by comparing it to a lower bound, which assumes no restrictions on either the placement or the delivery policy used. Following the initial papers, there have been significant improvements in terms of both the achievable rates \cite{amiri2017fundamental, zhang2018fundamental,wei2017novel, gomez2018fundamental, gomez2018novel} as well as the lower bound arguments \cite{wang2016new, wang2017improved, ghasemi2017improved, yu2019characterizing} which can be used to tighten the gap significantly. In fact, under the natural restriction of an uncoded placement phase where caches are not allowed to store coded content, the rate proposed in \cite{maddah2014fundamental} is shown to be exactly optimal in \cite{yu2018exact, wan2016optimality}. This is done by mapping the CCDN setup to the well-studied index coding problem (ICP) \cite{bar2011index, arbabjolfaei2018fundamentals} (described in Section \ref{sec:prelim}) and using the bounds available in the literature for the ICP.

%A large number of papers have studied similar problem setup with different goals. Especially \cite{wan2016optimality,yu2017exact}  showed that if we restrict to only uncoded placement policies, the rate achieved in \cite{maddah2014fundamental} is exactly optimal (Note that the policy in \cite{maddah2014fundamental} is also uncoded placement policy). The results in \cite{wan2016optimality}, are based on exploiting the connection between the Ali setup  with uncoded placement and well-known index coding problem (ICP).

Several variants of the above setup have been studied in the literature, see for example \cite{maddah2016coding} for an extensive survey. In particular, \cite{hachem2017codedmulti} studied a generalization of the Ali-Niesen setup where instead of each user accessing a single (distinct) cache, every user has access to multiple ($L$) consecutive caches (with a cyclic wrap-around), as shown in Figure \ref{fig:problemsetupmultiaccess}. This was motivated by the upcoming heterogeneous cellular architecture which will consist of  dense deployment of wireless access points (APs) with small coverage and relatively large data rates, in addition to sparse cellular base-stations (BSs) with large coverage and smaller data rates. Placing caches at local APs can help reduce the BS transmission rate, with each user capable of accessing the content stored at multiple nearby APs in addition to receiving the BS broadcast. For this multi-access CCDN, \cite{hachem2017codedmulti} proposed a coloring based placement and delivery policy and analyzed its' achievable rate $R_{\text{color}}(M)$. By comparing this rate to the information-theoretic optimal rate $R^*_{\text{inf}}(M)$ which puts no restrictions on either the placement or the delivery policy used, \cite[Theorem~4]{hachem2017codedmulti} showed that ${R_{\text{color}}(M)}/{R^*_{\text{inf}}(M)}\leq cL$, where $c$ is some constant, independent of all system parameters. Thus, the gap between the achievable rate and the information-theoretic lower bound increases linearly with $L$, the number of caches each user has access to, and the obvious challenge is to improve the achievable rate and/or lower bound to establish the exact optimal rate-memory trade-off or at least up to a constant factor for any $L$. 
%this tquestions of interest are  (\emph{i})  what is the exact memory -rate trade-off (at least for uncoded placement policies)?,  (\emph{ii}) are there any policy which achieves order optimality?.

% was also derived and showed that the rate achieved by the coloring based policy is  linear in $L$ (independent of other system parameters) with respect to information theoretic lower bound. More precisely, \cite[Theorem~4]{hachem2014multi} showed that $\frac{R_{\text{color}}(M)}{R^*_{\text{inf}}(M)}\leq cL$, where $c$ is some constant, independent of system parameters, $R_{\text{color}}(M)$ is the coloring based achievable rate, and $R^*_{\text{inf}}(M)$ is the information theoretical lower bound with no restrictions on either the placement or the delivery policy used. Now, the next questions of interest for this multi ($L$)- access setup are (\emph{i})  what is the exact memory -rate trade-off (at least for uncoded placement policies)?,  (\emph{ii}) are there any policy which achieves order optimality?.
\begin{figure}[t]
	\begin{center}
		\includegraphics[scale=0.33]{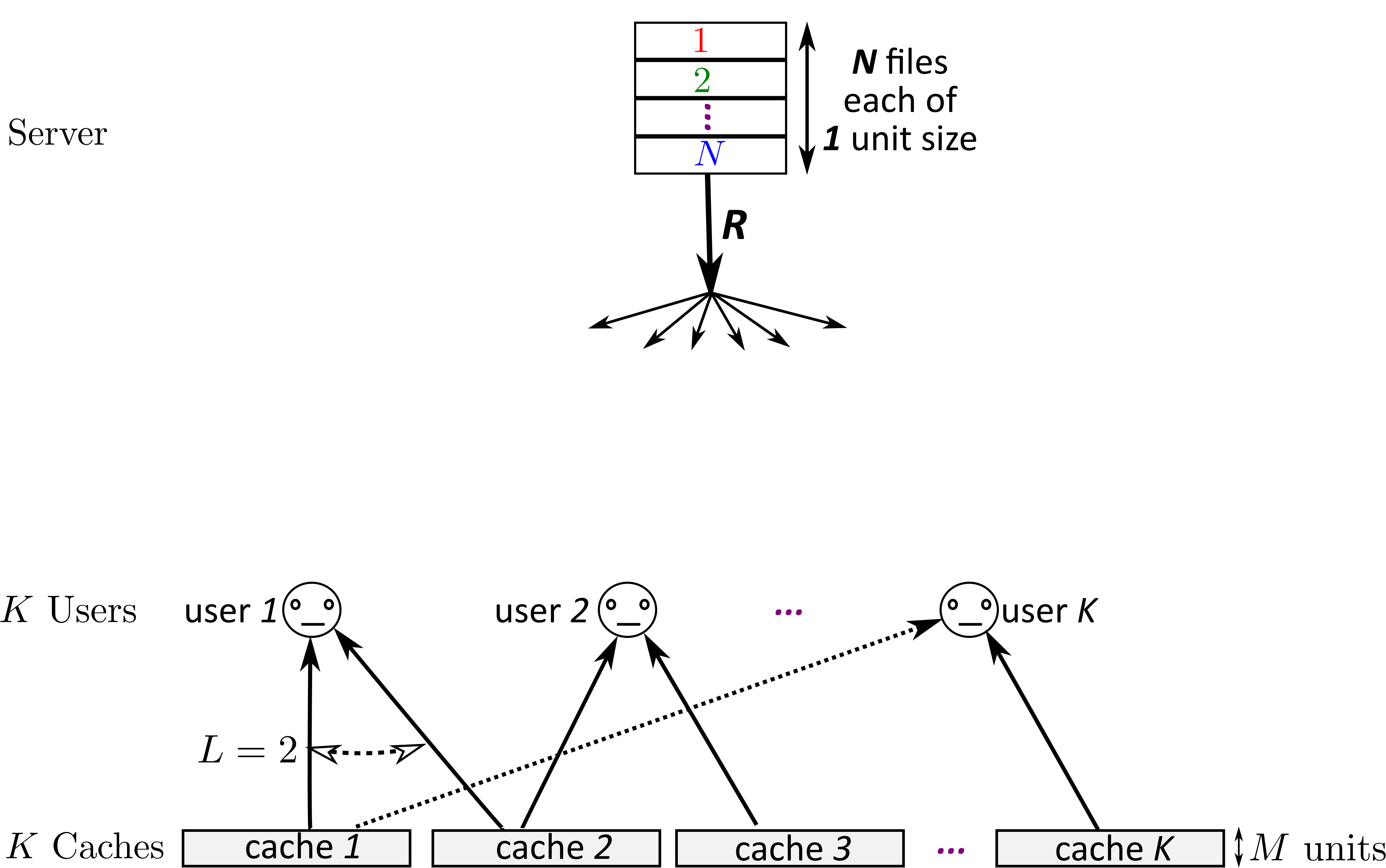}
		%\caption{\sl A multi-access CCDN consisting of $N$ files, $K$ caches, each of size $M$ units, and $K$ users, each user is connected to $L=2$ caches.   \label{fig:problemsetupmultiaccess}}
		\caption{A multi-access CCDN consisting of $N$ files, $K$ caches, each of size $M$ units, and $K$ users, each user is connected to $L=2$ caches.   \label{fig:problemsetupmultiaccess}}
	\end{center}
	\vspace{-0.5in}
\end{figure}

%Following the recent work \cite{yu2018exact, wan2016optimality} on the Ali-Niesen setup ($L = 1$) which studies the problem under the natural restriction of an uncoded placement phase, \cite{reddy2018multiaccess} considered a few specific examples for the multi-access setup $(L > 1)$ such as $L = K - 1$ and small CCDN with $K \le 4$, and demonstrated an improved scheme as well as a matching lower bound under the restriction of uncoded placement. 
In this paper, we study the same multi-access setup and make the following technical contributions: 
\begin{itemize}
	\item  derive a new achievable rate for the general multi-access CCDN with $L > 1$, based on a scheme using uncoded placement; can be order-wise better than the best previously known rate \cite{hachem2017codedmulti},
	\item  derive a general lower bound on the optimal rate for any multi-access CCDN with $L \geq K/2$, under the restriction of uncoded placement,
	\item  establish order-optimal (up to a multiplicative factor of $2$) uncoded placement rate-memory trade-off for any multi-access CCDN with $L \geq K/2$, 
	\item  establish exact optimal uncoded placement rate-memory trade-off for a few special cases, for example $L=K-1$; $L=K-2$; $L=K-3$ with $K$ even. 
\end{itemize} 
As mentioned before, \cite{hachem2017codedmulti} derived an achievable rate and an information-theoretic lower bound which differ by a multiplicative gap scaling linearly with $L$. Using an improved achievable rate as well as a better lower bound under the restriction of uncoded placement, we are able to establish the uncoded placement rate-memory trade-off for any multi-access CCDN with $L \geq K/2$ up to a constant factor, independent of $L$.
To illustrate  the gains of our policy over the coloring based policy in \cite{hachem2017codedmulti}, we plot 
\begin{enumerate}
	\item the rate-memory trade-off for the ($N=4, K=4, L= 2$)-CCDN in Figure \ref{fig:442i}.  For this system, our policy performs better than the coloring based policy in all the memory regimes. Moreover, as we show in Appendix I, our policy in fact achieves the optimal rate-memory trade-off.
\item the transmission rate  at the memory point $M = N / K$ as a function of the number of users/caches $K$  for an ($N \geq K, K, L =K/2$)-CCDN in Figure \ref{fig:lgk2i}. As the plot shows, the additive gap between the performance of the two schemes increases linearly with $K$. %The plot shows that the additive gain is increasing linearly with $K$.
\end{enumerate}
\begin{figure}[h]
	\centering

	%\hspace{0.1in}
	\begin{minipage}{0.48\textwidth}
		\centering
		\includegraphics[width = 1.0\textwidth]{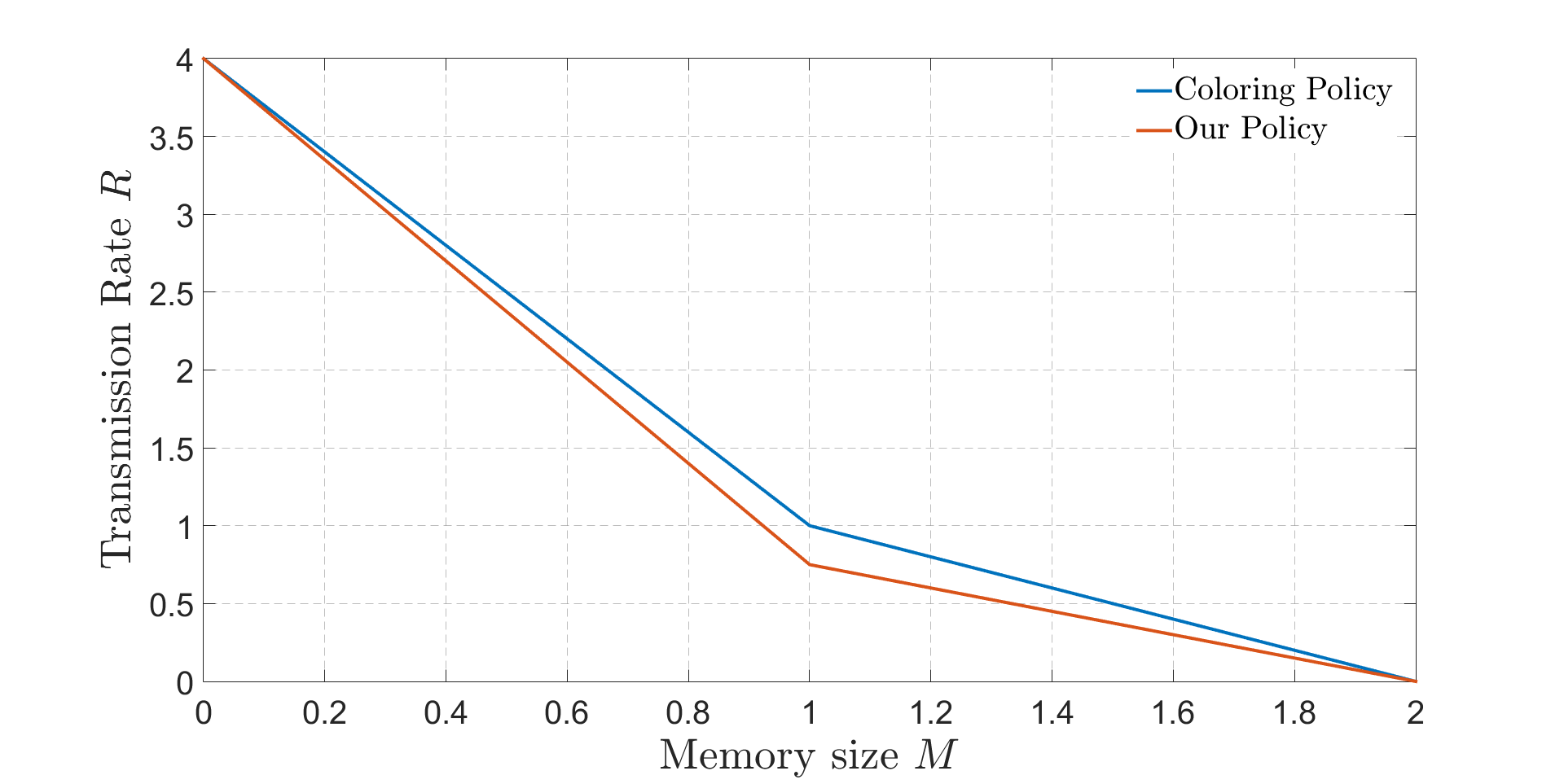}\\
		\caption{Plot of the transmission rate $R$ as a function of the memory size of the each cache $M$ for the ($N=4, K=4, L= 2$)-CCDN. }\label{fig:442i} 
		\vspace{-0.25in} 
		%	\end{figure}
		%\begin{figure}[h]
	\end{minipage}
	\hspace{0.1in}
	\begin{minipage}{0.48\textwidth}
		\centering
		\includegraphics[width = 1.0\textwidth]{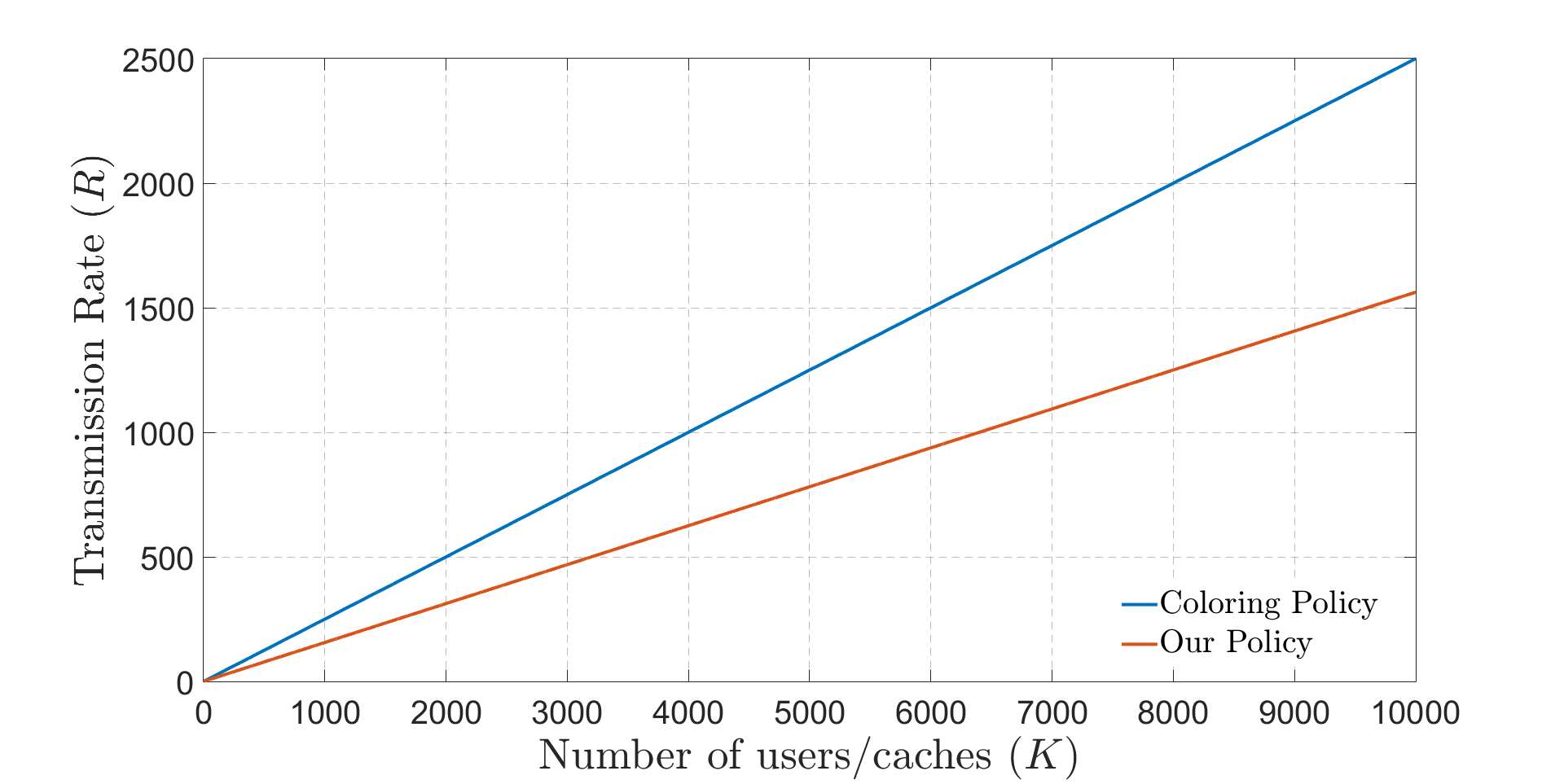}\\
		\caption{Plot of the transmission rate $R$ at memory point $M=N/K$ as a function of the number of users/caches $K$ for an ($N \geq K, K, L =K/2$)-CCDN.}\label{fig:lgk2i} 
		\vspace{-0.25in} 
	\end{minipage}
	%\hspace{-0.4in}	
	
\end{figure}

Our results are based on mapping our multi-access CCDN setup to appropriate index coding problems \cite{bar2011index} and finding bounds on their solutions, which might be of interest in their own right. The index coding problems that result from our mapping are similar in spirit to those studied in \cite{maleki2014index,vaddi2016optimal,vaddi2017capacity}. Similar index coding-based approaches have been used recently to study other variants of the original coded caching problem and for designing better achievability schemes as well as prove converses under the restriction of uncoded placement \cite{wan2017novel, bahrami2017towards, parrinello2018coded,  wan2017novelindex, jin2018uncoded, ozfatura2018uncoded, parrinello2018fundamental, wei2017coded, ji2014caching}. 

The rest of the paper is organized as follows. Section \ref{sec:setting} briefly describes the problem setting, while Sections~\ref{sec:not} and \ref{sec:prelim} describe some useful notations and preliminaries. Section \ref{sec:policy} and Section \ref{sec:results_multiacc} describe our policy and the main results. The summary of our work and future goals are given in Section \ref{sec:discussions} and the proofs of our main results are given in Section \ref{sec:proofs}. %For detailed proofs, please refer to the extended version of this manuscript \cite{reddy2018isitmultiaccess}. 

\section{Problem Setting} \label{sec:setting}

We consider a cache-aided content delivery network (CCDN) which consists of a central server, $K$ users, and $K$ caches as shown in Figure \ref{fig:problemsetupmultiaccess}. We assume,
\begin{itemize}
	\item  the central server contains $N$ $(\geq K)$ files $\mathcal{F}_1,\mathcal{F}_2,...,\mathcal{F}_N$, each of size $1$ unit (=$F$ bits\footnote{We will assume the file size $F$ to be large}),
	\item  each user has access to $L$ consecutive caches with a cyclic wrap-around\footnote{For symmetry, we assume that Caches 1 and $K$ are adjacent.} as shown in Figure \ref{fig:problemsetupmultiaccess},
	\item  cache sizes are uniform and are $M$ units each,
	\item  each user requests one file which has to be served by a central server's message, and the content in the $L$ caches it has access to, and
	 \item  the communication channel between the central server and the users is an error-free shared (broadcast) channel.
	 %\item  [--] our system is delay intolerant.
\end{itemize}
We will refer to the above system as the ($N, K, L$)$-$CCDN. 

The system runs in two phases: a placement phase and a delivery phase.
\subsubsection{Placement Phase} In the placement phase, we fill the caches  with the content  related to the $N$ files. Like \cite{wan2016optimality,yu2018exact}, we restrict to an uncoded placement phase.  We are allowed to split the files into parts and store the file parts, but coding across the file parts is not allowed while storing in the caches.  The placement phase occurs before users reveal their requests and hence is independent of user requests. 

After the placement phase, each user (User j) requests one file (File $d_j$) from the central server, chosen arbitrarily from amongst the $N$ files.  We call $\mathbf{d}=$ ($d_1,d_2,...,d_K$) as the request profile.

%We can think of placement phase as an operation which occurs during the network off peak time to reduce the network peak time traffic. i.e., transferring load from  peak time to off peak time. Hence, we don't consider the cost occurred (or bandwidth required) during the placement phase. Memory is the only constraint during the placement phase. 

\subsubsection{Delivery Phase} In the delivery phase,  depending on the request profile ($\mathbf{d}$) and content stored in the caches, the central server broadcasts a message of size $R$ units such that each user can recover their requested file using the server transmission and the content in the $L$ caches it has  access to. {We refer to the message size $R$ as the server transmission rate.}

A rate-memory pair ($R,M$) is said to be achievable for request profile $\mathbf{d}=(d_1,d_2,...,d_k)$   if there exists a placement and delivery scheme {with server transmission rate $R$} for cache size $M$,  such that every user (User $k$) can recover its requested file (File ${d_k}$). A rate-memory pair ($R,M$) is said to be achievable  if this pair is achievable for any possible arbitrary request profile.  For a given cache size $M$, we define optimal rate-memory trade-off  $R^*(M)$ as the smallest rate $R$ for which the rate-memory pair  ($R,M$) is achievable. Our goal is to characterize $R^*(M)$ under the restriction of uncoded placement and come up with a placement and delivery policy that achieves  $R^*(M)$.

%{\color{blue}A rate-memory pair ($R,M$) is said to be achievable if there exists a placement and delivery scheme {with server transmission rate $R$} for cache size $M$,  such that every user can recover its requested file. For given cache size $M$, we define $R^*(M)$ {(optimal rate-memory trade-off)} as the smallest rate $R$ for which the rate-memory pair  ($R,M$) is achievable for any arbitrary request profile. Our goal is to characterize $R^*(M)$ under the restriction of uncoded placement and come up with a placement and delivery policy which achieves  $R^*(M)$.}

As mentioned in Section \ref{sec:introduction}, \cite{hachem2017codedmulti} studied this setup with multi-cache access ($L > 1$) and proposed a coloring-based achievability scheme which builds on the coded delivery ideas presented in \cite{maddah2014fundamental, maddah-ali2013} for single cache access ($L=1$). For the general setup\remove{\footnote{The coloring based scheme also assumes that $K$ is an integer multiple of $L$}}, the server transmission rate $R_{color}(M)$ for this scheme is given by
$$R_{color}(M)=\frac{K-\frac{KLM}{N}}{1+\frac{KM}{N}}, \ M\in\bigg\{0,\frac{N}{K},\frac{2N}{K},...\frac{N}{L}\bigg\}.$$ For general $0 \le M \le N / L$, the lower convex envelope of these points is achievable via memory-sharing.  Incidentally, the proposed scheme used uncoded placement and coding is used only in the server broadcast message. 

In this paper, we derive a new achievable rate for the general\remove{\footnote{$K$ need not be integer multiples of $L$}} multi-access CCDN with $L > 1$, based on a scheme using uncoded placement, which can be order-wise better than the best previously known rate  $R_{color}(M)$ \cite{hachem2017codedmulti}. Our new achievable rate is exactly optimal for a few cases and order wise optimal for any $L\geq {K}/{2}$.

\section{Notations}\label{sec:not}

\begin{itemize}
	\item $[n]=\{1,2,3,...,n\}$.
	\item \begin{align*}
	[i:j]=\begin{cases}
	\{i,i+1,...,j\} & \text{ if } i\leq j, \\
	\{i,i+1,...,K,1,2,...,j\} & \text{ if } i>j.
	\end{cases}
	\end{align*}
	%\item $P(\mathcal{S})$ is power set of set $\mathcal{S}$
	%\item ${F}_{i,[j,k,...]}-$ parts of File $i$ stored in caches $j,k,...$
	%\item ${F}_{i,\mathcal{S}}-$ parts of File $i$ available to usrers in the set $\mathcal{S}$ 
	\item \begin{align*}
	<i>=\begin{cases}
	<i+K> & \text{ if } i\leq 0, \\
	i & \text{ if } 0<i\leq K,\\
	<i-K> & \text{ if } i> K.
	\end{cases}
	\end{align*}
	\item $|.|$ denotes the cardinality of a set or size of a  file/subfile. 
	\item {$\mathcal{F}_{i,\mathcal{S}}$} denotes parts of File $i$ exclusively available to users with index in set $\mathcal{S} \subseteq [K]$.
	\item $P(\mathcal{S})$ denotes the power set of set $\mathcal{S}$.	
%	\item $(i,j)-$ $i^{th}$ row and $j^{th}$ column in a table
%	\item $(i,\bar{j})-$ entire $i^{th}$ row except $(i,j)$
%	\item interfering node: In an ICP, Node $A$ is interfering with Node $B$, if Node $A$'s requested file is not the side information of Node $B$
%	\item Proper coloring: a node color should be different from all its interfering nodes
\end{itemize}
%\begin{align*}
%R_o(M)= \begin{cases}
%K-\Big[K-\frac{(K-L)^2}{K}\Big]\frac{MK}{N}, & \text{if } 0\leq M\leq \frac{N}{K},  \\
%\frac{(K-L)^2}{K}(2-\frac{MK}{N}), & \text{if } \frac{N}{K}\leq M\leq \frac{2N}{K},\\
%0 & \text{if } M\geq \frac{2N}{K}.
%\end{cases}
%\end{align*}
%
%{\small
%
%%
%%	\begin{table}[h]
%%	\centering
%%	\resizebox{\linewidth}{!}{%
%%		\begin{tabular}{ c }
%		\begin{align*}
%				{R}_e(M)= \begin{cases}
%				K-\Big[K-\frac{(K-L)(K-L+1)}{2K}\Big]\frac{MK}{N}, & \text{if } 0\leq M\leq \frac{N}{K},  \\
%				\frac{(K-L)(K-L+1)}{2K}(2-\frac{MK}{N}), & \text{if } \frac{N}{K}\leq M\leq \frac{2N}{K},\\
%				0 & \text{if } M\geq 2\frac{N}{K}.
%				\end{cases}
%		\end{align*}}
%	
%	{\small
%	\begin{align*}
%	R_a(M) = \begin{cases}
%	K-\Big[K-\frac{6}{K}-\frac{2}{K(K-1)}\Big]\frac{MK}{N}, & \text{if } 0\leq M\leq \frac{N}{K},  \\
%	\Big[\frac{6}{K}+\frac{2}{K(K-1)}\Big]\big(2-\frac{MK}{N}\big), & \text{if } \frac{N}{K}\leq M\leq \frac{2N}{K},\\
%	0 & \text{if } M\geq \frac{2N}{K}.
%	\end{cases}
%	\end{align*}}
%		
%		
%%	\end{tabular}}
%%	\vspace{0.5cm}
%%	\caption{Interfering sub nodes for sub node (3,K)} \label{interfere30}
%%\end{table}
%%\resizebox{\linewidth}{
%%	\begin{tabular}{c}
%%		\begin{align*}
%%		{R}_e(M)= \begin{cases}
%%		K-\Big[K-\frac{(K-L)(K-L+1)}{2K}\Big]\frac{MK}{N}, & \text{if } 0\leq M\leq \frac{N}{K},  \\
%%		\frac{(K-L)(K-L+1)}{2K}(2-\frac{MK}{N}), & \text{if } \frac{N}{K}\leq M\leq \frac{2N}{K},\\
%%		0 & \text{if } M\geq 2\frac{N}{K}.
%%		\end{cases}
%%		\end{align*}
%%	\end{tabular}
%%
%%}
\section{Preliminaries}\label{sec:prelim}
%We use a few graph parameters to characterize our system's optimal rate-memory trade-off and their description is given below.\\
%\textit{Proper Coloring Scheme:} A coloring scheme of a (directed) graph is defined as assigning colors to the nodes and is said to be proper if none of the edges contain same color to its both nodes (or no two adjacent nodes assigned same color\cite{localchromatic2013}). \\
%\textit{Local Chromatic Number $\mathcal{X}_l(\mathcal{G})$:} We define the closed out neighborhood of a Node $v$ as Node $v$ itself and all its out neighbors (the nodes which contain a directed edge from Node $v$ i.e., Node $i$ is an out neighbor of Node $v$ iff ($v,i$) is a directed edge in $\mathcal{G}$). The local chromatic number ($\mathcal{X}_l(\mathcal{G})$) is defined as the maximum number of different colors that appear in the closed out neighborhood of any vertex, minimized over all proper colorings. 
%
%Let the set $\mathcal{N}^+(i)$ denote the closed out-neighborhood (including $i$) of a given vertex  $i$ in a directed graph $ \mathcal{G}$, i.e., $j\in \mathcal{N}^+(i)$ iff ($i,j$) is a directed edge or $j=i$. Let $c:V\rightarrow[k]$ be a proper coloring scheme with $k$ (any positive integer) colors and $|c(\mathcal{N}^+(i))|$ denote the number of colors in the closed out neighborhood of the Node $i$. Then,
%$$\mathcal{X}_l(\mathcal{G})=\min_c\max_{i\in V}|c(\mathcal{N}^+(i))|.$$

We map our setup to the well-studied index coding problem (ICP) \cite{bar2011index} and use some of the ideas developed for this problem to characterize the optimal rate-memory trade-off for our setup. Similar to our setup, ICP has a server with a catalog of say $n$ files. There are $n$ nodes, such that Node $i$ requests File $i$ and has access to a  subset of the remaining files $J_i$ $\subseteq[n]/{i}$. $J_i$ is called the side information of Node $i$.  Depending on  the side information profile, the server broadcasts a message so that each node can recover its requested file using the broadcast message and the side information available. The goal is to characterize the minimum broadcast rate for any given instance of the ICP.    

Even though the ICP and our problem setting look similar \big(\{server, node requests, side information\} $\approx$ \{central server, user requests, accessible caches content\}\big), the differences between the ICP and our problem setup are
\begin{enumerate}
	\item in an ICP problem, the side information is already given whereas, in our setup, we have the choice of what to store in caches.
	\item in an ICP problem, the node requests are fixed whereas, in our setup, the user requests are arbitrary. 
\end{enumerate}
Once the cache contents are fixed and user requests are revealed, then the problem of minimizing the central server's broadcast rate in our setup is equivalent to that for a corresponding ICP.

\remove{We can represent an $n$ node ICP by an equivalent $n$ {\color{red}vertex} directed side information graph $\mathcal{G}$,   where each vertex corresponds to a unique node, and there exists an edge from Vertex $i$  to Vertex $j$ if Node $i$ has File $j$ (Node $j$'s requested file). Lemma \ref{lemma:lowerbound} gives a lower bound on the broadcast rate of an ICP using the side information graph ($\mathcal{G}$) and it follows from \cite[Corollary 1]{arbabjolfaei2013capacity}.

%\begin{lemma}
%	Given a $n$ user ICP with side side information graph $\mathcal{G}$, for any subset $\mathcal{J}\subseteq[n]$ such that the subgraph of $\mathcal{G}$ induced by the nodes in $\mathcal{J}$ does not contain a directed cycle, then the {\color{red}minimum} broadcast rate $S(\mathcal{G})$ of the ICP is $$S(\mathcal{G})\geq\sum_{j\in\mathcal{J}} M_j \ , $$  where $M_j$ is the size of file requested by User $j$.
%\end{lemma}
\begin{lemma}\label{lemma:lowerbound}
 Consider an $n$ node ICP with side information graph $\mathcal{G}$. Let $M_i$ be the file requested by Node $i$ and $S(\mathcal{G})$ be the minimum broadcast rate. Then, for any subset $\mathcal{J}\subseteq[n]$ such that the subgraph of $\mathcal{G}$ induced by the vertices in $\mathcal{J}$ does not contain a directed cycle, we have $$S(\mathcal{G})\geq\sum_{j\in\mathcal{J}}|M_j| \ . $$ % where $|M_j|$ indicates the size of file $ M_j$.
\end{lemma}}

%We upper bound the minimum broad cast rate of an ICP by the local chromatic number of the (directed) complement of the side information graph corresponding to the ICP. We call the (directed) complement of the side information graph as interference graph of the ICP.
%
%\textit{Interference Graph ($\overline{\mathcal{G}}$):} Interference graph ($\overline{\mathcal{G}}$) is a directed graph with $n$ nodes, where each node corresponds to a unique user, and there exists an edge from Node $i$  to Node $j$ if User $i$ does not have File $j$.

\remove{\color{blue}In an ICP, Node $j$ is said to be \textit{interfering with} Node $i$,  if Node $i$ does not have File $j$  as side information. \remove{In the side information graph $\mathcal{G}$, for any vertex $i$, all the other vertices except the out-neighbors are interfering vertices {\color{red}(vertices corresponding to the interfering nodes of Node $i$). }} The closed anti-outneighborhood of a Node  $i$ ($\mathcal{N}^+(i)$) is Node $i$ itself and all its interfering nodes. A coloring scheme for an ICP assigns a color to each node and is said to be \textit{proper} if no node shares its color with any of its interfering nodes. The local chromatic number of an ICP ($\mathcal{X}_l$)  is defined as the maximum number of different colors that appear in the closed anti-outneighborhood of any node, minimized over all proper colorings.
 %A coloring scheme for an ICP is defined as assigning colors to the {\color{red}nodes} in the ICP and is said to be proper if none of the {\color{red}nodes} contain same color with its interfering {\color{red}nodes}. \\
%\textit{Local Chromatic Number $\mathcal{X}_l(\mathcal{G})$:} We define 

Let $c:[n]\rightarrow[k]$ be a proper coloring scheme with $k$ (any positive integer) colors and $|c(\mathcal{N}^+(i))|$ denote the number of colors in the closed anti-outneighborhood of the Node $i$. Then,
$$\mathcal{X}_l=\min_c\max_{i\in V}|c(\mathcal{N}^+(i))|.$$}

In an ICP, Node $j$ is said to be {interfering with} Node $i$,  if Node $i$ does not have File $j$  as side information. \remove{In the side information graph $\mathcal{G}$, for any vertex $i$, all the other vertices except the out-neighbors are interfering vertices {\color{red}(vertices corresponding to the interfering nodes of Node $i$). }} Let $\mathcal{N}^+(i)$ denote the closed anti-outneighborhood of a Node  $i$ which is defined as the set containing Node $i$ itself and all its interfering nodes. A coloring scheme for an ICP assigns a color to each node and is said to be \textit{proper} if no node shares its color with any of its interfering nodes.  
	Let $c:[n]\rightarrow[k]$ be a proper coloring scheme with $k$ (any positive integer) colors and $|c(\mathcal{N}^+(i))|$ denote the number of colors in the closed anti-outneighborhood of the Node $i$. Then, the local chromatic number of an ICP ($\mathcal{X}_l$)  is defined as
	$$\mathcal{X}_l=\min_c\max_{i\in V}|c(\mathcal{N}^+(i))|.$$
In words, the local chromatic number of an ICP is defined as the maximum number of different colors that appear in the closed anti-outneighborhood of any node, minimized over all proper colorings.
%

%********************************************************************

We use the following result  to upper bound the broadcast rate of an ICP, which follows from \cite[Theorem~1]{localchromatic2013}. 
\remove{\begin{lemma}\label{lemma:upperbound}
	Given an $n$ node ICP with side information graph ${\mathcal{G}}$, the minimum broadcast rate $S({\mathcal{G}})$ of the ICP satisfies $$S({\mathcal{G}})\leq\mathcal{X}_l, $$  where $\mathcal{X}_l$ is the local chromatic number of the ICP.
\end{lemma}}

\begin{lemma}\label{lemma:upperbound}
	Given an $n$ node ICP, the minimum broadcast rate of the ICP is upper bounded by its local chromatic number $\mathcal{X}_l$.
\end{lemma}

We can represent an $n$ node ICP by an equivalent $n$ {vertex} directed side information graph $\mathcal{G}$,   where each vertex corresponds to a unique node, and there exists an edge from Vertex $i$  to Vertex $j$ if Node $i$ has File $j$ (Node $j$'s requested file). 

There are several available lower bounds for the ICP. In particular, we use Lemma \ref{lemma:lowerbound} to derive a lower bound on the server transmission rate for our multi-access CCDN and it follows from \cite[Corollary 1]{arbabjolfaei2013capacity}.

%\end{lemma}
\begin{lemma}\label{lemma:lowerbound}
	Consider an $n$ node ICP with side information graph $\mathcal{G}$. Let $M_i$ be the file requested by Node $i$ and $S(\mathcal{G})$ be the minimum broadcast rate. Then, for any subset $\mathcal{J}\subseteq[n]$ such that the subgraph of $\mathcal{G}$ induced by the vertices in $\mathcal{J}$ does not contain a directed cycle, we have $$S(\mathcal{G})\geq\sum_{j\in\mathcal{J}}|M_j| \ . $$ % where $|M_j|$ indicates the size of file $ M_j$.
\end{lemma}

\section{Our Policy} \label{sec:policy}

In this section, we 
describe our policy with an example. 
   Our policy can have a lower server transmission rate than the coloring-based policy proposed in \cite{hachem2017codedmulti}. Our placement and delivery policies, both are different to the coloring-based policy. Our delivery policy is based on a solution to an appropriately defined ICP. %the  ICP described in Section \ref{sec:prelim}. 

Consider a general $(N, K, L)$-CCDN setup. We first describe our achievable scheme for corner points, i.e., for memory points $M=iN/K$, where $i\in \Big[\big\lceil\frac{K}{L}\big\rceil\Big]\cup\{0\}$. 
%\begin{itemize}
\subsection{$i=0$}
%\item{$i=0$}
%\begin{itemize}
%	\item  $i=0$	
	 If $i=0$, then $M=0$, i.e.,  no file parts are stored in the cache. The worst-case request pattern is all users requesting different files. Hence, $R(0)=K\text{ units.}$ 
\subsection{$i\in \big[\big\lfloor\frac{K}{L}\big\rfloor\big]$}
	 %\item $i\in \big[\big\lfloor\frac{K}{L}\big\rfloor\big]$ 
	 
	 First, we define the term $\hat{\mathcal{S}}$ as follows:
	 \begin{align}\label{shat}
	 \hat{\mathcal{S}}=\{s=\{a_1,a_2,...,a_i\}\subseteq[K]:|a_j-a_l|\geq L,  |K-|a_j-a_l||\geq L \text{ }\forall 1\leq j\neq l\leq i\}.
	 \end{align}
	 In words, $\hat{\mathcal{S}}$ is the collection of subsets $s$ of $[K]$ which satisfy (\emph{i}) $|s|=i$, and 
	 (\emph{ii}) if { $i>1$}, every two distinct elements $a_j,a_l$ of $s$ satisfy $|a_j-a_l|\geq L$ and $|K-|a_j-a_l||\geq L$.
%	 Let $\hat{\mathcal{S}}$ be the collection of subsets $s$ of $[K]$ which satisfy (\emph{i}) $|s|=i$, and 
%	 (\emph{ii}) if { $i>1$}, every two elements $(j,l)$ of $s$ satisfy $|j-l|\geq L$ and $|K-|j-l||\geq L$. {\color{blue}Mathematically $\hat{\mathcal{S}}$ is defined as}
	 %\end{itemize}
	 
	 \subsubsection{Placement policy} Divide each file into $|\hat{\mathcal{S}}|={K-iL+i-1 \choose i-1}\frac{K}{i}$ parts, with one subfile corresponding to each subset $s \in \hat{S}$. Store the subfile corresponding to set $s$ in all the $i$ caches whose index belongs to $s$. 
	 
	  The above policy creates overlaps in the cache contents to increase coded-multicasting opportunities while ensuring that there are no redundant copies. For example, a user should not have two copies of the same subfile amongst its accessible caches. 
	  
	 \subsubsection{Delivery policy} After users have revealed their requests, form an instance of the ICP with the  file parts which are unavailable locally. The server transmits messages according to the solution of the ICP. 
	 \subsection{$i= \big\lceil\frac{K}{L}\big\rceil$}
	% \item $i= \big\lceil\frac{K}{L}\big\rceil$
	 
	 We first create a list (of size $\big\lceil\frac{K}{L}\big\rceil\times N$ elements) by repeating the sequence $\{1,2,...,N\}$ $MK/N=\big\lceil\frac{K}{L}\big\rceil$ times, i.e., $\{1,2,...N-1,N,1,2,...N-1,N,1,2,...,\}$. Then we fill the caches according to the list sequentially. Note that  the total memory required to fit the list is $i\times N$ units, which is equal to our total cache memory. Hence, the memory constraint is satisfied. 
	This storage policy makes sure that  each user has access to all files in the central server's catalog.  Hence, the worst-case transmission rate is 0 units.
	  
	  \textbf{\emph{Example:}} For ($N=3,K=3,L=2$) $-$ CCDN, at $i=2$,  the list is $\{1,2,3,1,2,3\}$ and Cache 1 is filled with files 1 and 2, Cache 2 is filled with files 3 and 1, Cache 3 is filled with files 2 and 3. Observe that with this storage policy, each user has access to all the files. Hence, the worst-case transmission rate is 0 units.
%\end{itemize}
For the remaining points $(M\neq iN/K)$, our achievable rate is derived by memory sharing. 
%\vspace{-0.5in}
%\section*{}

Now, we illustrate our policy with an example.

\textbf{\emph{Example:}} ($N=4,K=4,L=2$) $-$ CCDN\\
Consider an example with $N=4$ files, $K=4$ users / caches and each user connected to $L=2$ caches, as shown in Figure \ref{fig:m1storage}. 
%\textbf{Achievability:}
We discuss the achievability for  $\bigg|\Big[\big\lceil\frac{K}{L}\big\rceil\Big]\cup\{0\}\bigg|=3$ memory points $M=\{0,1,2\}$. At the remaining points, achievable rate is obtained by memory sharing. %Let the $4$ files be denoted by $\mathcal{F}_1=A$, $\mathcal{F}_2=B$, $\mathcal{F}_3=C$, and $\mathcal{F}_4=D$.
%\begin{figure}[t]
%	\begin{center}
%		\includegraphics[scale=0.33]{setup1b}
%		\caption{\sl A multi-access CCDN consisting of $N$ files, $K$ caches, each of size $M$ units, and $K$ users, each user is connected to $L=2$ caches.   \label{fig:problemsetupmultiaccess}}
%	\end{center}
%	\setlength{\belowcaptionskip}{-1in}
%\end{figure}
%{\usepackage{caption}
%\captionsetup{justification=centering}
\begin{figure}[t]
	\begin{center}
		\includegraphics[scale=0.45]{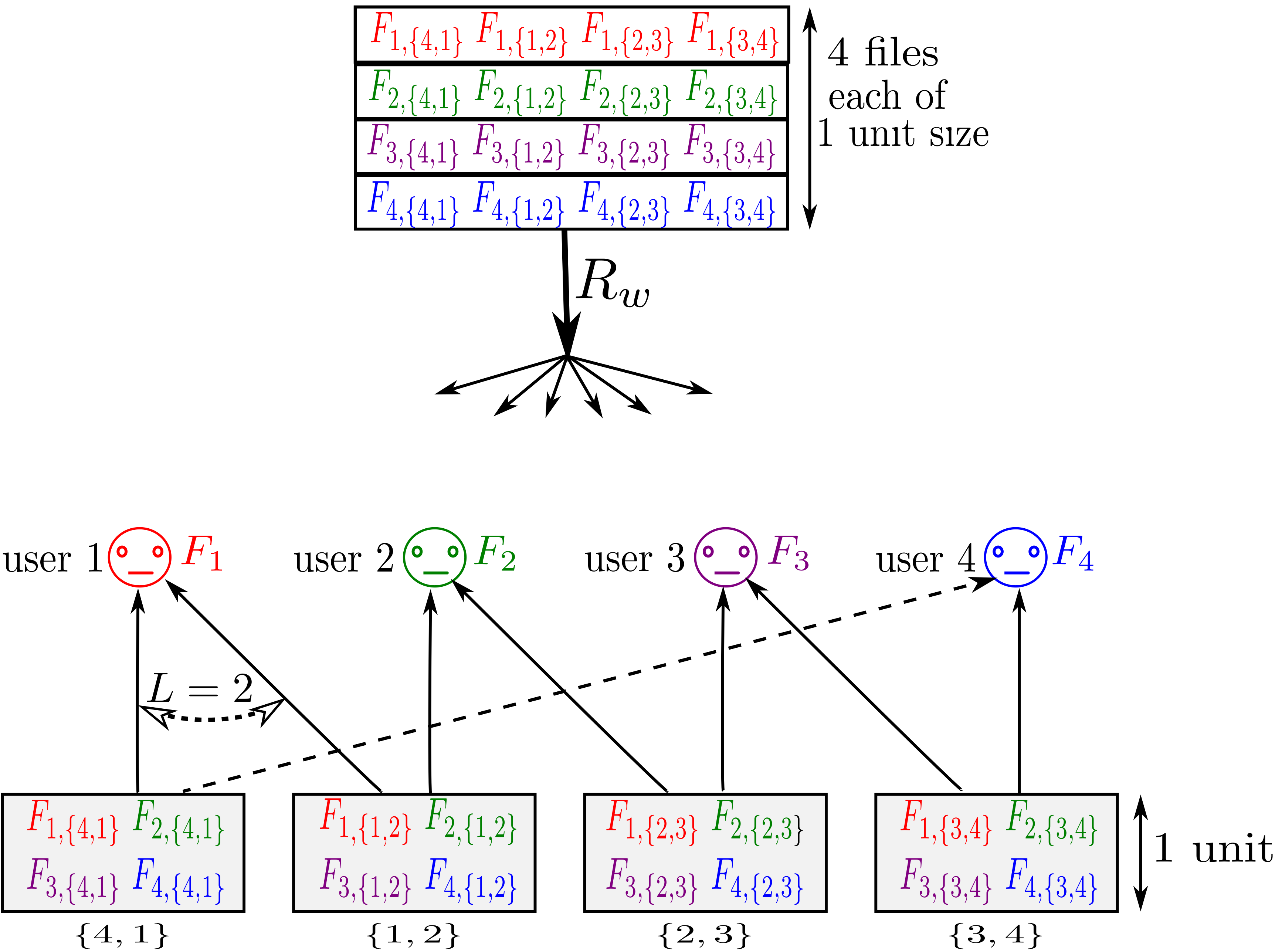}
		 \caption{Illustration of storage policy for ($N=4,K=4,L=2$)$-$CCDN at memory point $M=1$.} \label{fig:m1storage}
	\end{center}
\vspace{-0.5in}
\end{figure}
\addtocounter{subsection}{-3} 
\subsection{$M=0$ units} When M = 0, no file parts are stored in the cache. The
worst case request pattern is all users requesting different files.
Hence, $R(0)=4\text{ units.}$
\subsection{$M=1$ unit}
%\begin{enumerate}
%	\item[A.] $M=0$ $\hspace{1in}R(0)=4\text{ units.}$
%	\item[B.] $M=1$
	
	\subsubsection{Placement policy}
    At $M=1$, $i=\frac{MK}{N}=1$. Hence, $\hat{\mathcal{S}}=\{\{1\},\{2\},\{3\},\{4\}\}$.
	Split each file into $|\hat{\mathcal{S}}|=4$ parts and store the $1^{st}$ part of each file in Cache 1, the $2^{nd}$ part of each file in Cache 2, and so on as shown in Figure \ref{fig:m1storage}. Observe that the memory constraint ($M=1$) is satisfied. Since the $1^{st}$ cache is connected to User 4 and User 1, we subscript the stored content in Cache 1 with $\{4,1\}$, and repeat the same procedure for the other caches as well, i.e., 
	\begin{itemize}
		\item Cache 1 stores $\mathcal{F}_{1,\{4,1\}}$, $\mathcal{F}_{2,\{4,1\}}$, $\mathcal{F}_{3,\{4,1\}}$ and $\mathcal{F}_{4,\{4,1\}}$,
		\item Cache 2 stores $\mathcal{F}_{1,\{1,2\}}$, $\mathcal{F}_{2,\{1,2\}}$, $\mathcal{F}_{3,\{1,2\}}$ and $\mathcal{F}_{4,\{1,2\}}$,
		\item Cache 3 stores $\mathcal{F}_{1,\{2,3\}}$, $\mathcal{F}_{2,\{2,3\}}$, $\mathcal{F}_{3,\{2,3\}}$ and $\mathcal{F}_{4,\{2,3\}}$,
		\item Cache 4 stores $\mathcal{F}_{1,\{3,4\}}$, $\mathcal{F}_{2,\{3,4\}}$, $\mathcal{F}_{3,\{3,4\}}$ and $\mathcal{F}_{4,\{3,4\}}$.
	\end{itemize}
	\remove{
		Hence,
		\begin{itemize}
			\item {\color{blue}$A_{\{4,1\}}$, $B_{\{4,1\}}$, $C_{\{4,1\}}$, $D_{\{4,1\}}$, $A_{\{1,2\}}$, $B_{\{1,2\}}$, $C_{\{1,2\}}$, and $D_{\{1,2\}}$ are available {\color{blue}to} User 1. {\color{blue}($\because$ Subscripts contain 1).}}
			\item $A_{\{1,2\}}$, $B_{\{1,2\}}$, $C_{\{1,2\}}$, $D_{\{1,2\}}$, $A_{\{2,3\}}$, $B_{\{2,3\}}$, $C_{\{2,3\}}$, and $D_{\{2,3\}}$ are available {\color{blue}to} User 2.
			\item $A_{\{2,3\}}$, $B_{\{2,3\}}$, $C_{\{2,3\}}$, $D_{\{2,3\}}$, $A_{\{3,4\}}$, $B_{\{3,4\}}$, $C_{\{3,4\}}$, and $D_{\{3,4\}}$ are available {\color{blue}to} User 3.
			\item $A_{\{3,4\}}$, $B_{\{3,4\}}$, $C_{\{3,4\}}$, $D_{\{3,4\}}$, $A_{\{4,1\}}$, $B_{\{4,1\}}$, $C_{\{4,1\}}$, and $D_{\{4,1\}}$ are available {\color{blue}to} User 4.
		\end{itemize}
	}
	\subsubsection{Delivery policy}
	%	Let the request pattern be $\{A,B,C,D\}$. User 1 needs the files parts $\{A_{2,3}, A_{3,4}\}$, where, subscripts indicate file parts available at those users. Similarly, User 2 needs the file parts $\{B_{3,4}, B_{4,1}\}$, User 3 needs the file parts $\{C_{4,1}, C_{1,2}\}$, and User 4 needs the file parts $\{D_{1,2}, D_{2,3}\}$. 	
	Let the user request profile be $\{d_1,d_2.d_3,d_4\}$. In terms of subfiles,
	\begin{itemize}
		\item User 1 needs {\boldmath{\color{red}$\mathcal{F}_{d_1,\{4,1\}}$, $\mathcal{F}_{d_1,\{1,2\}}$,}} $\mathcal{F}_{d_1,\{2,3\}}$ and $\mathcal{F}_{d_1,\{3,4\}}$,
		\item Uset 2 needs {\boldmath{\color{red}$\mathcal{F}_{d_2,\{1,2\}}$, $\mathcal{F}_{d_2,\{2,3\}}$,}} $\mathcal{F}_{d_2,\{3,4\}}$ and $\mathcal{F}_{d_2,\{4,1\}}$,
		\item User 3 needs {\boldmath{\color{red}$\mathcal{F}_{d_3,\{2,3\}}$, $\mathcal{F}_{d_3,\{3,4\}}$,}} $\mathcal{F}_{d_3,\{4,1\}}$ and $\mathcal{F}_{d_3,\{1,2\}}$,
		\item User 4 needs {\boldmath{\color{red}$\mathcal{F}_{d_4,\{3,4\}}$, $\mathcal{F}_{d_4,\{4,1\}}$,}} $\mathcal{F}_{d_4,\{1,2\}}$ and $\mathcal{F}_{d_4\{2,3\}}$.
	\end{itemize}
	Note that the red color (bold font) subfiles are already available at the corresponding users, only the black color (normal font) subfiles are needed for them. Each user requires 2 subfiles and thus a total 8 subfiles are involved in the server transmission. We can map the problem here to an instance of the index coding problem described in Section~\ref{sec:prelim}, with $n=8$  nodes, each one requesting a distinct subfile. The side information at the  node representing (and requesting) some Subfile $i$ are the subfiles available to the  user which is requesting Subfile $i$. For example, the side information of the  node representing Subfile $\mathcal{F}_{d_1,\{3,4\}}$ are the subfiles available to User 1, i.e., $\mathcal{F}_{d_2,\{4,1\}}$, $\mathcal{F}_{d_3,\{4,1\}}$, $\mathcal{F}_{d_3,\{1,2\}}$, and $\mathcal{F}_{d_4,\{1,2\}}$.  
	We can solve this index coding problem to get the achievable transmission rate for our proposed scheme. Note that the number of messages transmitted by the central server in the CCDN is equal to the number of messages transmitted in the ICP, and the size of each message is equal to the size of a  subfile. 
 
 	To understand the structural properties of the above ICP, we form a $4\times2$ table (see Table \ref{Tab:4421}), such that the $p^{th}$ row and $q^{th}$ column contains Node ${F}_{d_p,\{<p+q>,<p+q+1>\}}$, i.e., the subfile requested by User $p$ and available at users $<p+q>,<p+q+1>$.  We refer to this entry as the Node $(p,q)$  where, $p\in[4],q\in[2]$. Note that the entries in the Row $p$ are the subfiles needed for User $p$.
 	We will use the notation $(p,\bar{q})$ to represent all the other nodes in Row $p$ excluding Node $(p,{q})$. For our example $(p,\bar{1})$ represents Node $(p,2)$ and $(p,\bar{2})$ represents Node $(p,1)$. 
 	
 		\begin{table}[h]
 				\vspace{-0.25in}
 					\begin{minipage}{.55\linewidth}
 			%\caption{}
 			\centering
 			%\resizebox{\linewidth}{!}{
 			\begin{tabular}{| c | c | }
 				\hline 
 				$\mathcal{F}_{d_1,\{2,3\}}$ & $\mathcal{F}_{d_1,\{3,4\}}$ \\
 				\hline 
 				$\mathcal{F}_{d_2,\{3,4\}}$ & $\mathcal{F}_{d_2,\{4,1\}}$ \\
 				\hline
 				$\mathcal{F}_{d_3,\{4,1\}}$ &  $\mathcal{F}_{d_3,\{1,2\}}$\\
 				\hline 
 				$\mathcal{F}_{d_4,\{1,2\}}$ &  $\mathcal{F}_{d_4,\{2,3\}}$ \\
 				\hline
 			\end{tabular}
 			\vspace{0.2cm}
 			\caption{ICP for $(N=4,K=4,L=2)$-CCDN. Row $p$ contains the subfiles needed for User $p$.}  \label{Tab:4421}
 					\end{minipage}%
 				\hspace{0.05\linewidth}
 		\begin{minipage}{.41\linewidth}
 			%\caption{}
 			\centering
 			%\resizebox{\linewidth}{!}{
 			\begin{tabular}{| c | c | }
 				\hline 
 				1 & 3\\
 				\hline 
 				2 & 4\\
 				\hline
 				3 &  1\\
 				\hline 
 				4 &  2 \\
 				\hline
 			\end{tabular}
 			\vspace{0.2cm}
 			\caption{coloring scheme for Table \ref{Tab:4421}}  \label{Tab:4422}
 		\end{minipage}
 	\vspace{-0.5in}
 			%\vspace{0cm}
 			%\caption{Global caption}
 	\end{table}
 Recall from Section \ref{sec:prelim} that the minimum broadcast rate of an ICP is upper bounded by its local chromatic number. A coloring scheme for an ICP  is said to be \textit{proper} if no node shares its color with any of its interfering nodes. %For Node $(p,q)$ in Table \ref{Tab:4421},
 Observe that for  Node $(p,q)$ in Table \ref{Tab:4421},  among the other nodes, only nodes $(p,\bar{q})$,   $(<p+1>,1)$, and  $(<p-1>,2)$ do not contain $p$ in the set which is represented by their subscripted rectangular brackets, and are thus interfering nodes. We take $4$ colors and for each $p\in[4]$ assign Color $p$ to nodes $(p,1)$ and $(<p-2>,2)$.  The coloring scheme is shown in Table \ref{Tab:4422}. The interfering nodes for nodes in the Row $i$ are present in rows $i-1$, $i$, $i+1$ but we repeat the color of Node ($i,1$) for its non-interfering Node ($<i-2>,2$) and the color of Node ($i,2$) for its non-interfering Node ($<i+2>,1$). So, this coloring scheme ensures that none of the nodes share its color with any of its interfering nodes.   Hence, this is a proper coloring scheme. 

	{The local chromatic number of an ICP is defined as the maximum number of different colors that appear in the closed anti-outneighborhood of any node, minimized over all proper colorings.} Observe that {for Node ($p,q$), 3 colors appear in the closed anti-outneighborhood, one color for Node ($p,q$), one color for Node ($p,\bar{q}$), one common color for nodes ($<p+1>,1$) and ($<p-1>,2$)}. Hence, an upper bound on  the local chromatic number of  the ICP is 3. Therefore, the central server transmits at most 3 messages to serve all the users.

	In this case, it is possible to get an explicit characterization of the optimal server transmission scheme and it is given by: 
	\begin{align*}
	&\mathcal{F}_{d_1,\{2,3\}}\oplus \mathcal{F}_{d_2,\{4,1\}}\oplus \mathcal{F}_{d_3,\{1,2\}} \oplus \mathcal{F}_{d_4,\{1,2\}},\\
	&\mathcal{F}_{d_2,\{3,4\}}\oplus \mathcal{F}_{d_3,\{1,2\}}\oplus \mathcal{F}_{d_4,\{2,3\}} \oplus \mathcal{F}_{d_1,\{2,3\}}, \\
	&\mathcal{F}_{d_3,\{4,1\}}\oplus \mathcal{F}_{d_4,\{2,3\}}\oplus \mathcal{F}_{d_1,\{3,4\}} \oplus \mathcal{F}_{d_2,\{3,4\}}.
	%	&\eqref{eq:4421} \oplus \eqref{eq:4422} \oplus \eqref{eq:4423}=D_{\{1,2\}}\oplus A_{\{3,4\}}\oplus B_{\{4,1\}}\oplus C_{\{4,1\}}. \label{eq:4424}
	\end{align*}
	It can be verified that using the above server transmissions and the accessible cache contents, each user can recover its requested file. 
	%Also, while we have illustrated the delivery scheme for the request profile $\{A,B,C,D\}$, a similar scheme works for any other request pattern say $\{X_1,X_2,X_3,X_4\}$ by substituting $A=X_1, B=X_2, C=X_3, D=X_4$ in equations \eqref{eq:4421},\eqref{eq:4422}, and \eqref{eq:4423}. 
	Since each message is of size $1/4$ units, the total server transmission size is given by $R(1)=3/4 \text{ units.}$ 
	\remove{Since, User 1 knows the file parts $\{B_{\{4,1\}},C_{\{1,2\}},C_{\{4,1\}}, \text{ and } D_{1,2}\}$ (recall that the subscripts represent the file part is available at corresponding users), User 1 gets $A_{\{2,3\}}$ from Equation \eqref{eq:4421} and $A_{\{3,4\}}$ from Equation \eqref{eq:4424}. Similarly, User 2 gets $B_{\{3,4\}}$ and $B_{\{4,1\}}$ from equations \eqref{eq:4421} and \eqref{eq:4422} and so on. Hence, by transmitting 3 messages \eqref{eq:4421},\eqref{eq:4422} and \eqref{eq:4423}, each user get their requested file in $\{A,B,C,D\}$ request pattern. 
		
		%For any other request pattern say $\{X_1,X_2,X_3,X_4\}$, substitute $A=X_1, B=X_2, C=X_3, D=X_4$ in equations \eqref{eq:4421},\eqref{eq:4422}, and \eqref{eq:4423} and recover the corresponding user requests.
	}
\subsection{$M=2$ units}
%	\item[C.] $M=2$
	
	Store files $\mathcal{F}_{1}$, $\mathcal{F}_2$ in Cache 1, Cache 3 and $\mathcal{F}_3$, $\mathcal{F}_4$ in Cache 2, Cache 4. Since each user has access to all files, the worst case transmission rate is $R(2)=0 \text{ units.}$
%\end{enumerate}

The transmission rate $R(M)$ at intermediate values is given by memory-sharing and thus the achievable rate-memory trade-off of our scheme is given by
\begin{align}
 R^*(M)\leq R(M) = \begin{cases}
4-{13M}/{4} & \text{if } 0\leq M\leq 1,  \\
{3}/{2}-{3M}/{4} & \text{if } 1\leq M\leq 2,\\
0 & \text{if } M\geq 2.
\end{cases}
\end{align}
On the other hand, the achievable rate ${R}_{color}(M)$ for the coloring-based scheme proposed in \cite{hachem2017codedmulti} is given by 
\begin{align}
 R^*(M)\leq {R}_{color}(M) = \begin{cases}
4-3M & \text{if } 0\leq M\leq 1,  \\
2-M & \text{if } 1\leq M\leq 2,\\
0 & \text{if } M\geq 2.
\end{cases}
\end{align}
Figure \ref{fig:442i} shows the comparison between the performance of the two schemes and demonstrates the improvement in rate using our modified placement scheme and an associated index coding-based delivery scheme. 

\remove{\begin{remark}\label{remark:k4u}
	For a system with $N$ files, $K=4$ caches, $K=4$ users, each one is connected to $L=2$ users,  
	\begin{itemize}
		\item  by considering $M=\{0,\frac{N}{4},\frac{2N}{4}\}$ as corner points, and following similar steps, we get
		\begin{align*}
		\therefore R^*(M)\leq R(M) = \begin{cases}
		4-\frac{13}{N}M, & \text{if } 0\leq M\leq \frac{N}{4},  \\
		\frac{3}{2}-\frac{3}{N}M, & \text{if } \frac{N}{4}\leq M\leq \frac{2N}{4},\\
		0 & \text{if } M\geq \frac{2N}{4}.
		\end{cases}
		\end{align*}
		
		\item by coloring scheme, we get
		\begin{align*}
		\therefore R^*(M)\leq \widetilde{R}(M) = \begin{cases}
		4-\frac{12}{N}M, & \text{if } 0\leq M\leq \frac{N}{4},  \\
		2-\frac{4}{N}M, & \text{if } \frac{N}{4}\leq M\leq \frac{2N}{4},\\
		0 & \text{if } M\geq \frac{2N}{4}.
		\end{cases}
		\end{align*}
	\end{itemize}
	The difference is 
	\begin{align*}
	\widetilde{R}(M)-R(M) = \begin{cases}
	\frac{1}{N}M, & \text{if } 0\leq M\leq \frac{N}{4},  \\
	\frac{1}{2}-\frac{1}{N}M, & \text{if } \frac{N}{4}\leq M\leq \frac{2N}{4},\\
	0 & \text{if } M\geq \frac{2N}{4}.
	\end{cases}
	\end{align*}
\end{remark}

The difference between the achievable rates of  coloring-based policy and Our policy is
\begin{align*}
{R}_{color}(M)-R(M) = \begin{cases}
\frac{M}{4}, & \text{if } 0\leq M\leq 1,  \\
\frac{1}{2}-\frac{M}{4}, & \text{if } 1\leq M\leq 2,\\
0 & \text{if } M\geq 2.
\end{cases}
\end{align*}}

%\begin{figure}[t]
%	\begin{center}
%		\includegraphics[scale=0.15]{rmplot1}
%		\caption{Rate vs Memory trade-off for the $(N=4,K=4,L=2)$-CCDN setup where the number of files are $N=4$, the number of caches and the number of users are $K=4$, and each user is connected to $L=2$ caches.    \label{fig:exampleratememory}}
%	\end{center}
%	\vspace{-0.5in}
%\end{figure}

\section{Main Results}
\label{sec:results_multiacc}
Now we evaluate the performance of our policy  in this section and characterize the optimal rate-memory trade-off $R^*(M)$ for multi-access coded caching system. As mentioned in Section \ref{sec:prelim}, most of our results are based on ICP bounds.  The proofs of our results are relegated to Section \ref{sec:proofs}.
\subsection{New achievable rate for general ($N$, $K$, $L$)$-$CCDN}
First, we provide a new upper bound on $R^*(M)$ for the general ($N$, $K$, $L$)$-$CCDN using our (uncoded) placement and delivery policy.  
\begin{theorem}\label{thm_new}
	Consider the general ($N,K,L$)$-$ CCDN with $L\in[K]$. Let $M$ be the cache size, and $R^*(M)$ be the optimal rate-memory trade-off under the restriction of uncoded placement. Then for $M=\frac{iN}{K}$, $i\in\{0\}\cup\big[\big\lceil\frac{K}{L}\big\rceil\big]$,  
	\begin{align*}
	R^*(M)\leq R_{\text{new}}(M)= \begin{cases}
	K\Big(1-\frac{LM}{N}\Big)^2 & \text{ if } i\in\{0\}\cup\big[\big\lfloor\frac{K}{L}\big\rfloor \big], \\
	0 & \text{ if } i= \big\lceil\frac{K}{L}\big\rceil
	\end{cases}
	\end{align*}
	is achievable using an uncoded placement scheme. For general $0\leq M \leq \frac{N}{K}\big\lceil\frac{K}{L}\big\rceil$,  the lower convex envelope of these points is achievable.
\end{theorem}

{As mentioned earlier,  \cite{hachem2017codedmulti} gave an upper bound{\footnote{The proposed scheme had an uncoded placement phase for the case when $L$ divides $K$.}} on $R^*(M)$, which \remove{using a scheme with a coded placement phase. The upper bound}is given by $R_{\text{color}}(M)= \min\{N/M,$ $K\}\Big(\frac{1-LM/N}{1+KM/N}\Big)$. Our achievable rate in Theorem \ref{thm_new} can sometimes be significantly  better than $R_{\text{color}}(M)$, see example below.\remove{for example when  {\color{red}$M\geq N(\frac{1}{L}-\frac{1}{K})$}. The following example illustrates this gap. }}

\textbf{\emph{Example:}} For \big($N,K,L=\frac{K-\sqrt{K}}{2}$\big)$-$CCDN,   $R_{\text{new}}\big({2N}/{K}\big) =1$ is a constant whereas $R_{\text{color}}\big({2N}/{K}\big)$ $={\sqrt{K}}/{6}$ grows unbounded as the number of users in the system $K$ increases.

Next, we specialize the result in Theorem \ref{thm_new} to the case of $L \ge K/2$, which is the regime for our remaining results. %From Theorem \ref{thm_new}, the achievable rate at memory points $0$, $N/K$, and $2N/K$ is $K$, ${(K-L)^2}/{K}$, and $0$ respectively. Denote the convex envelope of these points by $R_{ub}(M)$, given as follows:  

\begin{corollary} \label{cor_ub}
	Consider the $(N, K, L)$$-$CCDN with $L\geq K/2$. Let $M$ be the  cache size, $R^*(M)$ be the optimal rate-memory trade-off under the restriction of uncoded placement, then %and $R_{ub}(M)$ be as defined in \eqref{Eqn:UB1}. Then for any cache size $M$, there exists an uncoded placement and delivery policy such that the rate-memory pair ($R_{ub}(M),M$) is achievable, i.e., 
	\begin{align}
	R^*(M)\leq R_{ub}(M)= \begin{cases}
	K-\Big[K-\frac{(K-L)^2}{K}\Big]\frac{MK}{N} & \text{if } 0\leq M\leq \frac{N}{K},  \\
	\frac{(K-L)^2}{K}(2-\frac{MK}{N}) & \text{if } \frac{N}{K}\leq M\leq \frac{2N}{K},\\
	0 & \text{if } M\geq \frac{2N}{K}.
	\end{cases}
	\label{Eqn:UB1}
	\end{align}
%	\begin{align*}
%	 R_{ub}(M) .
%	\end{align*}
\end{corollary}
 \begin{proof}
 	{\bf Proof:}
 	From Theorem \ref{thm_new}, the achievable rate at memory points $0$, $\frac{N}{K}$, and $\frac{2N}{K}$ is $K$, $\frac{(K-L)^2}{K}$, and $0$ respectively. The convex envelope of these points is $R_{ub}(M)$. Hence $R^*(M)\leq R_{ub}(M)$.
 \end{proof}
% \begin{proof}
%	From Theorem \ref{thm_new}, the achievable rate at memory points $0$, $N/K$, and $2N/K$ is $K$, ${(K-L)^2}/{K}$, and $0$ respectively. The convex envelope of these points is $R_{ub}(M)$. Hence $R^*(M)\leq R_{ub}(M)$.
%\end{proof}

%Note that the achievable rate in above theorem is equal to the  $R_{ub}(M)$ if $L\geq \frac{K}{2}$.
\subsection{Order optimality for $L\geq {K}/{2}$}
As mentioned before, the gap between the achievable rate $R_{\text{color}}(M)$ and the information-theoretic lower bound derived in \cite{hachem2017codedmulti} scales with $L$, i.e., ${R_{\text{color}}(M)}/{R_{\text{inf}}^*(M)}\leq cL$, where $R^*_{\text{inf}}(M)$ is the information-theoretically optimal rate-memory trade-off and $c$ is some constant. While the characterization (exact or order-optimal) of $R^*_{\text{inf}}(M)$ for a general ($N, K, L$)$-$CCDN remains open, our following results establish the order-optimal (up to a factor of $2$) uncoded placement rate-memory trade-off $R^*(M)$ for any $(N, K, L)-$CCDN with $L \ge K/2$. For this, we  provide an improved lower bound on the server transmission rate for any valid (uncoded) placement and delivery scheme in Theorem \ref{thm_lb} and use the upper bound stated in Corollary \ref{cor_ub}. We derive the lower bound by mapping our setup to an appropriate ICP and using converse arguments for the ICP to derive lower bounds on the server transmission rate for our setup. Let $R_{lb}(M)$ be defined as follows: 

{
\begin{align}
{R}_{lb}(M)= \begin{cases}
K-\Big[K-\frac{(K-L)(K-L+1)}{2K}\Big]\frac{MK}{N} & \text{if } 0\leq M\leq \frac{N}{K},  \\
\frac{(K-L)(K-L+1)}{2K}(2-\frac{MK}{N}) & \text{if } \frac{N}{K}\leq M\leq \frac{2N}{K},\\
0 & \text{if } M\geq \frac{2N}{K}.
\end{cases}
\label{Eqn:LB1}
\end{align}}

\begin{theorem} \label{thm_lb}
Consider the $(N, K, L)$-CCDN with $L \ge K/2$. Let $M$ be the cache size, $R^*(M)$ be the optimal {rate-memory} trade-off under the restriction of uncoded placement, and $R_{lb}(M)$ be as defined in \eqref{Eqn:LB1}. Then we have
	\begin{align*}
	R^*(M)&\geq R_{lb}(M) .
	\end{align*}
\end{theorem}

%Next, we provide an upper bound on $R^*(M)$ by proposing a (uncoded) placement and delivery scheme for our setup and analyzing its achievable server transmission rate. 
The following corollary compares the upper and lower bounds on  $R^*(M)$ from Corollary~\ref{cor_ub} and Theorem~\ref{thm_lb} respectively and gives an approximate characterization of the optimal rate-memory trade-off 
$R^*(M)$.
\begin{corollary}\label{cor_ob}
		Consider the $(N, K, L)$-CCDN with $L \ge K/2$. Let $M$ be the cache size, $R^*(M)$ be the optimal rate-memory trade-off under the restriction of uncoded placement and $R_{ub}(M)$ be as defined in \eqref{Eqn:UB1}. Then we have 
		$$\frac{R_{ub}(M)}{R^*(M)}\leq 2.$$
\end{corollary}

We are able to give an approximate characterization of the optimal rate-memory trade-off $R^*(M)$ for $L\ge K/2$, because of improvement in both the upper and lower bounds. 

\subsection{Exact optimality for some special cases}
While Corollary~\ref{cor_ob} provides an approximate characterization of $R^*(M)$, for any $L \ge K/2$,  we now present some special cases where we are able to derive it exactly. %{\color{blue}In prior work, \cite{reddy2018multiaccess} characterized $R^*(M)$ for some small examples and also for the case of $L = K -1$.  }
\begin{theorem} \label{thm_ubexact}
	Consider the $(N, K, L)$-CCDN with $L \ge K/2$. Let $M$ be the cache size, $R^*(M)$ be the optimal rate-memory trade-off under the restriction of uncoded placement, and $R_{lb}(M)$ be as defined in \eqref{Eqn:LB1}. 
Then for any cache size $M$, we have ${R}^*(M)= {R}_{lb}(M)$ for the following scenarios: %(a) $L=K-1$, (b). $L=K-2$, (c). $L$ = $K - 3$, $K$ is even, and (d). $L = K-K/s + 1$ for some positive integer $s$.
\begin{enumerate}
	\item $L=K-1$,
\item $L=K-2$,
\item $L$ = $K - 3$, $K$ is even,
\item {$L = K-K/s + 1$ for some positive integer $s$.}
\end{enumerate}
\end{theorem}
The above result is proven by deriving improved achievability bounds for the mentioned cases and showing that the rate-memory pair ($R_{lb}(M),M$) is feasible for all $M$. This combined with Theorem \ref{thm_lb} then gives us the above result. Note that for $L$ = $K - 3$, we are able to characterize $R^*(M)$ exactly when $K$ is even. For the case of  $L$ = $K - 3$ with $K$ odd, we are able to show that the rate-memory pair ($R_{a}(M),M$) is achievable, where $R_{a}(M)$ is defined as follows:
\begin{align*}
R_a(M) = \begin{cases}
K-\Big[K-\frac{6}{K}-\frac{2}{K(K-1)}\Big]\frac{MK}{N} & \text{if } 0\leq M\leq \frac{N}{K},  \\
\Big[\frac{6}{K}+\frac{2}{K(K-1)}\Big]\big(2-\frac{MK}{N}\big) & \text{if } \frac{N}{K}\leq M\leq \frac{2N}{K},\\
0 & \text{if } M\geq \frac{2N}{K}.
\end{cases}
\end{align*}
Comparing $R_a(M)$ to the lower bound $R_{lb}(M)$ from Theorem \ref{thm_lb}, we can show that the additive gap is at most $\frac{2}{K(K-1)}$, which decreases to zero as the number of users in the system $K$ becomes large. 

\section{Discussions}\label{sec:discussions}
In summary, we derived new bounds for the ($N,K,L$)$-$CCDN and established the order-optimal uncoded placement rate-memory trade-off for the case of $L\geq K/2$. We also established the exact uncoded placement rate-memory trade-off for a few cases. Note that while our achievable rate works for any $L$, our lower bound is specifically tailored towards the case of $L \ge K/2$. Generalizing this bound to the case of $L < K/2$ and using it to extend the order-optimality result to this regime are natural directions for future research. 

While our ultimate goal is to characterize the exact rate-memory trade-off for the general $(N, K, L)$-CCDN problem, there are several challenges involved. While our proposed scheme works for the general setup, characterizing its rate-memory trade-off can be very difficult in general. This is because our achievability scheme is based on the solution of the corresponding ICP  and that in general is an NP-hard \cite{bar2011index} problem. It can also be quite difficult to get closed-form expressions for the known upper bounds. To get tighter upper bounds for the achievable rate of our proposed strategy is indeed a part of our future work.

 Our lower bound is specifically tailored towards the case of $L \ge K/2$. Generalizing this bound to the case of $L < K/2$ is also more challenging for the general setup.{\remove{For example, the achievable rate-memory trade-off of our policy for the $(N=6, K=6, L=2)$-CCDN is given by 
\begin{align*}
\therefore R^*(M)\leq R(M) = \begin{cases}
6-{25M}/{6}, & \text{if } 0\leq M\leq 1,  \\
{29}/{9}-{25M}/{8}, & \text{if } 1\leq M\leq 2,\\
{4}/{3}-{4M}/{9}, & \text{if } 2\leq M\leq 3,\\
0 & \text{if } M\geq 3.
\end{cases}
\end{align*}
Unfortunately, for this case we do not yet have matching lower bounds.} A critical step in proving the index coding-based lower bound (illustrated in Section~\ref{proof_lb}) was to generate several inequalities by considering different request patterns $\mathbf{d}$ and different user orders $\mathbf{u}$, and then combining them. For the $L=1$ case, 
\cite{wan2016optimality,yu2018exact} considered all possible permutations  $\mathbf{u}$ to get a tight lower bound. However, for the multi cache-access case with $L > 1$, not all user permutations are equivalent since some of the cache-access subsets are not feasible. 

For example, consider the $(N=4, K=4, L=2)$-CCDN. Recall the lower bound argument discussed in Section~\ref{proof_lb}. For any request profile $\mathbf{d}$, if we use the user permutation $\mathbf{u}=(2,4,3,1)$, we will get the inequality {$R^*(M)\geq |\mathcal{F}_{d_{2},\phi}|+ |\mathcal{F}_{d_{2},\{3,4\}}|+|\mathcal{F}_{d_{2},\{4,1\}}|+|\mathcal{F}_{d_{2},\{3,4,1\}}|+|\mathcal{F}_{d_{4},\phi}|+ |\mathcal{F}_{d_{3},\phi}|+|\mathcal{F}_{d_{1},\phi}|$}, which contains one fewer term $\big(|\mathcal{F}_{d_{3},\{4,1\}}|\big)$ than the inequality we get for the user permutation $\mathbf{u}=(2,3,4,1)$, given by $R^*(M)\geq |\mathcal{F}_{d_{2},\phi}|+ |\mathcal{F}_{d_{2},\{3,4\}}|+|\mathcal{F}_{d_{2},\{4,1\}}|+|\mathcal{F}_{d_{2},\{3,4,1\}}|+
|\mathcal{F}_{d_{3},\phi}|+|\mathcal{F}_{d_{3},\{4,1\}}|+ |\mathcal{F}_{d_{4},\phi}|+|\mathcal{F}_{d_{1},\phi}|$. 

For  $L\geq K/2$,  rotations (instead of permutations) for the user order $\mathbf{u}$ are all equivalent and suffice to get a tight lower bound.  Unfortunately, they do not suffice in the general case and hence, we need more sophisticated analysis to devise tight lower bounds for the general ($N$, $K$, $L$)-CCDN setup. This is also part of our future work. 

\section{Proofs}\label{sec:proofs}
Here, we provide the proofs for  all the results. The proofs are given in the following order:   Theorem 	\ref{thm_lb}, Theorem \ref{thm_ubexact}, Theorem \ref{thm_new}, and Corollary \ref{cor_ob}.
%\section{Proofs}\label{sec:proofs}
%Here, we provide the proofs for  all the results. The proofs are given in the following order:  Theorem 	\ref{thm_lb}, Theorem \ref{thm_ubexact}, Theorem \ref{thm_new}, and Corollary \ref{cor_ob}. Corollary \ref{cor_ub} proof is simple and the proof sketch is already given in  Section \ref{sec:results_multiacc}. %For detailed proofs of all the  results see \cite{reddy2018isitmultiaccess}.
\subsection{Proof of Theorem 	\ref{thm_lb}}
%\section{Lower bound for ( $N $, $K$, $L\geq K/2$ )-CCDN}
%\label{sec:lb}
\begin{proof}\label{proof_lb}
	In this section, we derive a lower bound on the server transmission rate for the general ($N $, $K$, $L\geq K/2$)$-$CCDN. We provide detailed calculations for the ($N=4$, $K=4$, $L=2$)$-$CCDN in  Appendix I. Instead of including all the details for the general case here, which would have been very cumbersome, we refer to calculations in the appendix, at appropriate junctures through the proof in this section.

%For better understanding, we briefly discuss the ($N=4$, $K=4$, $L=2$) $-$ CCDN lower bound in Appendix B. {\color{blue}The proof is {\color{red}a} generalization of the  ($N=4$, $K=4$, $L=2$)$-$CCDN lower bound.} 
The server has $N$ files $(\mathcal{F}_1,\mathcal{F}_2,...,\mathcal{F}_N)$. Any uncoded placement policy divides each file $\mathcal{F}_i$ into $2^K$ disjoint parts (subfiles), denoted by $\left\{\mathcal{F}_{i,\mathcal{W}}: \mathcal{W}\in P(\{1,2,...,K\})\right\}$, where $\mathcal{F}_{i,\mathcal{W}}$ denotes the part of file $F_i$ which is available (via the caches) exclusively to the users in $\mathcal{W}$, and $P(\mathcal{S}) \text{ denotes the power set of }\mathcal{S}$.  

Let $x_{i,j}$ denote the total size of the file parts (in units) which are are each stored on $j$ caches and are available to $i$ users. Hence,
\begin{align}
\sum_{i=0}^{K}\sum_{j=0}^{K}x_{i,j}=N \text{ (total size of all files)}, \label{eq:conv_totalfilesize1g}\\
\sum_{i=0}^{K}\sum_{j=0}^{K}jx_{i,j}\leq KM \text{ (total storage capacity)}. \label{eq:conv_totalcachesize1g}
\end{align}
In our setup $x_{i,j}$ can be non-zero for only some of the possible pairs $(i,j)$. After removing the combinations of $i$ and $j$ which are not possible, \eqref{eq:conv_totalfilesize1g}, \eqref{eq:conv_totalcachesize1g} become
\begin{align}
x_{0,0}+x_{L,1}+\sum_{j=2}^{K-L}\sum_{i=j-1}^{K-L-1}x_{L+i,j}+\sum_{j=2}^{K}x_{K,j}=N , \label{eq:conv_totalfilesizeg}\\
x_{L,1}+\sum_{j=2}^{K-L}\sum_{i=j-1}^{K-L-1}jx_{L+i,j}+\sum_{j=2}^{K}jx_{K,j}\leq KM .\label{eq:conv_totalcachesizeg}
\end{align}
Following along similar lines as the converse for the ($N=4$, $K = 4$, $L=2$)-CCDN case given in Appendix I, we calculate a lower bound for a particular user rotation and request pattern first and then sum it over  all possible user rotations and request patterns. Then, we get
\begin{align*}
K(N!)R^*(M)\geq&\sum_{\mathbf{d}}\sum_{\mathbf{u}}\sum_{j\in[1:K]}\sum_{W_j\in[1:K] \backslash \{u_1...u_j\}}|\mathcal{F}_{d_{u_j},W_j}|\\
=&K(N!)\frac{K}{N}x_{0,0}+K(N!)\frac{(K-L)(K-L+1)}{2NK}x_{L,1}+\\
&\sum_{j=2}^{K-L}\sum_{i=j-1}^{K-L-1}K(N!)\frac{(K-L-i)(K-L+1-i)}{2NK}x_{L+i,j}.
\end{align*}
%
%where recall that $x_{i,j}$ denotes the total size of the file parts (in units) which are  each stored on $j$ caches and are available to $i$ users.
Hence,

\begin{align}
R^*(M) \geq &\frac{K}{N}x_{0,0}+\frac{(K-L)(K-L+1)}{2NK}x_{L,1}+\nonumber \\
&\sum_{j=2}^{K-L}\sum_{i=j-1}^{K-L-1}\frac{(K-L-i)(K-L+1-i)}{2NK}x_{L+i,j} .\label{eq:conv_rateg}
\end{align}
If we substitute $x_{0,0}$ and $x_{L,1}$ from \eqref{eq:conv_totalfilesizeg} and \eqref{eq:conv_totalcachesizeg} in \eqref{eq:conv_rateg}, we get
\begin{align}
R^*(M)\geq K-\bigg[K-\frac{(K-L)(K-L+1)}{2K}\bigg]\frac{MK}{N}  \text{ units}.\label{eq:lb1}
\end{align}
If we substitute $x_{L,1}$ and $x_{K,2}$ from \eqref{eq:conv_totalfilesizeg} and \eqref{eq:conv_totalcachesizeg} in \eqref{eq:conv_rateg}, we get
\begin{align}
R^*(M)\geq \frac{(K-L)(K-L+1)}{2K}\Big[2-\frac{MK}{N}\Big]\text{ units}. \label{eq:lb2}
\end{align}
\eqref{eq:lb1} dominates for $0\leq M\leq \frac{N}{K}$, and \eqref{eq:lb2} dominates for $\frac{N}{K}\leq M\leq \frac{2N}{K}$. Hence, $R^*(M)\geq R_{lb}(M)$.
\end{proof}

\subsection{Proof of Theorem \ref{thm_ubexact}}\label{proofs_placement}
	\begin{proof}
		For all the cases discussed in Theorem \ref{thm_ubexact}, $L$ is greater than or equal to $K/2$. Therefore, for our achievability scheme, we consider  3 corner points $M=\{0,{N}/{K},{2N}/{K}\}$. As mentioned in Section \ref{sec:policy}  the rates $R=K$ and $R=0$ are achievable at memory points $M=0$ and $M=2N/K$ respectively.
		Now, we discuss the memory point $M=N/K$. %For better understanding, we also discuss the $K=5,L=3$ case in parallel.
		
		\subsubsection{Placement Policy for any $L\geq K/2$ at $M=N/K$}
		We have $i=\frac{MK}{N}=1$. From the definition of $\hat{\mathcal{S}}$ in \eqref{shat},  $\hat{\mathcal{S}}=\{\{j\}:j\in[K]\}$. We split each file into $|\hat{\mathcal{S}}|=K$ equal parts and store the $j^{th}$ part of each file in Cache $j$. Observe that this storage policy satisfies the memory constraint $M=\frac{N}{K}$ units. Since Cache $j$ is connected to User $j$, User $<j-1>$,..., User $<j-L+1>$, we subscript the stored content in Cache $j$ with the set $[<j-L+1>:j]$, i.e., 
		Cache $j$ stores $\mathcal{F}_{i,[<j-L+1>:j]}$ $\forall i \in [N]$.
		Recall that User $l$ has access to caches $l$, $<l+1>$, ..., $<l+L-1>$. Hence, for all $i\in[N]$, $k\in[L]$,  $\mathcal{F}_{i,[<l+k-L>:<l+k-1>]}$  are available {to} User $l$. %{\color{green}\big(Because, $\mathcal{F}_{i,[<l+k-L>:<l+k-1>]}$ is stored in Cache $<l+k-1>$ which is available to User $l$ if $k\in[L]$.\big)}

			\textbf{\emph{Example:}} Consider a ($N, K=5, L=3$)$-$CCDN. For all $i\in[N]$,
		\begin{itemize}
			\item  $\mathcal{F}_{i,[4:1]}$, $\mathcal{F}_{i,[5:2]}$ and $\mathcal{F}_{i,[1:3]}$ are available {to} User 1, 
			%{\color{green}(Since the set represented by subscripted rectangular brackets contains 1),}
			\item   $\mathcal{F}_{i,[5:2]}$, $\mathcal{F}_{i,[1:3]}$ and $\mathcal{F}_{i,[2:4]}$ are available {to} User 2,
			\item   $\mathcal{F}_{i,[1:3]}$,  $\mathcal{F}_{i,[2:4]}$ and $\mathcal{F}_{i,[3:5]}$ are available {to} User 3,
			\item   $\mathcal{F}_{i,[2:4]}$,  $\mathcal{F}_{i,[3:5]}$ and $\mathcal{F}_{i,[4:1]}$ are available {to} User 4, 
			\item  $\mathcal{F}_{i,[3:5]}$, $\mathcal{F}_{i,[4:1]}$ and $\mathcal{F}_{i,[5:2]}$ are available {to} User 5.
		\end{itemize}
	
			\subsubsection{Delivery Policy at $M=N/K$}
			Now, we discuss the delivery phase. Let the user request profile be $(d_1,d_2,...,d_K){\in[N]^K}$, i.e., User $l$ requests File $d_l$. User $l$ needs only those File $d_l$'s subfiles which are not available to him. Explicitly, User $l$ needs only $K-L$ subfiles $\mathcal{F}_{d_l,[<l+1>:<l+L>]}$, $\mathcal{F}_{d_l,[<l+2>:<l+L+1>]}$, ..., $\mathcal{F}_{d_l,[<l+K-L>:<l+K-1>]}$ which are stored in User $l$'s non accessible caches  $<l+L>$, $<l+L+1>$ ,...,  $<l+K-1>$ respectively. %Hence, 
%			\begin{itemize}
%				\item User 1 needs  $\mathcal{F}_{d_1,[2:<L+1>]}$, $\mathcal{F}_{d_1,[3:<L+2>]}$,..., $\mathcal{F}_{d_1,[K-L+1:K]}$,
%				\item User 2 needs  $\mathcal{F}_{d_2,[3:<L+2>]}$, $\mathcal{F}_{d_2,[4:<L+3>]}$,..., $\mathcal{F}_{d_2,[<K-L+2>:1]}$,
%				
%				%{\begin{center}
%				\hspace{1in}		\vdots
%				%\end{center}}
%				\item  User $l$ needs  $\mathcal{F}_{d_l,[<l+1>:<l+L>]}$,  $\mathcal{F}_{d_l,[<l+2>:<l+L+1>]}$,..., $\mathcal{F}_{d_l,[<l+K-L>:<l+K-1>]}$,
%				
%				
%				\hspace{1in}		\vdots
%				
%				\item User $K$ needs  $\mathcal{F}_{d_K,[1:K-L]}$, $\mathcal{F}_{d_K,[2:K-L+1]}$, ..., $\mathcal{F}_{d_K,[K-L:K-1]}$.
%			\end{itemize}

		\textbf{\emph{Example:}} For a ($N, K=5, L=3$)$-$CCDN, 
		%For $K=5$, $L=3$ case,  
		\begin{itemize}
			\item  User 1 needs  $\mathcal{F}_{d_1,[2:4]}$ and $\mathcal{F}_{d_1,[3:5]}$,
			\item  User 2 needs  $\mathcal{F}_{d_2,[3:5]}$ and $\mathcal{F}_{d_2,[4:1]}$,
			\item  User 3 needs  $\mathcal{F}_{d_3,[4:1]}$ and $\mathcal{F}_{d_3,[5:2]}$,
			\item  User 4 needs  $\mathcal{F}_{d_4,[5:2]}$ and $\mathcal{F}_{d_4,[1:3]}$,
			\item  User 5 needs  $\mathcal{F}_{d_5,[1:3]}$ and $\mathcal{F}_{d_5,[2:4]}$.
		\end{itemize}
	%Each user needs 2 subfiles and thus a total of $2K$ subfiles are involved in the server transmission. We can map the problem here to an instance of ICP described in Section~\ref{sec:prelim}, with $n=2K$ nodes, each one requesting a distinct subfile. The side information at the node representing (and requesting) some Subfile $r$ are the subfiles available to the user which is requesting Subfile $r$. 
	
	Each user needs $K-L$ subfiles and thus a total of $K(K-L)$ subfiles are involved in the server transmission. We can map the problem here to an instance of the ICP described in Section~\ref{sec:prelim}, with $n=K(K-L)$ nodes, each one requesting a distinct subfile. The side information at the node representing (and requesting) some Subfile $r$ are the subfiles available to the user which is requesting Subfile $r$.

Recall from Lemma \ref{lemma:upperbound} in Section \ref{sec:prelim}, the minimum broadcast rate of an ICP is upper bounded by its local chromatic number. % In our ICP, for the cases mentioned in Theorem \ref{thm_ubexact}, an upper bound on the local chromatic number is 
	 For the cases mentioned in Theorem \ref{thm_ubexact}, we show that an upper bound on the local chromatic number of the corresponding ICP's is ${(K-L)(K-L+1)}/{2}$, the details can be found in Appendix II. Since the server transmits at most ${(K-L)(K-L+1)}/{2}$ messages, each of size $1/K$ units, the broadcast rate at $M=N/K$ is upper bounded by $\frac{(K-L)(K-L+1)}{2K}$ units.

Hence, for all the cases in Theorem \ref{thm_ubexact}, the transmission rates  $R=K$, $R=\frac{(K-L)(K-L+1)}{2K}$ and $R=0$ are achievable at memory points $M=0$, $M=N/K$ and $M=2N/K$ respectively. %The transmission rate  at intermediate values is given by memory-sharing and is equal to $R_{lb}(M)$ defined in Equation \eqref{Eqn:LB1}. 
%Hence, we achieve $R_{lb}(M)$ and from Theorem \ref{thm_lb}, we  can conclude that  $R_{lb}(M)$ is the exact rate-memory trade-off ($R^*(M)$) for the cases mentioned in Theorem \ref{thm_ubexact}. 
The transmission rates at the intermediate values are derived using memory-sharing arguments. It can be verified that the achievable rate expression matches the lower bound $R_{lb}(M)$ derived in Theorem \ref{thm_lb} at all values of the memory $M$. This concludes the proof of the theorem.
\end{proof}

%**************************************************************************************************************************************************************************************************************************************************************************
%
%
%
%
\remove{
	\textit{(a). $L=K-1$:}  The ICP for $L=K-1$ is given in Table \ref{Tab:11}. Here, User $p$ needs only $1(=K-L)$ sub-file and it is $\mathcal{F}_{d_p,[<p+1>:<p+K-1>]}=\mathcal{F}_{d_p,[K]\backslash\{p\}}$.  We can observe that the sub-file needed for User $p$  is available to all the users in the system except User $p$.
%		\begin{table}[h]
%		\centering
%		\begin{tabular}{| c | c | c | c |}
%			\hline 
%			$F_{d_1,[2:<L+1>]}$ & $F_{d_1,[3:<L+2>]}$ & ... & $F_{d_1,[K-L+1:K]}$\\
%			\hline 
%			$F_{d_2,[3:<L+2>]}$ & $F_{d_2,[4:<L+3>]}$ & ... & $F_{d_2,[<K-L+2>:1]}$\\
%			\hline 
%			\vdots & \vdots & ... & \vdots \\
%			\hline
%			${F}_{d_p,[<p+1>:<p+L>]}$ &  ${F}_{d_p,[<p+2>:<p+L+1>]}$ & ... & ${F}_{d_p,[<p+K-L>:<p+K-1>]}$\\
%			\hline 
%			\vdots & \vdots & ... & \vdots \\
%			\hline
%			${F}_{d_K,[1:<L>]}$ &  ${F}_{d_K,[2:<L+1>]}$ & ... &  ${F}_{d_K,[<K-L>:<K-1>]}$\\
%			\hline
%		\end{tabular}
%	\vspace{0.2cm}
%	\caption{ICP for $L(\geq K/2)$-cache access CCDN. Row $p$ represents needed subfiles for User $p$.}  \label{Tab:g1}
%	\end{table}	

Hence, the central server sends the following  message:
\begin{align*}
\mathcal{F}_{d_1,[K]\backslash\{1\}}\oplus \mathcal{F}_{d_2,[K]\backslash\{2\}}\oplus \hdots \oplus \mathcal{F}_{d_K,[K]\backslash\{K\}}.
\end{align*}
It can be verified that using the above message and the accessible cache contents, each user can recover its requested file.  
Since each sub-file is of size ${1}/{K}$ units, $R({N}/{K})=\frac{1}{K}$ units $=\frac{(K-L)(K-L+1)}{2K}$ units.

\textit{(b). $L=K-2$:} The ICP for $L=K-2$ is given in Table \ref{Tab:221}.
	\begin{table}[h]
	\begin{minipage}{.6\linewidth}
		%\caption{}
		\centering
		%\resizebox{\linewidth}{!}{
		\begin{tabular}{| c | c |}
			\hline 
			$\mathcal{F}_{d_1,[2:K-1]}$ &  $\mathcal{F}_{d_1,[3:K]}$\\
			\hline 
			$\mathcal{F}_{d_2,[3:K]}$ & $\mathcal{F}_{d_2,[4:1]}$\\
			\hline 
			\vdots & \vdots  \\
			\hline
			$\mathcal{F}_{d_p,[<p+1>:<p+K-2>]}$ &   $\mathcal{F}_{d_p,[<p+2>:<p+K-1>]}$\\
			\hline 
			\vdots & \vdots \\
			\hline
			$\mathcal{F}_{d_K,[1:K-2]}$ & $\mathcal{F}_{d_K,[2:K-1]}$\\
			\hline
		\end{tabular}
		\vspace{0.2cm}
		\caption{ICP for $(N,K,L=K-2)$-CCDN.}  \label{Tab:221}
	\end{minipage}%
	\begin{minipage}{.4\linewidth}
		\centering
		%\caption{}
		%\resizebox{\linewidth}{!}{
		\begin{tabular}{| c | c |}
			\hline 
			$K-1$ & 1\\
			\hline 
			$K$ & $2$\\
			\hline 
			\vdots & \vdots \\
			\hline
			$<l-2>$ &  $l$ \\
			\hline
			\vdots & \vdots \\
			\hline 
			$K-2$ & $K$\\
			\hline
		\end{tabular}
		\vspace{0.2cm}
		\caption{coloring scheme for Table \ref{Tab:221} } \label{Tab:color22}
	\end{minipage} 
	%\vspace{0cm}
	%\caption{Global caption}
\end{table}

\begin{table}[h]
	\begin{minipage}{.3\linewidth}
		%\caption{}
		\centering
		%\resizebox{\linewidth}{!}{
		\begin{tabular}{| c | c |}
			\hline 
			$\mathcal{F}_{d_1,[2:4]}$ & $\mathcal{F}_{d_1,[3:5]}$\\
			\hline 
			$\mathcal{F}_{d_2,[3:5]}$ & $\mathcal{F}_{d_2,[4:1]}$\\
			\hline 
			$\mathcal{F}_{d_3,[4:1]}$ &  $\mathcal{F}_{d_3,[5:2]}$ \\
			\hline
			$\mathcal{F}_{d_4,[5:2]}$ &  $\mathcal{F}_{d_4,[1:3]}$\\
			\hline 
			$\mathcal{F}_{d_5,[1:3]}$ &  $\mathcal{F}_{d_5,[2:4]}$\\
			\hline
		\end{tabular}
		\vspace{0.2cm}
		\caption{Example: $(N,K=5,L=3)$-CCDN with request profile ($d_1,d_2,d_3,d_4,d_5$). } \label{Tab:222}
	\end{minipage}%
	\quad \quad
	\begin{minipage}{.2\linewidth}
		\centering
		%\caption{}
		%\resizebox{\linewidth}{!}{
		\begin{tabular}{| c | c |}
			\hline 
			$4$ & 1\\
			\hline 
			$5$ & $2$\\
			\hline 
			1 &  3 \\
			\hline
			2 & 4 \\
			\hline 
			3 & 5\\
			\hline
		\end{tabular}
		\vspace{0.2cm}
		\caption{coloring scheme for Table \ref{Tab:222} } \label{Tab:color22e}
	\end{minipage} 
	\quad \quad
	\begin{minipage}{.45\linewidth}
		\centering
		%\caption{}
		%\resizebox{\linewidth}{!}{
		\begin{tabular}{| c | c |}
			\hline 
			$4$ (NI)& {\color{red}1} (NI)\\
			\hline 
			$5$ (NI)& $2$ (I)\\
			\hline 
			{\color{red}1} &  3 (I)\\
			\hline
			2 (I)& 4 (NI)\\
			\hline 
			3 (NI)& 5 (NI)\\
			\hline
		\end{tabular}
		\vspace{0.2cm}
		\caption{An illustration of Node (3,1) for $K=5$. Here, NI means Non-Interfering node and  I means Interfering node } \label{Tab:color2e22}
	\end{minipage} 
	%\vspace{0cm}
	%\caption{Global caption}
\end{table}

Observe that for  Node $(p,q)$,  among the other nodes, only nodes $(p,\bar{q})$ (the other node in $p^{th}$ row),   $(<p+1>,1)$, and  $(<p-1>,2)$ do not contain $p$ in the set which is represented by their subscripted rectangular brackets, and are thus interfering nodes. Recall from the placement phase  that for a subfile of File $\mathcal{F}_{d_p}$, the subscripted rectangular bracket shows that the subfile is available at the users in the set represented by the subscripted rectangular brackets.As an illustration, we discuss the $K=5, L = 3$ example  in Table \ref{Tab:222}. Observe that for Node (3,1),  nodes (3,2), (4,1), and (2,2)  do not contain 3 in the set represented by their subscripted rectangular brackets, and are thus interfering nodes.

\textit{A proper coloring scheme:} Take $K$ colors and assign Color $p$ to nodes $(p,2)$ and $(<p+2>,1)$.  The coloring scheme for the general case is shown in Table \ref{Tab:color22} and in Table \ref{Tab:color22e} for $K=5$.  { The interfering nodes for nodes in the Row $i$ are present in rows $i-1$, $i$, $i+1$ but we repeat the colors of nodes in Row $i$ to nodes in rows  $i-2$, and $i+2$. i.e.,{\color{blue}So,} this coloring scheme ensures that none of the nodes share its color with any of its interfering nodes.  (See Table \ref{Tab:color2e22} for Node (3,1) color and its interfering {\color{blue}nodes'} colors).} Hence, this is a proper coloring scheme.

Observe that {for Node ($p,q$), 3 colors appear in the closed anti-outneighborhood, one color for Node ($p,q$), one color for Node ($p,\bar{q}$), one common color for nodes ($<p+1>,1$) and ($<p-1>,2$)}. Hence, the local chromatic number of  the ICP is 3. %Therefore, from Lemma \ref{lemma:upperbound}, we broadcast 3 messages. %Since, each sub-file is of size ${1}/{K}$ units, $R({N}/{K})={3}/{K}$ units$=\frac{(K-L)(K-L+1)}{2K}$ units.
Therefore, from Lemma \ref{lemma:upperbound}, an upper bound on the corresponding ICP is 3 and  here each sub-file is of size ${1}/{K}$ units. Hence, $R({N}/{K})={3}/{K}$ units$=\frac{(K-L)(K-L+1)}{2K}$ units.

\textit{(c). $L=K-3$, $K$ is even:} The ICP for $L=K-3$ is given in Table \ref{Tab:331}.
\begin{table}[h]
	\begin{minipage}{.6\linewidth}
		%\caption{}
		\centering
		%\resizebox{\linewidth}{!}{
		\begin{tabular}{| c | c | c|}
			\hline 
			$\mathcal{F}_{d_1,[2:K-2]}$ &  $\mathcal{F}_{d_1,[3:K-1]}$ &  $\mathcal{F}_{d_1,[4:K]}$\\
			\hline 
			$\mathcal{F}_{d_2,[3:K-1]}$ & $\mathcal{F}_{d_2,[4:K]}$ & $\mathcal{F}_{d_2,[5:1]}$\\
			\hline 
			\vdots & \vdots  & \vdots\\
			\hline
			$\mathcal{F}_{d_p,[<p+1>:<p+K-3>]}$ & 
			$\mathcal{F}_{d_p,[<p+2>:<p+K-2>]}$ &   $\mathcal{F}_{d_p,[<p+3>:<p+K-1>]}$\\
			\hline 
			\vdots & \vdots & \vdots\\
			\hline
			$\mathcal{F}_{d_K,[1:K-3]}$ & $\mathcal{F}_{d_K,[2:K-2]}$ & $\mathcal{F}_{d_K,[3:K-1]}$\\
			\hline
		\end{tabular}
		\vspace{0.2cm}
		\caption{ICP for $(N,K,L=K-2)$-CCDN.}  \label{Tab:331}
	\end{minipage}%
	\begin{minipage}{.45\linewidth}
		\centering
		%\caption{}
		%\resizebox{\linewidth}{!}{
		\begin{tabular}{| c | c | c |}
			\hline 
			$(1,K-2)$ & (2,1) & (1,1)\\
			\hline 
			$(1,K-1)$ & (2,2) & (1,2)\\
			\hline
			$(1,K)$ & (2,1) & (1,3)\\
			\hline 
			$(1,1)$ & (2,2) & (1,4)\\
			\hline
			\vdots & \vdots & \vdots \\
			\hline 
			$(1,K-4)$ & (2,1) &$(1,K-1)$\\
			\hline 
			$(1,K-3)$ & (2,2) &$(1,K)$\\
			\hline
		\end{tabular}
		\vspace{0.2cm}
		\caption{Coloring scheme for Table \ref{Tab:331} } \label{Tab:color33}
	\end{minipage} 
	%\vspace{0cm}
	%\caption{Global caption}
\end{table}
	
	Observe that for  Node $(p,q)$,  among the other nodes, only nodes $(p,\bar{q})$ (the other nodes in $p^{th}$ row),   $(<p+1>,1)$, $(<p+1>,2)$, $(<p+2>,1)$, $(<p-1>,2)$, $(<p-1>,3)$, and  $(<p-2>,3)$ do not contain $p$ in the set which is represented by their subscripted rectangular brackets, and are thus interfering nodes. 
	
	\textit{A proper coloring scheme:} Take $K+2$ colors. Let the colors be $\{(1,1), (1,2),..., (1,K),$ $ (2,1),(2,2)\}$  and assign color $(1,p)$ to nodes $(p,3)$ and $(<p+3>,1)$ for $1\leq p \leq K$, assign color $(2,1)$ to $(<p>,2)$ if $<p>$ is odd and assign color $(2,2)$ to $(<p>,2)$ if $<p>$ is even.  The coloring scheme for general case is shown in Table \ref{Tab:color33} %{\color{blue} and in Table \ref{Tab:color3e} for $K=5$}. 
	This coloring scheme ensures that none of the nodes share its color with any of its interfering nodes.  %{\color{blue}(See Table \ref{Tab:color3e2} for Node (3,1) color and its interfering node colors.)} 
	Hence, this is a proper coloring scheme.   % The coloring scheme is shown in below Table \ref{3color}.
	
	Observe that {for Node ($p,q$), 6 colors appear in the closed anti-outneighborhood\remove{, one color for Node ($p,q$), one color each  for nodes ($p,\bar{q}$), one common color for nodes ($<p+1>,1$) and ($<p-2>,3$), one common color for nodes ($<p+2>,1$) and ($<p-1>,3$), one common color for nodes ($<p+1>,2$) and ($<p-1>,2$)}}. Hence, the local chromatic number of  the ICP is 6. %Therefore, from Lemma \ref{lemma:upperbound}, we broadcast 6 messages. Since, each sub-file is of size ${1}/{K}$ units, $R({N}/{K})={6}/{K}$ units $=\frac{(K-L)(K-L+1)}{2K}$ units.
%\end{table}
Therefore, from Lemma \ref{lemma:upperbound}, an upper bound on the corresponding ICP is 6 and  here each sub-file is of size ${1}/{K}$ units. Hence, $R({N}/{K})={6}/{K}$ units$=\frac{(K-L)(K-L+1)}{2K}$ units.

\textit{(d). $L = K(s-1)/s + 1$ for some positive integer $s$:} The ICP for $L = K(s-1)/s + 1$ (for some positive integer $s$) is given in Table \ref{Tab:K1}.

\begin{table}[h]
	\centering
	\begin{tabular}{| c | c | c | c |}
		\hline 
		$\mathcal{F}_{d_1,[2:<L+1>]}$ & $\mathcal{F}_{d_1,[3:<L+2>]}$ & ... & $\mathcal{F}_{d_1,[K-L+1:K]}$\\
		\hline 
		$\mathcal{F}_{d_2,[3:<L+2>]}$ & $\mathcal{F}_{d_2,[4:<L+3>]}$ & ... & $\mathcal{F}_{d_2,[<K-L+2>:1]}$\\
		\hline 
		\vdots & \vdots & ... & \vdots \\
		\hline
		$\mathcal{F}_{d_p,[<p+1>:<p+L>]}$ &  $\mathcal{F}_{d_p,[<p+2>:<p+L+1>]}$ & ... & $\mathcal{F}_{d_p,[<p+K-L>:<p+K-1>]}$\\
		\hline 
		\vdots & \vdots & ... & \vdots \\
		\hline
		$\mathcal{F}_{d_K,[1:L]}$ &  $\mathcal{F}_{d_K,[2:<L+1>]}$ & ... &  $\mathcal{F}_{d_K,[K-L:K-1]}$\\
		\hline
	\end{tabular}
	\vspace{0.1cm}
	\caption{ICP for $(N,K,L=K(s-1)/s + 1)$-CCDN.}  \label{Tab:K1}
\end{table}

%To understand the structural properties of the above ICP, we form a $K\times(K-L)$ table (see Table \ref{Tab:K1}) such that $p^{th}$ row and $q^{th}$ column contains node $\mathcal{F}_{d_p,[<p+q>,<p+q+1>,...<p+q+L-1>]}$ i.e., the subfile requested by User $p$ and available at users $<p+q>,<p+q+1>,...,<p+q+L-1>$. We call this node as $(p,q)$ node. \big(Here, $p\in[K],q\in[K-L]$\big). Note that the entries in the Row $p$ are the subfiles needed for User $p$.%The table looks similar  to \big($N (\geq K)$, $6$, $L = K-3$ \big)-CCDN case with $K$ rows and $3$ columns.
%The interference nodes of  Node ($p,q$) are
%\begin{itemize}
%	\item In column 1, nodes $(<p>,1),(<p+1>,1),..., (<p+K-L-1>,1)$,
%	\item In column 2, nodes $(<p-1>,2),(<p>,2),..., (<p+K-L-2>,2)$,
%	\item In column 3, nodes $(<p-2>,3),(<p-1>,3),..., (<p+K-L-3>,2)$,
%	
%	\hspace{2in} \vdots
%%	\item In column $q$, nodes $(<p-q+1>,q),..., (<p-1>,q)$, \quad \quad $(<p+1>,q),..., (<p+K-L-q>,q)$,
%	
%%	\hspace{2in} \vdots
%	\item In column $K-L$, nodes $(<p-K-L+1>,K-L), ..., (<p>,K-L)$. 
%\end{itemize}

%
%We divide the interference nodes of any  Node ($p,q$) into two categories:
%%\end{itemize}
%\begin{enumerate}
%	\item intra interference nodes:  nodes $(p,\bar{q})$ (all the other nodes in Row $p$),
%	\item inter interference nodes: interference nodes present in other rows.
%\end{enumerate}
The interference nodes of  Node ($p,q$) are nodes $(p,\bar{q})$ (all the other nodes in Row $p$),
\begin{itemize}
	\item in column 1, nodes $(<p+1>,1),..., (<p+K-L-1>,1)$,
	\item in column 2, nodes $(<p-1>,2)$,\quad \quad  $(<p+1>,2),..., (<p+K-L-2>,2)$,
	%\item In column 3, nodes $(r-2,3), (r-1,3)$, \quad \quad $(r+1,3), (r+2,3),..., (r+K-L-3,3)$,
	
	\hspace{2in} \vdots
	\item in column $t$, nodes $(<p-t+1>,t),..., (<p-1>,t)$, \quad \quad $(<p+1>,t),..., (<p+K-L-t>,t)$,
	
	\hspace{2in} \vdots
	\item in column $K-L$, nodes $(<p-K-L+1>,K-L), ..., (<p-1>,K-L)$. 
\end{itemize}
%In each column exactly $K-L-1$ nodes are present.	

Our coloring scheme is different for $K-L$ odd and even cases. So, we will discuss them separately.

\begin{itemize}
	\item $K-L$ is odd

	%\textit{A proper coloring scheme:} Take $K$+2 colors. et the colors be $\{(1,1), (1,2),..., (1,K), (2,1),$ $(2,2)\}$  and assign color $(1,i)$ to nodes $(i,3)$ and $(<i+3>,1)$ for $1\leq i \leq K$, assign color $(2,1)$ to $(<i>,2)$ if $i$ is odd and assign color $(2,2)$ to $(<i>,2)$ if $i$ is even.  The coloring scheme for general case is shown in Table \ref{Tab:color2} and in Table \ref{Tab:color2e} for $K=5$. { This coloring scheme ensures that none of the nodes share its color with any of its interfering nodes.  (See Table \ref{Tab:color3e2} for Node (3,1) color and its interfering node colors).} Hence, this is a proper coloring scheme.
	
	\textit{A proper coloring scheme:} Take $\frac{(K-L)(K-L+1)}{2}$ colors. Let the colors be 
	$\{(q',p'): q'\in [\frac{K-L-1}{2}], p' \in [K-L+1]\} \quad \cup \quad \{(\frac{K-L+1}{2},p''):  p''\in[\frac{K-L+1}{2}]\}$. 
	
	\begin{itemize}
		\item $\forall p''\in[\frac{K-L+1}{2}]$, assign Color  $(\frac{K-L+1}{2},p'')$ to nodes $(<p''+d\frac{K-L+1}{2}>,\frac{K-L+1}{2})$ for all $d\in [2s]$. 
		\item $\forall q' \in [\frac{K-L-1}{2}],p'\in [K-L+1]$, assign Color $(q',p')$ to nodes 
		$(<p'+d(K-L+1)>, K-L+1-q')$ and $(<p'+d(K-L+1)+{K-L+1}-q'>,q')$   for all $d\in [s]$.
	\end{itemize}
	%{\color{red}	
		%If $L=K(s-1)/s + 1\implies K=s(K-L+1)$. We are repeating colors in a column with a circular gap of either  $(K-L+1)/2$ (for $\big(\frac{K-L+1}{2}\big)^{th}$ column) or $K-L+1$ (for remaining columns). Since, the number of nodes in a column are $K=i(K-L+1)$ and color repetition gap is $(K-L+1)/2$ (or $K-L+1$), the coloring scheme assigns only one color to a node (i.e., a color does not spread to all the nodes in a column, it restricts to $2i$ (or $i$) number of nodes in a column\remove{For example: Color (1,1) is assigned to columns 1 and $K-L$. In Column 1, Color (1,1) is assigned to nodes ($K-L+1$,1),...,($i(K-L+1),1$) but not to any other nodes}).}
	
	For Node $(p,q)$,
	\begin{itemize}
		\item [--] If $q=(K-L+1)/2$, color assigned to Node $(p,q)$ is repeated only in Column $q$ with a circular gap of $(<p-(K-L-1)/2>,q),...,(<p-1>,q)$ and $(K-L+1)/2$. Its interference nodes in $q^{th}$ column are nodes $(<p+1>,q),...,<(p+(K-L-1)/2>,q)$ and but we are assigning same color to (non interfering) nodes $(<p-d(K-L+1)/2>,q)$ and $(<p+d(K-L+1)/2>,q)$ for all $d\in [2s]$.
		\item [--] If $q<(K-L+1)/2$, color assigned to Node $(p,q)$ is repeated only in columns $q$ and $K-L+1-q$. Its interference nodes in $q^{th}$ column are nodes $(<p-q+1>,q),..., (<p-1>,q)$, $(<p+1>,q),..., (<p+K-L-q>,q)$, but we are assigning same color to (non interfering) nodes $(<p-(K-L+1)>,q)$ and $(<p+(K-L+1)>,q)$, in $(K-L+1-q)^{th}$ column are nodes $(<p-K+L+q>,K-L+1-q),..., (<p-1>,K-L+1-q)$, $(<p+1>,K-L+1-q),..., (<p+q-1>,K-L+1-q)$, but we are assigning same color to (non interfering) nodes $(<p-(K-L-q+1)>,K-L+1-q)$ and $(<p+q>,K-L+1-q)$.
		\item [--] If $q>(K-L+1)/2$, by using similar arguments, we can prove color is repeated only for non-interference nodes.
	\end{itemize}
	I.e., this coloring scheme ensures that none of the nodes share its color with any of its interfering nodes.   Hence, this is a proper coloring scheme. 
	\begin{table}[h]
		\centering
		\begin{tabular}{| c | c | c | c | c |}
			\hline 
			(1,2) NI & (2,3) NI & (3,1) NI & (2,1) NI & {\color{red}(1,1)} NI\\
			\hline 
			(1,3) NI & (2,4) NI & (3,2) NI & (2,2) NI & (1,2) I\\
			\hline
			(1,4) NI & (2,5) NI & (3,3) NI & (2,3) I & (1,3) I\\
			\hline 
			(1,5) NI & (2,6) NI & (3,1) I & (2,4) I & (1,4) I\\
			\hline
			(1,6) NI & (2,1) I& (3,2) I& (2,5) I& (1,5) I\\
			\hline 
			{\color{red}(1,1)} & (2,2) I & (3,3) I & (2,6) I & (1,6) I\\
			\hline 
			(1,2) I & (2,3) I & (3,1) I & (2,1) I & {\color{red}(1,1)} NI\\
			\hline 
			(1,3) I & (2,4) I & (3,2) I & (2,2) NI & (1,2) NI\\
			\hline
			(1,4) I & (2,5) I & (3,3) NI & (2,3) NI & (1,3) NI\\
			\hline 
			(1,5) I & (2,6) NI & (3,1) NI & (2,4) NI & (1,4) NI\\
			\hline
			(1,6) NI & (2,1) NI & (3,2) NI & (2,5) NI & (1,5) NI\\
			\hline 
			{\color{red}(1,1)} NI & (2,2) NI & (3,3) NI & (2,6) NI & (1,6) NI\\
			\hline 
		\end{tabular}
		\vspace{0.2cm}
		\caption{Coloring scheme for  $(N,K=12,L=7)-$CCDN.  Here, NI means Non-Interfering node and  I means Interfering node for Node (6,1). } \label{Tab:colorK}
	\end{table}
	
	For example consider $\{N, 12, 7\}-$CCDN. The coloring scheme is shown  in Table \ref{Tab:colorK}. We also highlight Node (6,1) interference and non-interference nodes.

	Observe that {for Node ($p,q$), $\frac{(K-L)(K-L+1)}{2}$ colors\footnote{one color for Node ($p,q$), $K-L-1$ colors for $K-L-1$ interference nodes  ($p,\bar{q}$) in the $p^{th}$ row, and $\frac{(K-L)(K-L-1)}{2}$ colors for $(K-L)(K-L-1)$  interference nodes in the other nodes} appear in the closed anti-outneighborhood\remove{, one color for Node ($i,j$), one color each  for nodes ($i,\bar{j}$), one common color for nodes ($<i+1>,1$) and ($<i-2>,3$), one common color for nodes ($<i+2>,1$) and ($<i-1>,3$), one common color for nodes ($<i+1>,2$) and ($<i-1>,2$)}}. Hence, the local chromatic number of  the ICP is $\frac{(K-L)(K-L+1)}{2}$. %Therefore, from Lemma \ref{lemma:upperbound}, we broadcast $\frac{(K-L)(K-L+1)}{2}$ messages. Since, each sub-file is of size ${1}/{K}$ units, $R({N}/{K})=\frac{(K-L)(K-L+1)}{2K}$ units.
	Therefore, from Lemma \ref{lemma:upperbound}, an upper bound on the corresponding ICP is $\frac{(K-L)(K-L+1)}{2}$ and  here each sub-file is of size ${1}/{K}$ units. Hence, $R({N}/{K})=\frac{(K-L)(K-L+1)}{2K}$ units.
	
	%The transmission rate  at intermediate values is given by memory-sharing and is equal to $R_{lb}(M)$ defined in \eqref{Eqn:LB1}. 
	%Hence, we achieve $R_{lb}(M)$ and from Theorem \ref{thm_lb}, we  can conclude that  $R_{lb}(M)$ is the exact rate-memory trade-off ($R^*(M)$). 
	
	\item $K-L$ is even 
	
	\textit{A proper coloring scheme:} Take $\frac{(K-L)(K-L+1)}{2}$ colors. Let the colors be 
	$\{(q',p'): q'\in [\frac{K-L}{2}],p' \in [K-L+1]\} $. 
	%\begin{itemize}
	%\item $\forall j\in[\frac{K-L+1}{2}]$, assign Color  $(\frac{K-L+1}{2},j)$ to nodes $(<j+r\frac{K-L+1}{2}>,\frac{K-L+1}{2})$ for all $r\in \mathbb{N}\cup \{0\}$. 
	$\forall q' \in [\frac{K-L}{2}],p'\in [K-L+1]$, assign Color $(q',p')$ to nodes 
	$(<p'+d(K-L+1)>,K-L+1-q')$ and $(<p'+d(K-L+1)+{K-L+1}-q'>,q')$   for all $d\in [s]$.
	
	For Node $(p,q)$,
	\begin{itemize}
		%\item [--] If $q=(K-L+1)/2$, color assigned to Node $(p,q)$ is repeated only in Column $q$ with a circular gap of $(K-L+1)/2$. Its interference nodes in $q^{th}$ column are nodes $(<p+1>,q),...,<(p+(K-L-1)/2>,q)$ and $(<p-1>,q),...,(<p-(K-L-1)/2>,q)$ but we are assigning same color to (non interfering) nodes $(<p-d(K-L+1)/2>,q)$ and $(<p+d(K-L+1)/2>,q)$ for all $d\in \mathbb{N}\cup \{0\}$.
		\item [--] If $q\leq(K-L+1)/2$, color assigned to Node $(p,q)$ is repeated only in columns $q$ and $K-L+1-q$. Its interference nodes in $q^{th}$ column are nodes $(<p-q+1>,q),..., (<p-1>,q)$, $(<p+1>,q),..., (<p+K-L-q>,q)$, but we are assigning same color to (non interfering) nodes $(<p-(K-L+1)>,q)$ and $(<p+(K-L+1)>,q)$, in $(K-L+1-q)^{th}$ column are nodes $(<p-K+L+q>,K-L+1-q),..., (<p-1>,K-L+1-q)$, $(<p+1>,K-L+1-q),..., (<p+q-1>,K-L+1-q)$, but we are assigning same color to (non interfering) nodes $(<p-(K-L-q+1)>,K-L+1-q)$ and $(<p+q>,K-L+1-q)$.
		\item [--] If $q>(K-L+1)/2$, by using similar arguments, we can prove color is repeated only for non-interference nodes.
	\end{itemize}
	%\end{itemize}
	I.e., this coloring scheme ensures that none of the nodes share its color with any of its interfering nodes.   Hence, this is a proper coloring scheme. 
%	\begin{table}[h]
%		\centering
%		\begin{tabular}{| c | c | c | c | c | c |}
%			\hline 
%			(1,2) NI & (2,3) NI & (3,4) NI & (3,1) NI & (2,1) NI & (1,1) I \\
%			\hline 
%			(1,3) NI & (2,4) NI & (3,5) NI & (3,2) NI & (2,2) I & (1,2) I\\
%			\hline
%			(1,4) NI & (2,5) NI & (3,6) NI & (3,3) I & (2,3) I & (1,3) I\\
%			\hline 
%			(1,5) NI & (2,6) NI & (3,7) I & (3,4) I & (2,4) I & (1,4) I\\
%			\hline
%			(1,6) NI & (2,7) I & (3,1) I & (3,5) I & (2,5) I & (1,5) I\\
%			\hline 
%			{\color{red}(1,7)} & (2,1) 	I & (3,2) I & (3,6) I & (2,6) I & (1,6) I\\
%			\hline 
%			(1,1) I & (2,2) I & (3,3) I & (3,7) I & (2,7) I & {\color{red}(1,7)} NI\\
%			\hline 
%			(1,2) I & (2,3) I & (3,4) I & (3,1) I & (2,1) NI & (1,1) NI\\
%			\hline 
%			(1,3) I & (2,4) I & (3,5) I & (3,2) NI & (2,2) NI & (1,2) NI\\
%			\hline
%			(1,4) I & (2,5) I & (3,6) NI & (3,3) NI & (2,3) NI & (1,3) NI\\
%			\hline 
%			(1,5) I & (2,6) NI & (3,7) NI & (3,4) NI & (2,4) NI & (1,4) NI\\
%			\hline
%			(1,6) NI & (2,7) NI & (3,1) NI & (3,5) NI & (2,5) NI & (1,5) NI\\
%			\hline 
%			{\color{red}(1,7)} NI & (2,1) NI & (3,2) NI & (3,6) NI & (2,6) NI & (1,6) NI\\
%			\hline 
%			(1,1) NI & (2,2) NI & (3,3) NI & (3,7) NI & (2,7) NI & {\color{red}(1,7)} NI\\
%			\hline
%		\end{tabular}
%		\vspace{0.1cm}
%		\caption{Coloring scheme for Table $(N,14,8)-$CCDN. NI: Non-interfering node I:Interfering node for  Node (6,1). } \label{Tab:colorKeven}
%	\end{table}
%	For example consider the $(N, 14, 8)-$CCDN. The coloring scheme is shown  in Table \ref{Tab:colorKeven}. We also highlight Node (6,1) interference and non-interference nodes.
	
	Observe that {for Node ($i,j$), $\frac{(K-L)(K-L+1)}{2}$ colors appear in the closed anti-outneighbor\\hood\remove{, one color for Node ($i,j$), one color each  for nodes ($i,\bar{j}$), one common color for nodes ($<i+1>,1$) and ($<i-2>,3$), one common color for nodes ($<i+2>,1$) and ($<i-1>,3$), one common color for nodes ($<i+1>,2$) and ($<i-1>,2$)}}. Hence, the local chromatic number of  the ICP is $\frac{(K-L)(K-L+1)}{2}$. %Therefore, from Lemma \ref{lemma:upperbound}, we broadcast $\frac{(K-L)(K-L+1)}{2}$ messages. Since, each sub-file is of size ${1}/{K}$ units, $R({N}/{K})=\frac{(K-L)(K-L+1)}{2K}$ units.
	Therefore, from Lemma \ref{lemma:upperbound}, an upper bound on the corresponding ICP is $\frac{(K-L)(K-L+1)}{2}$ and  here each sub-file is of size ${1}/{K}$ units. Hence, $R({N}/{K})=\frac{(K-L)(K-L+1)}{2K}$ units.
\end{itemize}

Hence, for all the cases in Theorem \ref{thm_ubexact}, the transmission rates  $R=K$, $R=\frac{(K-L)(K-L+1)}{2K}$ and $R=0$ are achievable at memory points $M=0$, $M=N/K$ and $M=2N/K$ respectively. The transmission rate  at intermediate values is given by memory-sharing and is equal to $R_{lb}(M)$ defined in \eqref{Eqn:LB1}. 

Hence, we achieve $R_{lb}(M)$ and from Theorem \ref{thm_lb}, we  can conclude that  $R_{lb}(M)$ is the exact rate-memory trade-off ($R^*(M)$) for the cases mentioned in Theorem \ref{thm_ubexact}.}
%\end{proof}
%%%%%\input{examplenewkeven.tex}
%\input{examplenewkorderwise.tex}
\subsection{Proof of Theorem \ref{thm_new}}
%\subsection{Theorem \ref{thm_ubexact}}
\begin{proof}
%Now, we discuss the Theorem \ref{thm_new} proof. Theorem \ref{thm_new} gives a new achievable rate of general $(N,K,L)-$CCDN. 
We first discuss the $L\geq K/2$ case. The proof follows along the same lines as Theorem \ref{thm_ubexact}. The placement scheme is identical and in the delivery phase, the key distinction is in the coloring scheme used for the corresponding ICPs.
%For better understanding, we first discuss {\color{red}the} $L\geq K/2$ case. {\color{blue}The proof follows the same steps of Theorem \ref{thm_ubexact}. The only difference is {\color{red}the} coloring scheme.}
\subsubsection{ $(N,K,L\geq K/2)-$CCDN}
%\subsubsection{$(N,K,L\geq K/2)-$CCDN}$\newline$
	For our achievability scheme, we consider  3 corner points $M=\{0,{N}/{K},{2N}/{K}\}$. As mentioned in Section \ref{sec:policy}  the rates $R=K$ and $R=0$ are achievable at memory points $M=0$ and $M=2N/K$ respectively.
	Now, we discuss  the memory point $M=N/K$.	
	
	We use the same placement scheme as given in the proof of Theorem \ref{thm_ubexact}, Section~\ref{proofs_placement} and following the same arguments before,
	 we get the ICP for $L\geq K/2$ as shown in Table \ref{Tab:lk1}.
		 %From Theorem \ref{thm_ubexact}, the ICP for $L\geq K/2$ is given in {\color{blue}Table \ref{Tab:lk1}}.  
		 In the table,
		 \begin{enumerate}
		 	\item Row $p$ of the table represents User $p$'s needed subfiles, 
		 	\item Node ($p,q$) is $F_{d_{p},[<p+q>:<p+q+L-1>]}$.
		 \end{enumerate}
	 Observe that the subscripts of the elements in a column are sliding windows of length $L$.
		 	\begin{table}[h]
		 	\begin{minipage}{.70\linewidth}
		 		%\caption{}
		 		\centering
		 		%\resizebox{\linewidth}{!}{
		 		\begin{tabular}{| c | c | c | c |}
		 			\hline 
		 			$F_{d_1,[2:<L+1>]}$ & $F_{d_1,[3:<L+2>]}$ & ... & $F_{d_1,[K-L+1:K]}$\\
		 			\hline 
		 			$F_{d_2,[3:<L+2>]}$ & $F_{d_2,[4:<L+3>]}$ & ... & $F_{d_2,[<K-L+2>:1]}$\\
		 			\hline 
		 			\vdots & \vdots & ... & \vdots \\
		 			\hline
		 			${F}_{d_p,[<p+1>:<p+L>]}$ &  ${F}_{d_p,[<p+2>:<p+L+1>]}$ & ... & ${F}_{d_p,[<p+K-L>:<p+K-1>]}$\\
		 			\hline 
		 			\vdots & \vdots & ... & \vdots \\
		 			\hline
		 			${F}_{d_K,[1:<L>]}$ &  ${F}_{d_K,[2:<L+1>]}$ & ... &  ${F}_{d_K,[<K-L>:<K-1>]}$\\
		 			\hline
		 		\end{tabular}
		 		\vspace{0.2cm}
		 		\caption{ICP for $(N,K,L\geq K/2)-$CCDN. Row $p$ contains the subfiles needed for User $p$. } \label{Tab:lk1}
		 	\end{minipage}%
		 	\hspace{0.05\linewidth}
		 	\begin{minipage}{.24\linewidth}
		 		\centering
		 		%\caption{}
		 		%\resizebox{\linewidth}{!}{
		 		\begin{tabular}{| c | c | c | c | }
		 			\hline 
		 			1 & 10 & 19 & 28  \\
		 			\hline 
		 			2 & 11 & 20 & 29 \\
		 			\hline
		 			3 & 12 & 21 & 30 \\
		 			\hline 
		 			4 & 13 & 22 & 31 \\
		 			\hline
		 			5 & 14 & 23 & 32 \\
		 			\hline 
		 			6 & 15 & 24 & 33 \\
		 			\hline 
		 			7 & 16 & 25 & 34 \\
		 			\hline 
		 			8 & 17 & 26 & 35 \\
		 			\hline 
		 			9 & 18 & 27 & 36 \\
		 			\hline
		 		\end{tabular}
		 		\vspace{0.2cm}
		 		\caption{coloring scheme for ($N,K=9,L=5$)$-$CCDN  } \label{Tab:color230}
		 	\end{minipage} 
		 	%\vspace{0cm}
		 	%\caption{Global caption}
		 		\vspace{-0.5in}
		 \end{table}	
	 	
		Recall from Section \ref{sec:prelim} that a coloring scheme for an ICP  is said to be \textit{proper} if no node shares its color with any of its interfering nodes. For our ICP, we take $K(K-L)$ colors and assign a unique color to each node. Explicitly, let the colors be $\{1,2,...K(K-L)\}$. For Node $(p,q)$ assign Color $(q-1)K+p$. This is a proper coloring scheme. The coloring scheme for $(N,K=9,L=5)-$CCDN is shown in Table \ref{Tab:color230}. %Since each node is assigned with a different color, this is a proper coloring scheme.

		Consider any Node ($p,q$). It can be verified that in each column, the index $p$
			%{\color{blue}Because of symmetrical property (Property (\emph{ii})) of the table, in each column, any $p$ 
			appears exactly $L$ times as subscript.  Hence, every node has $L$ non-interfering nodes in every column. Therefore the closed anti-outneighborhood of any node contains $(K-L)^2$ nodes. For each node, we assign a different color. Hence, the local chromatic number of  the ICP is $(K-L)^2$. Therefore from Lemma \ref{lemma:upperbound}, we broadcast $(K-L)^2$ messages. Since each sub-file is of size ${1}/{K}$ units, $R(N/K)={(K-L)^2}/{K}$ units $=	K\big(1-\frac{LM}{N}\big)^2$ units.

	%	We can observe this coloring scheme gives, for each node local chromatic number value 6.

%\end{proof}
\subsubsection{ $(N,K,L< K/2)-$CCDN}
%Now, we will discuss  the general $(N,K,L\leq K/2)-$CCDN new achievable rate.
	%\textit{(b). $(N,K,L\leq K/2)-$CCDN:}
%\subsubsection{$(N,K,L\leq K/2)-$CCDN}$\newline$
%\begin{proof}
	For our achievability scheme, we consider the corner points $M={iN}/{K}$, $i\in\{0\} \cup [\lceil K/L \rceil]$.  As mentioned in Section \ref{sec:policy}  the rates $R=K$ and $R=0$ are achievable at memory points $M=0$ and $M=\big\lceil\frac{K}{L} \big\rceil \frac{N}{K}$ respectively. Now, we  discuss  the remaining memory points. Let $M={iN}/{K}$, where $i\in [\lfloor K/L \rfloor]$.

	%The number of elements ($X$)in $\hat{\mathcal{S}}$ are 
	%$$X=|\hat{\mathcal{S}}|={K-iL+i-1 \choose i-1}\frac{K}{i}$$.
	Recall from Section \ref{sec:policy} that in our policy, we divide each file into $|\hat{\mathcal{S}}|={K-iL+i-1 \choose i-1}\frac{K}{i}$  parts, with one subfile corresponding to each subset $s \in \hat{\mathcal{S}}$ and store the subfile corresponding to set $s$ in all the $i$ caches whose index belongs to $s$.

	%Let $\hat{\mathcal{S}}$ be the union of subsets $s$ of $\{1,2,...,K\}$ which satisfy (\emph{i}) $|s|=i$, and (\emph{ii}) if { $i>1$}, every two elements $(j,l)$ of $s$ satisfy $|j-l|\geq L$ and $|K-|j-l||\geq L$. The number of elements ($X$)in $\hat{\mathcal{S}}$ are $$X=|\hat{\mathcal{S}}|={K-iL+i-1 \choose i-1}\frac{K}{i}$$.
	
	%\textit{Placement policy:}  Divide each file into $X$ parts, with one subfile corresponding to each subset $s \in \hat{S}$. Store the subfile corresponding to set $s$ in all the caches whose index belongs to $s$. 
	
	%\textit{Delivery policy:} 
	Let the user request profile be $\{d_1,d_2,...,d_K\}$, i.e., User $l$ requests File $d_l$. Some of the subfiles of File $d_l$  are already stored in User $l$'s accessible caches. Hence, User $l$ needs only those subfiles of File $d_l$ which are not stored in its accessible caches. It can be verified that the the number of subfiles needed for each user is $\frac{K}{i}{K-iL+i-1 \choose i-1}-L{K-iL+i-1 \choose i-1}={K-iL+i-1 \choose i}$.  Like $L\geq K/2$ case, we map the  problem here to an instance of the ICP described in Section~\ref{sec:prelim}, with $n=K{K-iL+i-1 \choose i}$ nodes, each one corresponding to a distinct subfile.  Now, we form a table with the following properties:
	\begin{enumerate}
		\item  Row $p$ of the table represents User $p$'s needed subfiles,
		\item  If Node ($1,q$) is $F_{d_{1},[s_1:<s_1+L-1>]\cup[s_2:<s_2+L-1>]\cup,...,\cup [s_i:<s_i+L-1>]}$, then  Node ($p,q$) is\\ $F_{d_{p},[s_1+p-1:<s_1+L-1+p-1>]\cup[s_2+p-1:<s_2+L-1+p-1>]\cup,...,\cup [s_i+p-1:<s_i+L-1+p-1>]}$.
	\end{enumerate} 
%	(\emph{i}) \\
%	(\emph{ii}) \\
	%Note that  $F_{d_r,[s_1+r-1:<s_1+L-1+r-1>]\cup[s_2+r-1:<s_2+L-1+r-1>]\cup,...,\cup [s_i+r-1:<s_i+L-1+r-1>]}$ is the subfile structure. The 
	The number of columns in the table is equal to the number of File $d_1$'s subfiles needed for User $1$, which is equal to ${K-iL+i-1 \choose i}$.

	For this ICP, we take $K{K-iL+i-1 \choose i}$ colors and assign one color for one node. Explicitly, let the colors be $\Big\{1,2,...K{K-iL+i-1 \choose i}\Big\}$. For Node $(p,q)$ assign Color $(q-1)K+p$. 
	
	The subscript of a node contains $iL$ elements. 	Consider any Node ($p,q$). It can be verified that in each column, the index $p$   appears exactly $iL$ times as subscript.  Hence, every node has $iL$ non-interfering nodes in every column. Therefore the closed anti-outneighborhood of any node contains $(K-iL){K-iL+i-1 \choose i}$ nodes. For each node, we assign a different color. Hence, the local chromatic number of  the ICP is bounded by $(K-iL){K-iL+i-1 \choose i}$. Therefore from Lemma \ref{lemma:upperbound}, we broadcast $(K-iL){K-iL+i-1 \choose i}$ messages. Since each sub-file is of size ${1}/|\hat{\mathcal{S}}|$ units, $R(M)={(K-iL){K-iL+i-1 \choose i}}/|\hat{\mathcal{S}}|={K\Big(1-\frac{LM}{N}\Big)^2}$ units.
\end{proof}

\subsection{Proof of Corollary \ref{cor_ob}}
\begin{proof}
\remove{An upper bound $R_{ub}(M)$  on the optimal rate-memory trade-off for $(N,K,L\geq K/2)-$CCDN is given in Corollary \ref{cor_ub}, and is 
\begin{align*}
R_{ub}(M)= \begin{cases}
K-\Big[K-\frac{(K-L)^2}{K}\Big]\frac{MK}{N}, & \text{if } 0\leq M\leq \frac{N}{K},  \\
\frac{(K-L)^2}{K}(2-\frac{MK}{N}), & \text{if } \frac{N}{K}\leq M\leq \frac{2N}{K},\\
0 & \text{if } M\geq \frac{2N}{K}.
\end{cases}
%\label{Eqn:UB1}
\end{align*}
A lower bound $R_{lb}(M)$ on the optimal rate-memory trade-off for $(N,K,L\geq K/2)-$CCDN is given in Theorem 	\ref{thm_lb}, and is 
\begin{align*}
{R}_{lb}(M)= \begin{cases}
K-\Big[K-\frac{(K-L)(K-L+1)}{2K}\Big]\frac{MK}{N}, & \text{if } 0\leq M\leq \frac{N}{K},  \\
\frac{(K-L)(K-L+1)}{2K}(2-\frac{MK}{N}), & \text{if } \frac{N}{K}\leq M\leq \frac{2N}{K},\\
0 & \text{if } M\geq \frac{2N}{K}.
\end{cases}
%\label{Eqn:LB1}
\end{align*}}
The upper bound $R_{ub}(M)$ and lower bound $R_{lb}(M)$ on the optimal rate-memory trade-off for $(N,K,L\geq K/2)-$CCDN are given in Corollary \ref{cor_ub} and Theorem 	\ref{thm_lb} respectively. 

%From Section \ref{sec:lb}, the lower bound is 
%	
%	\begin{align}
%	R^*(M)\geq K-\bigg[K-\frac{(K-L)(K-L+1)}{2K}\bigg]\frac{MK}{N}   \text{ units} \label{eq:conv_rate1g1} \\
%	R^*(M)\geq \frac{(K-L)(K-L+1)}{2K}\Big[2-\frac{MK}{N}\Big]\text{ units} \label{eq:conv_rate2}
%	\end{align}

Hence,	for $0 \leq M \leq \frac{N}{K}$,
	
	\begin{align*}
	\frac{R_{ub}(M)}{R^*(M)} \leq \frac{R_{ub}(M)}{R_{lb}(M)} & \leq \frac{K-\Big[K-\frac{(K-L)^2}{K}\Big]\frac{MK}{N}}{K-\bigg[K-\frac{(K-L)(K-L+1)}{2K}\bigg]\frac{MK}{N}}
%	&=\frac{K\big[1-\frac{MK}{N}\big]+\big[\frac{(K-L)^2}{K}\big]\frac{MK}{N}}{K\big[1-\frac{MK}{N}\big]+\big[\frac{(K-L)(K-L+1)}{2K}\big]\frac{MK}{N}}\\
	\leq \frac{K\big[1-\frac{MK}{N}\big]+\big[\frac{(K-L)^2}{K}\big]\frac{MK}{N}}{\frac{K}{2}\big[1-\frac{MK}{N}\big]+\big[\frac{(K-L)^2}{2K}\big]\frac{MK}{N}}
	=2,
	\end{align*}
	
	and for $\frac{N}{K}\leq M \leq \frac{2N}{K}$,
	\begin{align*}
	\frac{R_{ub}(M)}{R^*(M)} \leq \frac{R_{ub}(M)}{R_{lb}(M)} & \leq \frac{\frac{(K-L)^2}{K}\Big[2-\frac{MK}{N}\Big]}{\frac{(K-L)(K-L+1)}{2K}\Big[2-\frac{MK}{N}\Big]}
	\leq \frac{\frac{(K-L)^2}{K}\Big[2-\frac{MK}{N}\Big]}{\frac{(K-L)^2}{2K}\Big[2-\frac{MK}{N}\Big]}
	=2.
	\end{align*}
\end{proof}

\bibliographystyle{IEEEtran}
\bibliography{myref2}
\section*{Appendix I} \label{sec:442lb}
\section*{Lower bound for $(N=4.K=4,L=2)-$CCDN}
%\subsection{Lower bound for $(N=4.K=4,L=2)-$CCDN}\label{sec:442lb}
Recall that in ($N=4$, $K=4$, $L=2$) $-$ CCDN, the server has $N=4$ files $(\mathcal{F}_1,\mathcal{F}_2,\mathcal{F}_3,\mathcal{F}_4)$. Any uncoded placement policy divides each file $\mathcal{F}_i$ into $2^K=16$ disjoint parts (subfiles), denoted by $\left\{\mathcal{F}_{i,\mathcal{W}}: \mathcal{W}\in P([4])\right\}$, where $\mathcal{F}_{i,\mathcal{W}}$ denotes the part of file $\mathcal{F}_i$ which is available (via the caches) exclusively to the users in $\mathcal{W}$, and $P(\mathcal{S})$  denotes the power set of $\mathcal{S}$.  

Let $x_{i,j}$ denote the total size of the file parts (in units) which are each stored on $j$ caches and are available to $i$ users. For example, $x_{0,0 }$ indicates, total size of file  parts which are each stored in  none of the caches and are available to none of the  users, i.e., $x_{0,0 }=|\mathcal{F}_{1,\phi}|+|\mathcal{F}_{2,\phi}|+|\mathcal{F}_{3,\phi}|+|\mathcal{F}_{4,\phi}|$. Hence,
\begin{align}
\sum_{i=0}^{4}\sum_{j=0}^{4}x_{i,j}=4 \text{ (total size of all files)}, \label{eq:conv_totalfilesize1}\\
\sum_{i=0}^{4}\sum_{j=0}^{4}jx_{i,j}\leq 4M \text{ (total storage capacity)}. \label{eq:conv_totalcachesize1}
\end{align}
Our setup with $K = 4$ and $L = 2$ implies that $x_{i,j}$ can be non-zero for only some of the possible pairs $(i,j)$. {For example, $x_{1,1 }$ is not possible because if we store a file part in any cache, it will be available to 2 users, and hence $x_{1,1}=0$.} Table~\ref{tab:impo442} lists the combinations of $i$ and $j$ which are not possible, and hence for which $x_{i,j}=0$.

\begin{table}[h]
	\centering
	\begin{tabular}{|c|c|}
		\hline
		$i$ & $j$\\
		\hline
		0 & 1,2,3,4\\
		\hline
		1 & 0,1,2,3,4\\
		\hline
		2 & 0,2,3,4 \\
		\hline
		3 & 0,1,3,4\\
		\hline
		4 & 0,1\\
		\hline
	\end{tabular}
	\vspace{0.1in}
	\caption{Pairs $(i,j)$ for which $x_{i,j} = 0$}
	\label{tab:impo442}
		\vspace{-0.5in}
\end{table}

After removing the combinations of $i$ and $j$ which are not possible,  \eqref{eq:conv_totalfilesize1}, \eqref{eq:conv_totalcachesize1} become
\begin{align}
x_{0,0}+x_{2,1}+x_{3,2}+x_{4,2}+x_{4,3}+x_{4,4}=4 , \label{eq:conv_totalfilesize}\\
x_{2,1}+2x_{3,2}+2x_{4,2}+3x_{4,3}+4x_{4,4}\leq4M .\label{eq:conv_totalcachesize}
\end{align}
Let the request profile be $\mathbf{d}=(d_1,d_2,d_3,d_4)\in[1:4]^4$, where $d_i\neq d_j$ for all $i\neq j$. According to $\mathbf{d}$, User 1,2,3,4 requests $\mathcal{F}_{d_1}, \mathcal{F}_{d_2}, \mathcal{F}_{d_3},\mathcal{F}_{d_4}$ respectively. Similar to the  analysis of the $(N=4.K=4,L=2)-$CCDN achievable rate in Section \ref{sec:policy}, we generate an instance of the index coding problem for each request profile $\mathbf{d}$. There is a node corresponding to each subfile $\mathcal{F}_{i, \mathcal{W}}$ demanded by a real user in the caching system, which does not already have it in its cache. As before, the side information at the node representing (and requesting) some Subfile $i$ are the subfiles available to the user which is requesting Subfile $i$. Based on this, the side information graph for the index coding problem instance can be created.

Recall from Lemma~\ref{lemma:lowerbound} that any acyclic induced subgraph of the side information graph provides a lower bound on the server transmission rate of the index coding problem. For a request profile $\mathbf{d}$, consider a rotation $\mathbf{u}=(u_1,u_2,u_3,u_4)$ of $\{1,2,3,4\}$. For any such $\mathbf{u}$, a set of nodes inducing an acyclic subgraph in the side information graph is as follows:
\begin{itemize}
	\item $\mathcal{F}_{d_{u_1},\mathcal{W}_1}$ for all valid{\footnote{Our problem setup doesn't support some subsets. One example is $\{2,4\}$, because no cache is common to User 2 and User 4 and if it is stored in 2 caches then it will be available to at least 3 users.}} $\mathcal{W}_1 \subseteq [1:4]\backslash \{u_1\}$,
	\item $\mathcal{F}_{d_{u_2},\mathcal{W}_2}$ for all valid $\mathcal{W}_2 \subseteq [1:4]\backslash \{u_1,u_2\}$,
	\item $\mathcal{F}_{d_{u_3},\mathcal{W}_3}$ for all valid $\mathcal{W}_3 \subseteq [1:4]\backslash \{u_1,u_2,u_3\}$,
	\item $\mathcal{F}_{d_{u_4},\mathcal{W}_4}$ for all valid $\mathcal{W}_4 \subseteq [1:4]\backslash \{u_1,u_2,u_3,u_4\}$.
\end{itemize}   

For example, when $\mathbf{d}=(1,3,4,2)$ and $\mathbf{u}=(2,3,4,1)$, the selected nodes include
\begin{itemize}
	\item $d_{u_1}=d_2=3: \mathcal{F}_{3,\mathcal{W}_1}$ for valid subsets $\mathcal{W}_1 \subseteq \{3,4,1\}$,
	\item $d_{u_2}=d_3=4:\mathcal{F}_{4,\mathcal{W}_2}$ for valid subsets $\mathcal{W}_2 \subseteq\{4,1\}$,
	\item $d_{u_3}=d_4=2:F_{2,\mathcal{W}_3}$ for valid subsets $\mathcal{W}_3 \subseteq \{1\}$,
	\item $d_{u_4}=d_1=1:F_{1,\mathcal{W}_4}$ for valid subsets $\mathcal{W}_4 \subseteq\phi$.
\end{itemize}
This collection of nodes is {depicted} in Figure \ref{fig:acyclic}. {As} the figure illustrates, the corresponding subset $\mathcal{J}$ of nodes $(\mathcal{F}_{3,\phi}, \mathcal{F}_{3,\{3,4\}},\mathcal{F}_{3,\{4,1\}},\mathcal{F}_{3,\{3,4,1\}}, \mathcal{F}_{4,\phi},\mathcal{F}_{4,\{4,1\}}, \mathcal{F}_{2,\phi},\mathcal{F}_{1,\phi})$ does not induce a cycle in the side information graph. Then from Lemma \ref{lemma:lowerbound}, we have  
\begin{align}
R^*(M)\geq |\mathcal{F}_{3,\phi}|+ |\mathcal{F}_{3,\{3,4\}}|+|\mathcal{F}_{3,\{4,1\}}|+|\mathcal{F}_{3,\{3,4,1\}}|+
|\mathcal{F}_{4,\phi}|+|\mathcal{F}_{4,\{4,1\}}|+ |\mathcal{F}_{2,\phi}|+|\mathcal{F}_{1,\phi}|.\label{eq:particular1}
\end{align}

\begin{figure}[t]
	\begin{center}
		\includegraphics[scale=0.13]{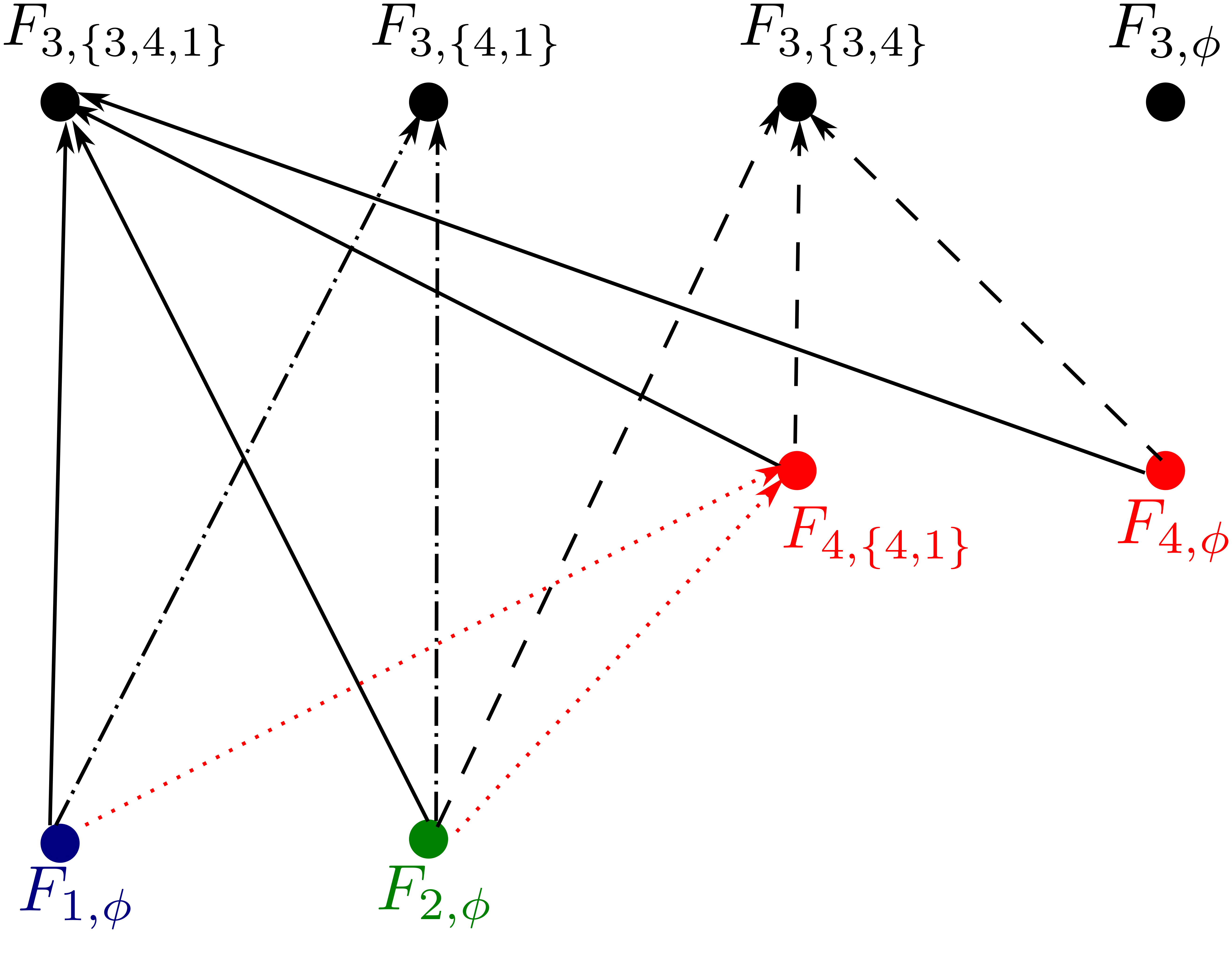}
		\caption{Directed acyclic subgraph  for request pattern $\mathbf{d}=(1,3,4,2)$ and user rotation $\mathbf{u}=(2,3,4,1)$.  All edges are in upward direction and hence the graph has no cycles.  \label{fig:acyclic}}
	\end{center}
	\vspace{-0.5in}
\end{figure}

\remove{{\color{red}Note that, if we take all possible combinations of user rotation patterns and user request profiles then each term will appear equal number of times. Hence, we are interested in number of terms instead of exact terms.
		For request pattern $\mathbf{d}=(1,3,4,2)$ and user rotation $\mathbf{u}=(2,3,4,1)$, the number of terms involved in the lower bound  are given in Table~\ref{Tab:not}\footnote{For better understanding, the total number of terms summed in each $x_{i,j}$ are given below.
			\begin{itemize}
				\item[--] $x_{0,0}=\sum_{r=1}^N |{F}_{r,\phi}|$, the number of terms involved is $N$,
				\item[--] $x_{2,1}=\sum_{r=1}^N\sum_{s=1}^{K} |{F}_{r,\{<s>,<s+1>\}}|$, the number of terms involved is $NK$,
				\item[--] $x_{3,2}=\sum_{r=1}^N\sum_{s=1}^{K} |{F}_{r,\{<s>,<s+1>,<s+2>\}}|$, the number of terms involved is $NK$,
				\item[--] $x_{4,2}=\sum_{r=1}^N |{F}_{r,\{1,2,3,4\}}|$, the number of terms involved is $N$  (Here, ${F}_{r,\{1,2,3,4\}}$ are the file parts which are stored in 2 caches and are available to 4 users),
				\item[--] $x_{4,3}=\sum_{r=1}^N |{F}_{r,\{1,2,3,4\}}|$, the number of terms involved is $N$ (Here, ${F}_{r,\{1,2,3,4\}}$ are the file parts which are stored in 3 caches and are available to 4 users),
				\item[--] $x_{4,4}=\sum_{r=1}^N |{F}_{r,\{1,2,3,4\}}|$, the number of terms involved is $N$  (Here, ${F}_{r,\{1,2,3,4\}}$ are the file parts which are stored in 4 caches and are available to 4 users).
		\end{itemize}}
		
		\begin{table}[h]
			\centering
			\begin{tabular}{|c|c|c|c|}
				\hline
				Related to & Total no. of terms    & No. of terms involved & Terms\\
				\hline
				$x_{0.0}$  & $N$ & $K=4$ & $F_{3,\phi},
				\mathcal{F}_{4,\phi},\mathcal{F}_{2,\phi},\mathcal{F}_{1,\phi}$ \\
				\hline
				$x_{{\color{blue}2},1}$ & $NK$ &  $\frac{(K-{\color{blue}2})(K-{\color{blue}2}+1)}{2}$ =3& $\mathcal{F}_{3,\{3,4\}},\mathcal{F}_{3,\{4,1\}},
				\mathcal{F}_{4,\{4,1\}}$\\
				\hline
				$x_{{\color{blue}3},2}$ & $NK$  & $\frac{(K-{\color{blue}3})(K-{\color{blue}3}+1)}{2}$  =1 & $\mathcal{F}_{3,\{3,4,1\}}$\\
				\hline
				$x_{{\color{blue}4},2}$ & $N$  & 0 &-\\
				\hline
				$x_{{\color{blue}4},3}$  & $N$  & 0 &-\\
				\hline
				$x_{{\color{blue}4},4}$ & $N$  & 0 & -\\
				\hline
			\end{tabular}
			\vspace{0.1cm}
			\caption{Description of number of terms involved in one particular user rotation pattern and request profile}
			\label{Tab:not}
\end{table}}}

In general, we can find a similar inequality as \eqref{eq:particular1} for each combination of request profiles $\mathbf{d}$ with distinct demands amongst the users ($4!$ permutations) and rotations $\mathbf{u}$ of the users ($4$ rotations). We then sum all the $4\times 4!$ ($K(N!)$) inequalities to obtain
\begin{align*}
K(N!)R^*(M)\geq&\sum_{\mathbf{d}}\sum_{\mathbf{u}}\sum_{j\in[1:4]}\sum_{W_j\in[1:4] \backslash \{u_1...u_j\}}|F_{d_{u_j},W_j}|\\
=&K(N!)\frac{K}{N}x_{0,0}+K(N!)\frac{(K-L)(K-{L}+1)}{2NK}x_{{L},1}+\\
&K(N!)\sum_{j=2}^{K-L}\sum_{i=j-1}^{K-L-1}\frac{(K-({L+i}))(K-({L+i})+1)}{2NK}x_{{L+i},j}\\
=&4\times (4!)\times x_{0,0}+6\times 3\times x_{2,1}+6 \times x_{3,2},
\end{align*}
where recall that $x_{i,j}$ denotes the total size of the file parts (in units) which are  each stored on $j$ caches and are available to $i$ users.

Hence,
\begin{align}
R^*(M) \geq x_{0,0}+3{x_{2,1}}/{16}+{x_{3,2}}/{16} .\label{eq:conv_rate}
\end{align}

If we substitute $x_{0,0}$ and $x_{2,1}$ from \eqref{eq:conv_totalfilesize} and \eqref{eq:conv_totalcachesize} in \eqref{eq:conv_rate}, we get
\begin{align}
&R^*(M)\geq 4-\frac{13}{4}M+\frac{11}{16}x_{3,2}+\frac{5}{8}x_{4,2}+\frac{23}{16}x_{4,3}+\frac{36}{16}x_{4,4}, \nonumber \\
\Longrightarrow \ &R^*(M)\geq 4-{13M}/{4} \text{ units}.\label{eq:conv_rate1}
\end{align}

If we substitute $x_{2,1}$ and $x_{4,2}$ from \eqref{eq:conv_totalfilesize} and \eqref{eq:conv_totalcachesize} in \eqref{eq:conv_rate}, we get
\begin{align}
%&R\geq \frac{3}{2}-\frac{3}{4}M+\frac{5}{8}x_{0,0}+\frac{1}{16}x_{3,2}+\frac{3}{16}x_{4,3}+\frac{3}{8}x_{4,4} \nonumber\\
R^*(M)\geq {3}/{2}-{3M}/{4}\text{ units}. \label{eq:conv_rate2}
\end{align}

Hence, from  \eqref{eq:conv_rate1} and \eqref{eq:conv_rate2},
\begin{align*}
R^*(M)\geq  \begin{cases}
4-13M/4 & \text{if } 0\leq M\leq 1,  \\
3/2-3M/4 & \text{if } 1\leq M\leq 2,\\
0 & \text{if } M\geq 2.
\end{cases}
\end{align*}

Note that the above lower bound matches with the upper bound we derived for ($N=4$, $K=4$, $L=2$)$-$CCDN in Section \ref{sec:policy}.  

%As illustrated in Figure \ref{fig:exampleratememory}, these lower bounds match exactly with the achievable rate of our proposed scheme for this setup. Thus, for the $(N=4, K=4, L=2)$-CCDN setup, we have characterized the exact rate-memory trade-off under the restriction of uncoded placement. 
%

\remove{\begin{remark}\label{remark:k4l}
		For a system with $N$ files, $K=4$ caches, $K=4$ users, each one is connected to $L=2$ users,  
		by following similar steps, we get
		\begin{align*}
		 R^*(M)\geq  \begin{cases}
		4-\frac{13}{N}M & \text{if } 0\leq M\leq \frac{N}{4},  \\
		\frac{3}{2}-\frac{3}{N}M & \text{if } \frac{N}{4}\leq M\leq \frac{2N}{4},\\
		0 & \text{if } M\geq \frac{2N}{4}.
		\end{cases}
		\end{align*}
		By Remark \ref{remark:k4u}, we can see that our policy is exactly optimal for a system with $N$ files, $K=4$ caches, $K=4$ users, each one is connected to $L=2$ users.
\end{remark}}
%\appendix
%\begin{appendices}
\section*{Appendix II} \label{upperbo}
\section*{Local chromatic number for the special cases}
%\subsection{\color{blue}Local chromatic number for the special cases}\label{upperbo}{\tiny }

Recall from the proof of Theorem \ref{thm_ubexact}  in Section \ref{proofs_placement} that each user needs $K-L$ subfiles and we map the problem there to an instance of ICP with $n=K(K-L)$ nodes.  Lemma \ref{lemma:upperbound} in  Section \ref{sec:prelim} states that the minimum broadcast rate of an ICP is upper bounded by its local chromatic number. In this section, we prove that the local chromatic number of the ICPs corresponding to the cases mentioned in Theorem \ref{thm_ubexact} is upper bounded by ${(K-L)(K-L+1)}/{2}$. 
	
	To understand the structural properties of the  ICP, we form a $K\times(K-L)$ table (see Table \ref{Tab:g1}), such that the entry corresponding to the $p^{th}$ row and $q^{th}$ column contains Node $\mathcal{F}_{d_p,[<p+q>:<p+q+L-1>]}$, i.e., the subfile requested by User $p$ and available at users $<p+q>,$ $<p+q+1>,...,<p+q+L-1>$.  We refer to this entry as the Node $(p,q)$  where, $p\in[K],q\in[K-L]$. Note that the entries in the Row $p$ are the subfiles needed for User $p$. We use the notation $(p,\bar{q})$ to represent all the other nodes in Row $p$ excluding Node ($p,q$).
	
	%{\color{cyan}In Table \ref{Tab:g1}, the subscripted rectangular subsets are circular sliding windows of length $L$ i.e., for $q^{th}$ column, first row's subscripted rectangular subset is $[<q+1>:<q+1+L-1>]$, second row's subscripted rectangular subset is $[<q+2>:<q+2+L-1>]$ and so on. }
	\begin{table}[h]
		%		\begin{minipage}{.70\linewidth}
		%\caption{}
		\centering
		%\resizebox{\linewidth}{!}{
		\begin{tabular}{| c | c | c | c |}
			\hline 
			$\mathcal{F}_{d_1,[2:<L+1>]}$ & $\mathcal{F}_{d_1,[3:<L+2>]}$ & ... & $\mathcal{F}_{d_1,[K-L+1:K]}$\\
			\hline 
			$\mathcal{F}_{d_2,[3:<L+2>]}$ & $\mathcal{F}_{d_2,[4:<L+3>]}$ & ... & $\mathcal{F}_{d_2,[<K-L+2>:1]}$\\
			\hline 
			\vdots & \vdots & ... & \vdots \\
			\hline
			$\mathcal{F}_{d_p,[<p+1>:<p+L>]}$ &  $\mathcal{F}_{d_p,[<p+2>:<p+L+1>]}$ & ... & $\mathcal{F}_{d_p,[<p+K-L>:<p+K-1>]}$\\
			\hline 
			\vdots & \vdots & ... & \vdots \\
			\hline
			$\mathcal{F}_{d_K,[1:<L>]}$ &  $\mathcal{F}_{d_K,[2:<L+1>]}$ & ... &  $\mathcal{F}_{d_K,[<K-L>:<K-1>]}$\\
			\hline
		\end{tabular}
		\vspace{0.2cm}
		\caption{ICP for $(N,K,L\geq K/2)$-CCDN. Row $p$ contains the subfiles needed for User $p$.}  \label{Tab:g1}
		%		\end{minipage}%
		%	\hspace{0.05\linewidth}
		%		\begin{minipage}{.24\linewidth}
		%			\centering
		%			%\caption{}
		%			%\resizebox{\linewidth}{!}{
		%			\begin{tabular}{| c |}
		%				\hline 
		%				$F_{d_1,[2:K]}$\\
		%				\hline 
		%				$F_{d_2,[3:1]}$\\
		%				\hline 
		%				\vdots \\
		%				\hline
		%				$F_{d_p,[p+1:p-1]}$ \\
		%				\hline
		%				\vdots  \\
		%				\hline 
		%				$F_{d_K,[1:K-1]}$\\
		%				\hline
		%			\end{tabular}
		%			\vspace{0.2cm}
		%			\caption{ICP for $(N,K,L=K-1)$-CCDN. } \label{Tab:11}
		%		\end{minipage} 
		%\vspace{0cm}
		%\caption{Global caption}
			\vspace{-0.5in}
	\end{table}
	
	%Recall from the placement phase of Theorem \ref{thm_ubexact} proof in Section \ref{sec:proofs} that for a subfile of File $\mathcal{F}_{d_j}$, the subscripted rectangular bracket shows that the subfile is available at the users in the set represented by the subscripted rectangular brackets. 
	For Node ($p,q$), the nodes which contain $p$ in the set which is represented by their subscripted rectangular brackets are available as side-information and are thus non-interfering nodes.   In Table \ref{Tab:g1}, we observe that any $p\in[K]$ will appear in the subscripted rectangular subsets of  $L$ nodes in each column and $p$ will not appear in $p^{th}$ row's nodes. Hence,  for any Node $(p,q)$, $K(K-L)$ nodes can be partitioned as follows: Node $(p,q)$ itself, $(K-L)L$ non-interfering nodes with $L$ in each column, $K-L-1$ interfering nodes in Row $p$ and the remaining $(K-L)(K-L-1)$ nodes which are also interfering nodes to  Node $(p,q)$. We call the $K-L-1$ interfering nodes in Row $p$ as intra-interference nodes and the remaining $(K-L)(K-L-1)$ interference nodes as inter-interference nodes  of Node $(p,q)$. Therefore, for a Node $(p,q)$, the total number of interfering nodes are $(K-L+1)(K-L-1)$ and it's closed anti-outneighborhood contains $(K-L)^2$ nodes.
	%To understand the structural properties of the above ICP, we form a $K\times2$ table (see Table \ref{Tab:21}), such that $p^{th}$ row and $q^{th}$ column contains  $\mathcal{F}_{d_p,[<p+q>:<p+q+K-3>]}$ i.e., the subfile requested by User $p$ and available at users $<p+q>,<p+q+1>,...,<p+q+K-3>$. We refer to this entry as the Node $(p,q)$  (Here, $p\in[K],q\in[2]$). Note that the entries in the Row $p$ are the subfiles needed for User $p$.

Recall from Section \ref{sec:prelim} that a coloring scheme for an ICP  is said to be \textit{proper} if no node shares its color with any of its interfering nodes. The local chromatic number of an ICP is defined as the maximum number of different colors that appear in the closed anti-outneighborhood of any node, minimized over all proper colorings. A viable proper coloring scheme is one which assigns a unique color to each of the $K(K-L)$ nodes.  From the proof of Theorem \ref{thm_new}, this gives an upper bound of $(K-L)^2$ on the local chromatic number. For the cases mentioned in Theorem \ref{thm_ubexact}, we devise an alternate proper coloring scheme which provides an improved upper bound on the local chromatic number.  In our improved coloring scheme, we assign colors such that $(K-L)(K-L-1)$ inter-interference nodes share $(K-L)(K-L-1)/2$ colors without violating the proper coloring scheme restrictions. Hence, the closed anti-outneighborhood of any node contains one color for the node itself, $K-L-1$ colors for intra-interference nodes and  $(K-L)(K-L-1)/2$ colors for inter-interference nodes, i.e., $(K-L)(K-L+1)/2$ colors in total. Therefore, the local chromatic number of the ICP is upper bounded by $(K-L)(K-L+1)/2$.{\footnote{Using Lemma \ref{lemma:lowerbound}, it can be shown that this bound is indeed a lower bound for the ICP mentioned in Table \ref{Tab:g1}.}}  %{\color{cyan}$(K-L)(K-L+1)/2$ is also a lower bound for the ICP mentioned in Table \ref{Tab:2a21}. Therefore, no further improvements are possible.}
\addtocounter{subsection}{-4}
\subsection{An upper bound on the local chromatic number for $L=K-1$}
The ICP for $L=K-1$ is given in Table \ref{Tab:1a1}.
\begin{table}[h]
	\begin{minipage}{.6\linewidth}
		%\caption{}
		\centering
		%\resizebox{\linewidth}{!}{
		\begin{tabular}{| c |}
			\hline 
			$\mathcal{F}_{d_1,[2:K]}$\\
			\hline 
			$\mathcal{F}_{d_2,[3:1]}$\\
			\hline 
			\vdots \\
			\hline
			$\mathcal{F}_{d_p,[p+1:p-1]}$ \\
			\hline
			\vdots  \\
			\hline 
			$\mathcal{F}_{d_K,[1:K-1]}$\\
			\hline
		\end{tabular}
		\vspace{0.2cm}
		\caption{ICP for $(N,K,L=K-1)$-CCDN. } \label{Tab:1a1}
	\end{minipage}%
	\begin{minipage}{.4\linewidth}
		\centering
		%\caption{}
		%\resizebox{\linewidth}{!}{
		\begin{tabular}{| c |}
			\hline 
			$1$ \\
			\hline 
			$2$ \\
			\hline 
			\vdots  \\
			\hline
			$l$  \\
			\hline
			\vdots  \\
			\hline 
			$K$ \\
			\hline
		\end{tabular}
		\vspace{0.2cm}
		\caption{coloring scheme for Table \ref{Tab:1a1} } \label{Tab:color1a2}
	\end{minipage} 
	%\vspace{0cm}
	%\caption{Global caption}
		\vspace{-0.5in}
\end{table}

%{\color{green}Recall from Theorem \ref{thm_ubexact} proof in Section \ref{sec:proofs}, User $p$ needs only $K-L=1$ sub-file and it is $\mathcal{F}_{d_p,[<p+1>:<p+K-1>]}(\mathcal{F}_{d_p,[K]\backslash\{p\}})$. For a subfile of File $\mathcal{F}_{d_p}$, the subscripted rectangular bracket shows that the subfile is available at the users in the set represented by the subscripted rectangular brackets.}
 Observe that, for  Node $(p,1)$,  all the other nodes contain $p$ in the set, which is represented by their subscripted rectangular brackets. Hence, all the other nodes' sub-files are available at Node $(p,1)$ and  thus, all the other nodes are non-interfering nodes. 
%{\color{green}Recall from Section \ref{sec:prelim} that a coloring scheme for an ICP  is said to be \textit{proper} if no node shares its color with any of its interfering nodes.   
%For the ICP of $L=K-1$ case, we take $K$ colors and assign Color $p$ to nodes $(p,1)$.  Table \ref{Tab:color1a2} shows the coloring scheme for the general case.   For any Node $(p,1)$, all the other nodes are non-interfering nodes.    Hence, this is a proper coloring scheme. 
%Again recall from Section \ref{sec:prelim} that the local chromatic number of an ICP is defined as the maximum number of different colors that appear in the closed anti-outneighborhood of any node, minimized over all proper colorings. For the $L=K-1$ case,  any node closed anti-out neighborhood size is 1. Hence, the local chromatic number is 1.} {\color{green}In fact, any coloring scheme is a proper coloring scheme and  gives the local chromatic number 1.}

 Therefore, the closed anti-outneighborhood cardinality is 1. A coloring scheme that assigns $K$ colors to the $K$ nodes is proper and gives the upper bound on the local chromatic number of $(K-L)(K-L+1)/2= 1$. Table \ref{Tab:color1a2} shows one such proper coloring scheme.
%
%{\color{blue}
%%\begin{remark}
%	Note that for $L=K-1$ case, the message $$
%	\mathcal{F}_{d_1,[K]\backslash\{1\}}\oplus \mathcal{F}_{d_2,[K]\backslash\{2\}}\oplus \hdots \oplus \mathcal{F}_{d_K,[K]\backslash\{K\}}
%	$$ is sufficient to serve all the users.
%%\end{remark}
%}
%
\subsection{An upper bound on the local chromatic number for $L=K-2$}
\label{K-2}
The ICP for $L=K-2$ is given in Table \ref{Tab:2a21}.
\begin{table}[h]
	\begin{minipage}{.6\linewidth}
		%\caption{}
		\centering
		%\resizebox{\linewidth}{!}{
		\begin{tabular}{| c | c |}
			\hline 
			$\mathcal{F}_{d_1,[2:K-1]}$ &  $\mathcal{F}_{d_1,[3:K]}$\\
			\hline 
			$\mathcal{F}_{d_2,[3:K]}$ & $\mathcal{F}_{d_2,[4:1]}$\\
			\hline 
			\vdots & \vdots  \\
			\hline
			$\mathcal{F}_{d_p,[<p+1>:<p+K-2>]}$ &   $\mathcal{F}_{d_p,[<p+2>:<p+K-1>]}$\\
			\hline 
			\vdots & \vdots \\
			\hline
			$\mathcal{F}_{d_K,[1:K-2]}$ & $\mathcal{F}_{d_K,[2:K-1]}$\\
			\hline
		\end{tabular}
		\vspace{0.2cm}
		\caption{ICP for $(N,K,L=K-2)$-CCDN.}  \label{Tab:2a21}
	\end{minipage}%
	\begin{minipage}{.4\linewidth}
		\centering
		%\caption{}
		%\resizebox{\linewidth}{!}{
		\begin{tabular}{| c | c |}
			\hline 
			1 & $<3>$\\
			\hline 
			2 & $<4>$\\
			\hline 
			\vdots & \vdots \\
			\hline
			$p$ &  $<p+2>$ \\
			\hline
			\vdots & \vdots \\
			\hline 
			$K$ & $<K+2>$\\
			\hline
		\end{tabular}
		\vspace{0.2cm}
		\caption{coloring scheme for Table \ref{Tab:2a21} } \label{Tab:color2a2}
	\end{minipage} 
	%\vspace{0cm}
	%\caption{Global caption}
		\vspace{-0.5in}
\end{table}

\begin{table}[h]
	\begin{minipage}{.6\linewidth}
		%\caption{}
		\centering
		%\resizebox{\linewidth}{!}{
		\begin{tabular}{| c | c |}
			\hline 
			$\mathcal{F}_{d_1,[2:4]}$ & $\mathcal{F}_{d_1,[3:5]}$\\
			\hline 
			$\mathcal{F}_{d_2,[3:5]}$ & $\mathcal{F}_{d_2,[4:1]}$\\
			\hline 
			$\mathcal{F}_{d_3,[4:1]}$ &  $\mathcal{F}_{d_3,[5:2]}$ \\
			\hline
			$\mathcal{F}_{d_4,[5:2]}$ &  $\mathcal{F}_{d_4,[1:3]}$\\
			\hline 
			$\mathcal{F}_{d_5,[1:3]}$ &  $\mathcal{F}_{d_5,[2:4]}$\\
			\hline
		\end{tabular}
		\vspace{0.2cm}
		\caption{ $(N,K=5,L=3)$-CCDN with request profile ($d_1,d_2,d_3,d_4,d_5$). } \label{Tab:2a22}
	\end{minipage}%
	\quad \quad
	\begin{minipage}{.35\linewidth}
		\centering
		%\caption{}
		%\resizebox{\linewidth}{!}{
		\begin{tabular}{| c | c |}
			\hline 
			1 & 3\\
			\hline 
			2 & 4\\
			\hline 
			3 &  5 \\
			\hline
			4 &1 \\
			\hline 
			5 & 2\\
			\hline
		\end{tabular}
		\vspace{0.2cm}
		\caption{coloring scheme for Table \ref{Tab:2a22} } \label{Tab:color2a2e}
	\end{minipage} 
	\vspace{-0.5in}
%	\quad \quad
%	\begin{minipage}{.45\linewidth}
%		\centering
%		%\caption{}
%		%\resizebox{\linewidth}{!}{
%		\begin{tabular}{| c | c |}
%			\hline 
%			$4$ (NI)& {\color{red}1} (NI)\\
%			\hline 
%			$5$ (NI)& $2$ (I)\\
%			\hline 
%			{\color{red}1} &  3 (I)\\
%			\hline
%			2 (I)& 4 (NI)\\
%			\hline 
%			3 (NI)& 5 (NI)\\
%			\hline
%		\end{tabular}
%		\vspace{0.2cm}
%		\caption{An illustration of Node (3,1) for $K=5$. Here, NI means Non-Interfering node and  I means Interfering node } \label{Tab:color2ae22}
%	\end{minipage} 
	%\vspace{0cm}
	%\caption{Global caption}
\end{table}

Observe that for  Node $(p,q)$,  among the other nodes, only nodes $(p,\bar{q})$,   $(<p+1>,1)$, and  $(<p-1>,2)$ do not contain $p$ in the set which is represented by their subscripted rectangular brackets, and are thus interfering nodes. %{\color{green}Recall from the placement phase of \ref{thm_ubexact} proof in Section VIII  that for a subfile of File $\mathcal{F}_{d_p}$, the subscripted rectangular bracket shows that the subfile is available at the users in the set represented by the subscripted rectangular brackets.} 
As an illustration, we discuss the $(N,K=5,L=3)-$CCDN example  in Table \ref{Tab:2a22}. Observe that for Node (3,1),  nodes (3,2), (4,1), and (2,2)  do not contain 3 in the set represented by their subscripted rectangular brackets, and are thus interfering nodes.
%{\color{blue}Recall from Section \ref{sec:prelim} that a coloring scheme for an ICP  is said to be \textit{proper} if no node shares its color with any of its interfering nodes and  
% the local chromatic number of an ICP is defined as the maximum number of different colors that appear in the closed anti-outneighborhood of any node, minimized over all proper colorings.}

%{\color{green}Recall from Section \ref{sec:prelim} that a coloring scheme for an ICP  is said to be \textit{proper} if no node shares its color with any of its interfering nodes.} 
For the ICP of $L=K-2$ case, we take $K$ colors and assign Color $p$ to nodes $(p,1)$ and $(<p-2>,2)$.  The coloring scheme for the general case is shown in Table \ref{Tab:color2a2} and in Table \ref{Tab:color2a2e} for $(N,K=5,L=3)-$CCDN.   The interfering nodes for nodes in  Row $i$ are present in rows $i-1$, $i$, $i+1$ but we repeat the color of Node ($i,1$) for its non-interfering Node ($<i-2>,2$) and the color of Node ($i,2$) for its non-interfering Node ($<i+2>,1$). So, this coloring scheme ensures that none of the nodes share its color with any of its interfering nodes.   Hence, this is a proper coloring scheme. Note that  the two inter-interference nodes (Node($i+1,1$) and Node ($i-1,2$)) of any node (Node ($i,q$)) share one color (Color $i+1$)}.
%{\color{green}This coloring scheme ensures that the a (every) color of any nodes interfering nodes in other rows repeated twice among them so that the local chromatic number value is $(K-L)(K-L+1)/2$, which is a lower bound for the ICP mentioned in Table \ref{Tab:2a21}.}  
%{\color{green}For the ICP of $L=K-2$ case, we take $K$ colors and assign Color $p$ to nodes $(p,2)$ and $(<p+2>,1)$.  The coloring scheme for the general case is shown in Table \ref{Tab:color22} and in Table \ref{Tab:color2a2e} for $K=5$.  { The interfering nodes for nodes in the Row $i$ are present in rows $i-1$, $i$, $i+1$ but we repeat the colors of nodes in Row $i$ to nodes in rows  $i-2$, and $i+2$. i.e.,{\color{blue}So,} this coloring scheme ensures that none of the nodes share its color with any of its interfering nodes.  (See Table \ref{Tab:color2e22} for Node (3,1) color and its interfering {\color{blue}nodes'} colors).} Hence, this is a proper coloring scheme. }  

%{\color{green}Again recall from Section \ref{sec:prelim} that the local chromatic number of an ICP is defined as the maximum number of different colors that appear in the closed anti-outneighborhood of any node, minimized over all proper colorings.} 
In Table \ref{Tab:2a22}, observe that for Node ($p,q$), 3 colors appear in the closed anti-outneighborhood, one color for Node ($p,q$), one color for intra-interference Node ($p,\bar{q}$), one common color for inter-interference nodes ($<p+1>,1$) and ($<p-1>,2$). Hence, an upper bound on  the local chromatic number of  the ICP is 
 ${(K-L)(K-L+1)}/{2}=3$. %Therefore, from Lemma \ref{lemma:upperbound}, we broadcast 3 messages. %Since, each sub-file is of size ${1}/{K}$ units, $R({N}/{K})={3}/{K}$ units$=\frac{(K-L)(K-L+1)}{2K}$ units.
%%Therefore, from Lemma \ref{lemma:upperbound}, an upper bound on the corresponding ICP is 3 and  here each sub-file is of size ${1}/{K}$ units. Hence, $R({N}/{K})={3}/{K}$ units{\color{green}$=\frac{(K-L)(K-L+1)}{2K}$ units}.
\subsection{An upper bound on the local chromatic number for  $L=K-3$, $K$ even}
\label{K-3} 
The ICP for $L=K-3$ is given in Table \ref{Tab:3a31}.
\begin{table}[h]
	\begin{minipage}{.6\linewidth}
		%\caption{}
		\centering
		%\resizebox{\linewidth}{!}{
		\begin{tabular}{| c | c | c|}
			\hline 
			$\mathcal{F}_{d_1,[2:K-2]}$ &  $\mathcal{F}_{d_1,[3:K-1]}$ &  $\mathcal{F}_{d_1,[4:K]}$\\
			\hline 
			$\mathcal{F}_{d_2,[3:K-1]}$ & $\mathcal{F}_{d_2,[4:K]}$ & $\mathcal{F}_{d_2,[5:1]}$\\
			\hline 
			\vdots & \vdots  & \vdots\\
			\hline
			$\mathcal{F}_{d_p,[<p+1>:<p+K-3>]}$ & 
			$\mathcal{F}_{d_p,[<p+2>:<p+K-2>]}$ &   $\mathcal{F}_{d_p,[<p+3>:<p+K-1>]}$\\
			\hline 
			\vdots & \vdots & \vdots\\
			\hline
			$\mathcal{F}_{d_K,[1:K-3]}$ & $\mathcal{F}_{d_K,[2:K-2]}$ & $\mathcal{F}_{d_K,[3:K-1]}$\\
			\hline
		\end{tabular}
		\vspace{0.2cm}
		\caption{ICP for $(N,K,L=K-3)$-CCDN.}  \label{Tab:3a31}
	\end{minipage}%
	\begin{minipage}{.45\linewidth}
		\centering
		%\caption{}
		%\resizebox{\linewidth}{!}{
		\begin{tabular}{| c | c | c |}
			\hline 
			$(1,1)$ & (2,1) & (1,4)\\
			\hline 
			$(1,2)$ & (2,2) & (1,5)\\
			\hline
			$(1,3)$ & (2,1) & (1,6)\\
			\hline 
			$(1,4)$ & (2,2) & (1,7)\\
			\hline
			\vdots & \vdots & \vdots \\
			\hline 
			$(1,K-1)$ & (2,1) &$(1,2)$\\
			\hline 
			$(1,K)$ & (2,2) &$(1,3)$\\
			\hline
		\end{tabular}
		\vspace{0.2cm}
		\caption{Coloring scheme for Table \ref{Tab:3a31} } \label{Tab:color3a3}
	\end{minipage} 
	%\vspace{0cm}
	%\caption{Global caption}
		\vspace{-0.5in}
\end{table}

Observe that for  Node $(p,q)$,  among the other nodes, only nodes $(p,\bar{q})$,   $(<p+1>,1)$, $(<p+2>,1)$, $(<p-1>,2)$, $(<p+1>,2)$, $(<p-2>,3)$, and  $(<p-1>,3)$ do not contain $p$ in the set which is represented by their subscripted rectangular brackets, and are thus interfering nodes. 

In this case, we take $K+2$ colors. Let the colors be $\{(1,1), (1,2),..., (1,K),$ $ (2,1),(2,2)\}$  and assign color $(1,p)$ to nodes $(p,1)$ and $(<p-3>,3)$ for $1\leq p \leq K$, assign color $(2,1)$ to $(p,2)$ if $p$ is odd and assign color $(2,2)$ to $(p,2)$ if $p$ is even.  The coloring scheme for general case is shown in Table \ref{Tab:color3a3}. %{\color{blue} and in Table \ref{Tab:color3e} for $K=5$}. 
This coloring scheme ensures that none of the nodes share its color with any of its interfering nodes. For example, the interfering nodes for Node ($1,1$) are nodes ($1,2$), ($1,3$), ($2,1$), ($3,1$), ($K,2$), ($2,2$), ($K-1,3$) and ($K,3$). We repeat its color ($1,1$) for a non-interfering node ($K-2, 3$).  %{\color{blue}For example, Node ($1,1$),   interference nodes are nodes   ($1,2$),   and , we are assigning its color ($1,1$) for its non-interfering node ($K-2,3$).} %{\color{blue}(See Table \ref{Tab:color3e2} for Node (3,1) color and its interfering node colors.)} 
Hence, this is a proper coloring scheme.   % The coloring scheme is shown in below Table \ref{3color}.

Observe that {for Node ($p,q$), 6 colors appear in the closed anti-outneighborhood\remove{, one color for Node ($p,q$), one color each  for nodes ($p,\bar{q}$), one common color for nodes ($<p+1>,1$) and ($<p-2>,3$), one common color for nodes ($<p+2>,1$) and ($<p-1>,3$), one common color for nodes ($<p+1>,2$) and ($<p-1>,2$)}}.  Hence, an upper bound on  the local chromatic number of  the ICP is 
${(K-L)(K-L+1)}/{2}=6$.
\subsection{An upper bound on the local chromatic number for $L = K(s-1)/s + 1$, $s\in\mathbb{N}$}
\label{K-j}
The ICP for $L = K(s-1)/s + 1$ (for some positive integer $s$) is given in Table \ref{Tab:Ka1}.
\begin{table}[h]
	\centering
	\begin{tabular}{| c | c | c | c |}
		\hline 
		$\mathcal{F}_{d_1,[2:<L+1>]}$ & $\mathcal{F}_{d_1,[3:<L+2>]}$ & ... & $\mathcal{F}_{d_1,[K-L+1:K]}$\\
		\hline 
		$\mathcal{F}_{d_2,[3:<L+2>]}$ & $\mathcal{F}_{d_2,[4:<L+3>]}$ & ... & $\mathcal{F}_{d_2,[<K-L+2>:1]}$\\
		\hline 
		\vdots & \vdots & ... & \vdots \\
		\hline
		$\mathcal{F}_{d_p,[<p+1>:<p+L>]}$ &  $\mathcal{F}_{d_p,[<p+2>:<p+L+1>]}$ & ... & $\mathcal{F}_{d_p,[<p+K-L>:<p+K-1>]}$\\
		\hline 
		\vdots & \vdots & ... & \vdots \\
		\hline
		$\mathcal{F}_{d_K,[1:L]}$ &  $\mathcal{F}_{d_K,[2:<L+1>]}$ & ... &  $\mathcal{F}_{d_K,[K-L:K-1]}$\\
		\hline
	\end{tabular}
	\vspace{0.1cm}
	\caption{ICP for $(N,K,L=K(s-1)/s + 1)$-CCDN.}  \label{Tab:Ka1}
		\vspace{-0.5in}
\end{table}
%To understand the structural properties of the above ICP, we form a $K\times(K-L)$ table (see Table \ref{Tab:K1}) such that $p^{th}$ row and $q^{th}$ column contains node $\mathcal{F}_{d_p,[<p+q>,<p+q+1>,...<p+q+L-1>]}$ i.e., the subfile requested by User $p$ and available at users $<p+q>,<p+q+1>,...,<p+q+L-1>$. We call this node as $(p,q)$ node. \big(Here, $p\in[K],q\in[K-L]$\big). Note that the entries in the Row $p$ are the subfiles needed for User $p$.%The table looks similar  to \big($N (\geq K)$, $6$, $L = K-3$ \big)-CCDN case with $K$ rows and $3$ columns.
%The interference nodes of  Node ($p,q$) are
%\begin{itemize}
%	\item In column 1, nodes $(<p>,1),(<p+1>,1),..., (<p+K-L-1>,1)$,
%	\item In column 2, nodes $(<p-1>,2),(<p>,2),..., (<p+K-L-2>,2)$,
%	\item In column 3, nodes $(<p-2>,3),(<p-1>,3),..., (<p+K-L-3>,2)$,
%	
%	\hspace{2in} \vdots
%%	\item In column $q$, nodes $(<p-q+1>,q),..., (<p-1>,q)$, \quad \quad $(<p+1>,q),..., (<p+K-L-q>,q)$,
%	
%%	\hspace{2in} \vdots
%	\item In column $K-L$, nodes $(<p-K-L+1>,K-L), ..., (<p>,K-L)$. 
%\end{itemize}

%
%We divide the interference nodes of any  Node ($p,q$) into two categories:
%%\end{itemize}
%\begin{enumerate}
%	\item intra interference nodes:  nodes $(p,\bar{q})$ (all the other nodes in Row $p$),
%	\item inter interference nodes: interference nodes present in other rows.
%\end{enumerate}
%The interference nodes of  Node ($p,q$) are nodes $(p,\bar{q})$ and in column $t$, nodes $$(<p-t+1>,t)\text{ to } (<p-1>,t),  (<p+1>,t)\text{ to } (<p+K-L-t>,t)  \forall t\in[K-L]$$.
The interference nodes of  Node ($p,q$) are nodes $(p,\bar{q})$, in column $t$ nodes $(<p-t+1>,$ $t)\text{ to } (<p-1>,t)$ and  $(<p+1>,t)\text{ to } (<p+K-L-t>,t) $ for all $ t\in[K-L]$ .\\
%\begin{itemize}
%	\item in column 1, nodes $(<p+1>,1)\text{ to } (<p+K-L-1>,1)$,
%	\item in column 2, nodes $(<p-1>,2)$,  $(<p+1>,2)\text{ to } (<p+K-L-2>,2)$,
%	%\item In column 3, nodes $(r-2,3), (r-1,3)$, \quad \quad $(r+1,3), (r+2,3),..., (r+K-L-3,3)$,
%	
%	\hspace{2in} \vdots
%	\item in column $t$, nodes $(<p-t+1>,t)\text{ to } (<p-1>,t)$,  $(<p+1>,t)\text{ to } (<p+K-L-t>,t)$,
%	
%	\hspace{2in} \vdots
%	\item in column $K-L$, nodes $(<p-K-L+1>,K-L)\text{ to } (<p-1>,K-L)$. 
%\end{itemize}
%In each column exactly $K-L-1$ nodes are present.	
%Our coloring scheme is different for $K-L$ odd and even cases. So, we will discuss them separately.
%\subsubsection{$K-L$ is even}
%\begin{enumerate}
%	\item[i)] 
\vspace{-0.1in}

	For $K-L$ even, we take ${(K-L)(K-L+1)}/{2}$ colors. Let the colors be 
	$$\{(q',p'): q'\in [(K-L)/{2}],p' \in [K-L+1]\}. $$ 
	%\begin{itemize}
	%\item $\forall j\in[\frac{K-L+1}{2}]$, assign Color  $(\frac{K-L+1}{2},j)$ to nodes $(<j+r\frac{K-L+1}{2}>,\frac{K-L+1}{2})$ for all $r\in \mathbb{N}\cup \{0\}$.
	For all	$ q' \in [(K-L)/{2}],p'\in [K-L+1]$, assign Color $(q',p')$ to nodes 
	$(<p'+d(K-L+1)>,$ $q')$ and $(<p'-d(K-L+1)+q'>, K-L+1-q')$   for all $d\in [s]$. 
	
		For example consider the $(N,K= 14, L=8)-$CCDN. The coloring scheme is shown  in Table \ref{Tab:coloraKeven}. We also highlight Node (6,1) interference and non-interference nodes. Note that  the Node (6,1)'s color  is repeated at its non-interfering nodes. 
	%{\color{green}This coloring scheme ensures that the a (or every) color of any nodes interfering nodes in other rows repeated twice among them so that the local chromatic number value is $(K-L)(K-L+1)/2$, which is a lower bound for the ICP mentioned in Table \ref{Tab:coloraKeven}.}
	
	\begin{table}[h]
		\centering
		\begin{tabular}{| c | c | c | c | c | c |}
			\hline 
			(1,1) NI & (2,1) NI & (3,1) NI & (3,5) NI & (2,6) NI & (1,7) I \\
			\hline 
			(1,2) NI & (2,2) NI & (3,2) NI & (3,6) NI & (2,7) I & (1,1) I\\
			\hline
			(1,3) NI & (2,3) NI & (3,3) NI & (3,7) I & (2,1) I & (1,2) I\\
			\hline 
			(1,4) NI & (2,4) NI & (3,4) I & (3,1) I & (2,2) I & (1,3) I\\
			\hline
			(1,5) NI & (2,5) I & (3,5) I & (3,2) I & (2,3) I & (1,4) I\\
			\hline 
			{\cellcolor{red}(1,6)} & (2,6) 	I & (3,6) I & (3,3) I & (2,4) I & (1,5) I\\
			\hline 
			(1,7) I & (2,7) I & (3,7) I & (3,4) I & (2,5) I & (1,6) NI\\
			\hline 
			(1,1) I & (2,1) I & (3,1) I & (3,5) I & (2,6) NI & (1,7) NI\\
			\hline 
			(1,2) I & (2,2) I & (3,2) I & (3,6) NI & (2,7) NI & (1,1) NI\\
			\hline
			(1,3) I & (2,3) I & (3,3) NI & (3,7) NI & (2,1) NI & (1,2) NI\\
			\hline 
			(1,4) I & (2,4) NI & (3,4) NI & (3,1) NI & (2,2) NI & (1,3) NI\\
			\hline
			(1,5) NI & (2,5) NI & (3,5) NI & (3,2) NI & (2,3) NI & (1,4) NI\\
			\hline 
			(1,6) NI & (2,6) NI & (3,3) NI & (3,6) NI & (2,4) NI & (1,5) NI\\
			\hline 
			(1,7) NI & (2,7) NI & (3,7) NI & (3,4) NI & (2,5) NI & (1,6) NI\\
			\hline
		\end{tabular}
		\vspace{0.1cm}
		\caption{Coloring scheme for $(N,K=14,L=8)-$CCDN. Here NI represents the node is a Non-Interfering node for  Node (6,1) and  I represents the node is an Interfering node for  Node (6,1). } \label{Tab:coloraKeven}
		\vspace{-0.5in}
	\end{table}
	%$\forall q' \in [(K-L)/{2}],p'\in [K-L+1]$, assign Color $(q',p')$ to nodes 
	%$(<p'+d(K-L+1)>, K-L+1-q')$ and $(<p'+d(K-L+1)+{K-L+1}-q'>,q')$   for all $d\in [s]$.
	%\vspace{-0.1in}
	
	For any Node $(p,q)$,
	\begin{itemize}
		%\item [--] If $q=(K-L+1)/2$, color assigned to Node $(p,q)$ is repeated only in Column $q$ with a circular gap of $(K-L+1)/2$. Its interference nodes in $q^{th}$ column are nodes $(<p+1>,q),...,<(p+(K-L-1)/2>,q)$ and $(<p-1>,q),...,(<p-(K-L-1)/2>,q)$ but we are assigning same color to (non interfering) nodes $(<p-d(K-L+1)/2>,q)$ and $(<p+d(K-L+1)/2>,q)$ for all $d\in \mathbb{N}\cup \{0\}$.
		\item  if $ q\leq(K-L+1)/2$, color assigned to Node $(p,q)$ is repeated only in columns $q$ and $K-L+1-q$. 
		Node $(p,q)$'s interference nodes are
		\begin{itemize}
			\item in the $q^{th}$ column, nodes $$(<p-q+1>,q)\text{ to } (<p-1>,q), (<p+1>,q)\text{ to } (<p+K-L-q>,q).$$ In our coloring scheme, we are assigning Node $(p,q)$'s color in the $q^{th}$ column, for the non-interfering nodes $ (<p+d(K-L+1)>,q) \text{ } \forall d\in [s],$			
			\item in the $(K-L+1-q)^{th}$ column,  nodes $$(<p-K+L+q>,K-L+1-q)\text{ to } (<p-1>,K-L+1-q),$$ $$(<p+1>,K-L+1-q)\text{ to } (<p+q-1>,K-L+1-q).$$ In our coloring scheme, we are assigning Node $(p,q)$'s color in the $(K-L+1-q)^{th}$ column, for the non-interfering nodes $(<p+q-d(K-L+1)>,K-L+1-q)\text{ } \forall d\in [s],$
		\end{itemize}
%		  %{\color{blue}and \color{green}its interference nodes }
		\item if $q>(K-L+1)/2$, by using similar arguments, we can prove Node $(p,q)$'s  color is repeated only for non-interference nodes.
	\end{itemize}
	%\end{itemize}

%	\begin{table}[h]
%		\centering
%		\begin{tabular}{| c | c | c | c | c | c |}
%			\hline 
%			(1,2) NI & (2,3) NI & (3,4) NI & (3,1) NI & (2,1) NI & (1,1) I \\
%			\hline 
%			(1,3) NI & (2,4) NI & (3,5) NI & (3,2) NI & (2,2) I & (1,2) I\\
%			\hline
%			(1,4) NI & (2,5) NI & (3,6) NI & (3,3) I & (2,3) I & (1,3) I\\
%			\hline 
%			(1,5) NI & (2,6) NI & (3,7) I & (3,4) I & (2,4) I & (1,4) I\\
%			\hline
%			(1,6) NI & (2,7) I & (3,1) I & (3,5) I & (2,5) I & (1,5) I\\
%			\hline 
%			{\color{red}(1,7)} & (2,1) 	I & (3,2) I & (3,6) I & (2,6) I & (1,6) I\\
%			\hline 
%			(1,1) I & (2,2) I & (3,3) I & (3,7) I & (2,7) I & {\color{red}(1,7)} NI\\
%			\hline 
%			(1,2) I & (2,3) I & (3,4) I & (3,1) I & (2,1) NI & (1,1) NI\\
%			\hline 
%			(1,3) I & (2,4) I & (3,5) I & (3,2) NI & (2,2) NI & (1,2) NI\\
%			\hline
%			(1,4) I & (2,5) I & (3,6) NI & (3,3) NI & (2,3) NI & (1,3) NI\\
%			\hline 
%			(1,5) I & (2,6) NI & (3,7) NI & (3,4) NI & (2,4) NI & (1,4) NI\\
%			\hline
%			(1,6) NI & (2,7) NI & (3,1) NI & (3,5) NI & (2,5) NI & (1,5) NI\\
%			\hline 
%			{\color{red}(1,7)} NI & (2,1) NI & (3,2) NI & (3,6) NI & (2,6) NI & (1,6) NI\\
%			\hline 
%			(1,1) NI & (2,2) NI & (3,3) NI & (3,7) NI & (2,7) NI & {\color{red}(1,7)} NI\\
%			\hline
%		\end{tabular}
%		\vspace{0.1cm}
%		\caption{Coloring scheme for Table $(N,K=14,L=8)-$CCDN. Here NI represents the node is a Non-Interfering node for  Node (6,1) and  I represnts the node is an Interfering node for  Node (6,1). } \label{Tab:coloraKeven}
%			\vspace{-0.5in}
%	\end{table}
	
This coloring scheme ensures that none of the nodes share its color with any of its interfering nodes.   Hence, this is a proper coloring scheme. Note that $(K-L)(K-L-1)$ inter-interference nodes of any node share $(K-L)(K-L-1)/2$ colors. 
	
	%{\color{green}Observe that {for Node ($i,j$), $\frac{(K-L)(K-L+1)}{2}$ colors\footnote{one color for Node ($p,q$), $K-L-1$ colors for $K-L-1$ interference nodes  ($p,\bar{q}$) in the $p^{th}$ row, and $\frac{(K-L)(K-L-1)}{2}$ colors for $(K-L)(K-L-1)$  interference nodes in the other nodes} appear in the closed anti-outneighbor\\hood\remove{, one color for Node ($i,j$), one color each  for nodes ($i,\bar{j}$), one common color for nodes ($<i+1>,1$) and ($<i-2>,3$), one common color for nodes ($<i+2>,1$) and ($<i-1>,3$), one common color for nodes ($<i+1>,2$) and ($<i-1>,2$)}}. Hence, the local chromatic number of  the ICP is $\frac{(K-L)(K-L+1)}{2}$.}
	
	Note that for Node ($p,q$), ${(K-L)(K-L+1)}/{2}$ colors appear in the closed anti-outneigh\\borhood, one color for Node ($p,q$), $K-L-1$ colors for $K-L-1$ intra-interference nodes, and ${(K-L)(K-L-1)}/{2}$ colors for $(K-L)(K-L-1)$  inter-interference nodes in the other rows.  Hence, an upper bound on  the local chromatic number of  the ICP is 
	${(K-L)(K-L+1)}/{2}$.
	%\subsubsection{$K-L$ is odd}
	%\item[ii)] 
	
	For $K-L$ odd, we take $\frac{(K-L)(K-L+1)}{2}$ colors. Let the colors be 
	$$\bigg\{(q',p'): q'\in \Big[\frac{K-L-1}{2}\Big], p' \in [K-L+1]\bigg\} \cup  \bigg\{\Big(\frac{K-L+1}{2},p''\Big):  p''\in\Big[\frac{K-L+1}{2}\Big]\bigg\}.$$
	
	\begin{itemize}
		\item $\forall p''\in[\frac{K-L+1}{2}]$, assign Color  $(\frac{K-L+1}{2},p'')$ to nodes $(<p''+d\frac{K-L+1}{2}>,\frac{K-L+1}{2})$ for all $d\in [2s]$. 
		\item $\forall q' \in [\frac{K-L-1}{2}],p'\in [K-L+1]$, assign Color $(q',p')$ to nodes 
		$(<p'+d(K-L+1)>, q')$ and $(<p'-d(K-L+1)+q'>,K-L+1-q')$   for all $d\in [s]$.
	\end{itemize}

It can be verified that the  coloring scheme is a proper coloring scheme and an upper bound on  the local chromatic number of  the ICP using this scheme  is 
${(K-L)(K-L+1)}/{2}$. We skip the details here for brevity.

	\remove{%{\color{red}	
	%If $L=K(s-1)/s + 1\implies K=s(K-L+1)$. We are repeating colors in a column with a circular gap of either  $(K-L+1)/2$ (for $\big(\frac{K-L+1}{2}\big)^{th}$ column) or $K-L+1$ (for remaining columns). Since, the number of nodes in a column are $K=i(K-L+1)$ and color repetition gap is $(K-L+1)/2$ (or $K-L+1$), the coloring scheme assigns only one color to a node (i.e., a color does not spread to all the nodes in a column, it restricts to $2i$ (or $i$) number of nodes in a column\remove{For example: Color (1,1) is assigned to columns 1 and $K-L$. In Column 1, Color (1,1) is assigned to nodes ($K-L+1$,1),...,($i(K-L+1),1$) but not to any other nodes}).}
	For Node $(p,q)$,
	\begin{itemize}
		\item  if $q=(K-L+1)/2$, color assigned to Node $(p,q)$ is repeated only in Column $q$.  Its interference nodes in $q^{th}$ column are nodes $$(<p-(K-L-1)/2>,q)\text{ to }(<p-1>,q),$$ $$(<p+1>,q)\text{ to }<(p+(K-L-1)/2>,q)$$ but we are assigning its color to (non-interfering) nodes  $$(<p+d(K-L+1)/2>,q) \text{ } \forall d\in [2s],$$
		\item  if $q<(K-L+1)/2$, color assigned to Node $(p,q)$ is repeated only in columns $q$ and $K-L+1-q$.
			{\color{blue}Node $(p,q)$'s interference nodes}
		\begin{itemize}
			\item in $q^{th}$ column are nodes $$(<p-q+1>,q)\text{ to } (<p-1>,q), (<p+1>,q)\text{ to } (<p+K-L-q>,q)$$ but we are assigning {\color{blue}Node $(p,q)$'s} color to it's non-interfering nodes $$ (<p+d(K-L+1)>,q) \text{ } \forall d\in [s],$$			
			\item in $(K-L+1-q)^{th}$ column are nodes $$(<p-K+L+q>,K-L+1-q)\text{ to } (<p-1>,K-L+1-q),$$ $$(<p+1>,K-L+1-q)\text{ to } (<p+q-1>,K-L+1-q)$$ but we are assigning {\color{blue}Node $(p,q)$'s} color to it's non-interfering nodes $$(<p+q+d(K-L+1)>,K-L+1-q)\text{ } \forall d\in [s],$$
		\end{itemize} %Its interference nodes in $q^{th}$ column are nodes $(<p-q+1>,q),..., (<p-1>,q)$, $(<p+1>,q),..., (<p+K-L-q>,q)$, but we are assigning same color to (non interfering) nodes $(<p-(K-L+1)>,q)$ and $(<p+(K-L+1)>,q)$, in $(K-L+1-q)^{th}$ column are nodes $(<p-K+L+q>,K-L+1-q),..., (<p-1>,K-L+1-q)$, $(<p+1>,K-L+1-q),..., (<p+q-1>,K-L+1-q)$, but we are assigning same color to (non interfering) nodes $(<p-(K-L-q+1)>,K-L+1-q)$ and $(<p+q>,K-L+1-q)${\color{blue}, and}
		\item  if $q>(K-L+1)/2$, by using similar arguments, we can prove Node $(p,q)$'s color is repeated only for non-interference nodes.
	\end{itemize}
	So, this coloring scheme ensures that none of the nodes share its color with any of its interfering nodes.   Hence, this is a proper coloring scheme. 
	%\begin{table}[h]
	%	\centering
	%	\begin{tabular}{| c | c | c | c | c |}
	%		\hline 
	%		(1,2) NI & (2,3) NI & (3,1) NI & (2,1) NI & {\color{red}(1,1)} NI\\
	%		\hline 
	%		(1,3) NI & (2,4) NI & (3,2) NI & (2,2) NI & (1,2) I\\
	%		\hline
	%		(1,4) NI & (2,5) NI & (3,3) NI & (2,3) I & (1,3) I\\
	%		\hline 
	%		(1,5) NI & (2,6) NI & (3,1) I & (2,4) I & (1,4) I\\
	%		\hline
	%		(1,6) NI & (2,1) I& (3,2) I& (2,5) I& (1,5) I\\
	%		\hline 
	%		{\color{red}(1,1)} & (2,2) I & (3,3) I & (2,6) I & (1,6) I\\
	%		\hline 
	%		(1,2) I & (2,3) I & (3,1) I & (2,1) I & {\color{red}(1,1)} NI\\
	%		\hline 
	%		(1,3) I & (2,4) I & (3,2) I & (2,2) NI & (1,2) NI\\
	%		\hline
	%		(1,4) I & (2,5) I & (3,3) NI & (2,3) NI & (1,3) NI\\
	%		\hline 
	%		(1,5) I & (2,6) NI & (3,1) NI & (2,4) NI & (1,4) NI\\
	%		\hline
	%		(1,6) NI & (2,1) NI & (3,2) NI & (2,5) NI & (1,5) NI\\
	%		\hline 
	%		{\color{red}(1,1)} NI & (2,2) NI & (3,3) NI & (2,6) NI & (1,6) NI\\
	%		\hline 
	%	\end{tabular}
	%	\vspace{0.2cm}
	%	\caption{Coloring scheme for  $(N,K=12,L=7)-$CCDN.  Here, NI means Non-Interfering node and  I means Interfering node for Node (6,1). } \label{Tab:colorK}
	%\end{table}
	%
	%For example consider $\{N, 12, 7\}-$CCDN. The coloring scheme is shown  in Table \ref{Tab:colorK}. We also highlight Node (6,1) interference and non-interference nodes.

	Observe that for Node ($p,q$), ${(K-L)(K-L+1)}/{2}$ colors appear in the closed anti-outneighborhood. Therefore, the local chromatic number of  the ICP is $(K-L)(K-L+1)/{2}$.
	
	{\color{green}If space is not sufficient, then we will give only coloring scheme and mention that this is a proper color scheme without explanation and mention it gives the local chromatic number value as $(K-L)(K-L+1)/2$.}
}

\end{document}